\documentclass[aps,prb,twocolumn, floatfix,groupedaddress,11.9pt]{revtex4-1}
\pdfoutput=1
\usepackage{amssymb}
\usepackage{amsmath}
\usepackage{natbib}
\usepackage{braket}
\usepackage{float}
\usepackage[]{hyperref}
\usepackage{graphicx}
\usepackage{mathrsfs}
\usepackage{float}
\usepackage{amsbsy}
\usepackage{dcolumn}
\usepackage{bm}
\usepackage{dsfont}
\usepackage{color}

 \usepackage{amsthm}
\newtheorem{mydef}{Definition}
\newtheorem{mythm}{Theorem}
\newtheorem{mylemm}{Lemma}
\newtheorem{mycorol}{Corollary}
\newtheorem{conjecture}{Conjecture}

\begin{document}


\title{Boson condensation and instability in the tensor network representation of string-net states}



\author{Sujeet K. \surname{Shukla}}
\affiliation{Institute of Quantum Information and Matter, California Institute of Technology, California, USA}
\author{M.~Burak \surname{\c{S}ahino\u{g}lu}}
\affiliation{Vienna Center for Quantum Technology, University of Vienna, Boltzmanngasse
5, 1090 Vienna, Austria}
\author{Frank \surname{Pollmann }}
\affiliation{Max-Planck-Institut f$\ddot{u}$r Physik komplexer Systeme, D-01187 Dresden, Germany}

\author{Xie \surname{Chen}}
\affiliation{Department of Physics and Institute for Quantum Information and Matter, California Institute of Technology, Pasadena, CA 91125, USA}


\date{\today}

\begin{abstract}
The tensor network representation of many-body quantum states, given by local tensors, provides a promising numerical tool for the study of strongly correlated topological phases in two dimension. However, the representation may be vulnerable to instabilities caused by small variations in the local tensors. For example, the topological order in the tensor network representations of the toric code ground state has been shown in Ref.\onlinecite{Chen10} to be unstable if the variations break certain $Z_2$ symmetry of the tensor. In this work, we ask whether other types of topological orders suffer from similar kinds of instability and if so, what is the underlying physical mechanism and whether we can protect the order by enforcing certain symmetries on the tensor. We answer these questions by showing that the tensor network representation of all string-net models are indeed unstable, but the matrix product operator (MPO) symmetries of the tensors identified in Ref.\onlinecite{csahinouglu2014characterizing} can help to protect the order. In particular we show that a subset of variations that break the MPO symmetries lead to instability by inducing the condensation of bosonic quasi-particles which destroys the topological order in the system. Therefore, such variations must be forbidden for the encoded topological order to be reliably extracted from the local tensors. On the other hand, if a tensor network based variational algorithm is used to simulate the phase transition due to boson condensation, such variation directions may prove important to access the continuous transition correctly.

\end{abstract}


\maketitle

\tableofcontents

\section{Introduction}
The tensor network representation of quantum states (including the matrix product states in $1$D)\cite{Fannes92,White93,Verstraete08,Vidal09} provides a generic tool for the numerical study of strongly interacting systems. As variational wave functions, the tensor network states can be used to find the ground state wave function of local Hamiltonians and identify the phase at zero temperature. In particular, it has become a powerful approach in the study of topological phases, whose long range entanglement is hard to capture with conventional methods. It has been shown that a large class of topological states, the string-net condensed states~\cite{Levin05}, can be represented exactly with simple tensors~\cite{Gu09,Buerschaper09}. Moreover, numerical studies applied to realistic models have identified nontrivial topological features in the ground state wave function (see e.g. Ref.\onlinecite{Yan11,Jiang12,Depenbrock12}).

In the numerical program, the parameters in the tensors are varied so as to find the representation of the lowest energy state. After that, topological properties are extracted from these tensors in order to determine the topological phase diagram at zero temperature. However, this problem might not be numerically `well-posed'. That is, arbitrarily small variations in the local tensor may lead to completely different result as to what topological order it represents. In particular, Ref.~\onlinecite{Chen10} demonstrates that this happens in the case of $Z_2$ toric code topological order. While this presents a serious problem for the tensor network approach to study topological phases, Ref.~\onlinecite{Chen10} also showed that such instabilities can be avoided if certain $Z_2$ symmetry is preserved in the local tensor. It has been shown that the topological order in the toric code model is stable against arbitrary local perturbation to the Hamiltonian of the system~\cite{Bravyi10}. The fact that a certain variation direction of the tensor network representation may induce an immediate change in the topological order indicates that such a variation corresponds to highly nonlocal changes in the ground state wave function.

Does similar problem occur for general string-net states as well? This is the question we address in this paper. In particular, we ask:
\begin{enumerate}
\item{Does the tensor network representation of other string-net states also have such unstable directions of variation?}
\item{If so, can they be avoided by preserving certain symmetries in the local tensor?}
\item{What is the physical reason behind such instabilities and their prevention?}
\end{enumerate}

While the $Z_2$ symmetry requirement for toric code is naturally related to the $Z_2$ gauge symmetry of the theory, for more general string-nets which are not related to gauge theory, it is not clear whether similar symmetry requirement is necessary and if so what they are.

In this paper, we answer the above questions as follows:

\begin{enumerate}

\item{All string-net tensors have unstable directions of variation.}

\item{To avoid such instabilities, we need to avoid `stand-alone' variations that break the Matrix-Product-Operator(MPO) symmetry introduced in Ref.\onlinecite{csahinouglu2014characterizing, Buerschaper2014}. (We are going to explain in detail what `stand-alone' and MPO symmetry means in the following sections). }

\item{The physical reason for the instability is that `stand-alone' variations which violate these symmetries induce condensation of bosonic quasi-particles and hence destroys (totally or partially) the topological order.}

\end{enumerate}

 \par 

To support the above claims, we calculate the topological entanglement entropy $S_{\text{topo}}$\cite{Kitaev06,Levin06} from the representing tensor and (partially) characterize the encoded topological order.
In particular, consider a tensor network state represented by a local tensor $T$. We are interested in varying the local tensor $T$ everywhere on the lattice, in such a way that, $T \rightarrow T+\epsilon T'$, where $\epsilon \ll 1$. In order to study whether topological order is lost or still present after a variation in the direction $T'$, we calculate topological entanglement entropy of the original and the modified state as a function of $\epsilon$, $S_{\text{topo}}(\epsilon)$. We say the variation is unstable in $T'$ direction if
\begin{align}
\lim_{\epsilon \rightarrow 0} S_{\text{topo}}(\epsilon) \neq  S_{\text{topo}}(0). \label{instabilityequation}
\end{align}
If $\lim_{\epsilon \rightarrow 0} S_{\text{topo}}(\epsilon)$ is smaller than $S_{\text{topo}}(0)$, we say that topological order is (partially) lost. If $\lim_{\epsilon \rightarrow 0} S_{\text{topo}}(\epsilon) = S_{\text{topo}}(0)$ we call that direction stable meaning that topological order is still present and remains the same. 
This understanding of tensor instability is important not only for the identification of topological order for a particular model, but also for the numerical study of phase transitions between topological phases. In particular, if one is to use the tensor network approach to study phase transition due to boson condensation, then the corresponding variation direction must be \textit{allowed} in order for the simulation to give correct results. For example, in Ref.~\onlinecite{Gu08}, it was shown that if such variation directions are not included as variational parameters, then we see a first order transition even though in fact it is second order. We are going to elaborate more on this point later in the paper.

The paper is organized as follows.  
In section \ref{tc}, we start from the simplest string-net model -- the toric code model\cite{Kitaev03}, and study two types of tensor network representation of its ground state. The single line representation was studied in Ref.~\onlinecite{Chen10} and here we recover the result on the instability of the tensor with respect to certain $Z_2$ symmetry breaking variations. While reproducing the result, we introduce a new algorithm which allows us to investigate more complicated string-net models in later parts of this paper. The second representation we study for the toric code is the double line representation, as discussed in Ref.\onlinecite{Gu08}. While the single line representation has only one virtual $Z_2$ symmetry, the double line representation has multiple of them. Do they all protect the encoded topological order in the same way? To find out, we calculate, with the new algorithm, the topological entanglement entropy for the double line representation with different variations. It reveals that there are actually two kinds of symmetries, and their relation to the topological order is actually \textit{opposite} to one another. That is, the only variations that change the topological order are the ones that \textit{respect} the first kind of symmetry (which we call `stand-alone') but actually \textit{break} the second kind (MPO symmetry \cite{csahinouglu2014characterizing}). One of the key results of this paper is contained in the next section, \ref{sec: virtual subspaces}, where we first identify the source of these two symmetries and define them precisely, and make the general conjecture of TNR-instability. This conjecture says that all TNR have these two kinds of symmetries it is always the variations that respect the stand-alone but break the MPO symmetry that are unstable.

Then in section \ref{sec: physical understanding} we put forward the physical understanding of TNR instability by systematically understanding the physical significance of the two symmetries and their interplay. This understanding concludes that the instability occurs due to condensation of topological bosons of the string-net model under consideration. In the next section we note that tensor instability actually has implications for phase transition simulations using tensor network ansatz. And hence, our result should guide the choice of tensor network anstaz in phase transition simulations. 
To generalize our study to generic string-net models, we study next the double semion model in section \ref{ds}. We first predict instabilities using our conjecture, and then find them to be true in our numerical calculation. 

We then directly apply our study to the general string-net model and its triple-line TNR in section \ref{sec: sn}. We calculate the required symmetries and conclude that our conjecture predicts that triple-line TNR of \textit{all} string-net have unstable directions of variations. An analytical proof of this prediction can be found in Appendix \ref{sninstability proof}. Finally in section \ref{sec: fb} we test our understanding of the general string-model on the double-Fibonacci model, which has a non-abelian topological order as opposed to toric code and double-semion. We again find our conjectures to be precisely accurate.  
 Our results also reveal the physical meaning of the virtual tensor network symmetries for topologically ordered ground states that have been found for Kitaev quantum double models (G-injectivity\cite{Schuch2010}) and later generalized to twisted quantum doubles (twisted G-injectivity\cite{Buerschaper2014} and MPO-injectivity\cite{csahinouglu2014characterizing}) and general string-net models (MPO-injectivity\cite{csahinouglu2014characterizing}).
 
Finally, a summary of the results is given in section \ref{conclusion} and open questions are discussed. Some details of our analysis are explained in the appendices, including relations of MPO symmetries to the Wilson-loop operators, a brief review of string-net models, their tensor network representation and their transformation under the application of string-operators, proof of the existence of unstable directions in triple-line representations of general string-net ground states, and finally the dependence of topological entanglement entropy on the choice of boundary condition in our calculation.

\section{Instabilities in TNRs of the toric code}
\label{tc}
\begin{figure}
    \centering
    \includegraphics[width=0.8\columnwidth]{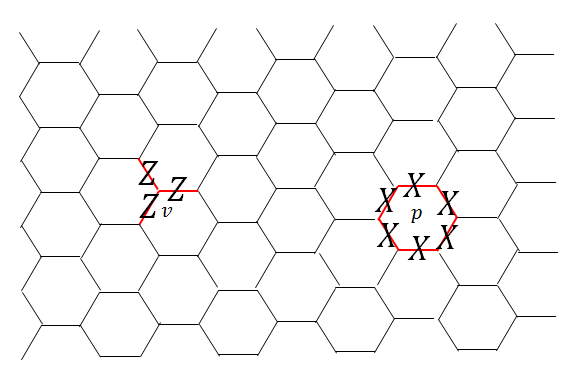}
    \caption{Vertex and plaquate terms of the toric code Hamiltonian}
    \label{fig:TCH}
\end{figure}
We start from the simplest illustrative example of nonchiral intrinsic topological order: the toric code \citep{Kitaev03}. We work on a hexagonal lattice and assign local degrees of freedom, i.e. $0$-spin down- or $1$-spin up, on the edges of the lattice.  It is convenient to consider spin up as a presence of a string and $0$ as the absence of the string. So the total Hilbert space can be thought of as the space of all string configurations on a hexagonal lattice.  The toric code Hamiltonian is a sum of local commuting projectors, given as
\begin{eqnarray}
H &=& -\sum_v A_v - \sum_p B_p \nonumber \\
&=& -\sum_v \prod_{l\in v} Z_l - \sum_p \prod_{l\in p} X_l \label{TCHamiltonian}
\end{eqnarray}
where $v$ denotes the vertices, and $p$ denotes the plaquettes. $l\in v$ denotes the edges attached to $v$ and $l \in p$ denotes the edges on the boundary of plaquette $p$ (see Fig.~\ref{fig:TCH}). Vertex terms restrict the ground states to closed strings of $1$s and plaquette terms make all possible loop configurations of equal weight. Hence, the toric code ground state (up to normalization) can be written as 
\begin{align}
|\Psi_{\textrm{gs}} \rangle = \sum_{X \in \text{closed}} |X\rangle \label{TCgs}
\end{align}
where $X$ denotes the string configurations on the lattice. So, the ground state of toric code Hamiltonian is an equal weight superposition of all closed string configurations. It has topological order and has topological entanglement entropy $S_{\text{topo}}= \log2$. The toric code model has 4 anyons (superselection sectors), $\lbrace \mathbf{1},e,m,em \rbrace $. $\mathbf{1}$ is the vacuum, $e$-particle is the $Z_2$-gauge charge (violates the vertex term) and $m$-particle is the $Z_2$-gauge flux (violates the plaquette term). Both $e$ and $m$ have a trivial topological spin, so we call them \textit{topological bosons}, or simply bosons. Braiding $e$ with $m$ produces a phase factor of $-1$. 

Now we look at tensor network representations (TNR) of the toric code ground state. Specifically, we will first explain the \textit{Single-line} tensor representation, and then the \textit{Double-line} tensor representation. We will see that different TNRs of the same state can have different kinds of instabilities.

\subsection{Single-line TNR of the toric code and its instability}

\begin{figure}
\includegraphics[width=\columnwidth]{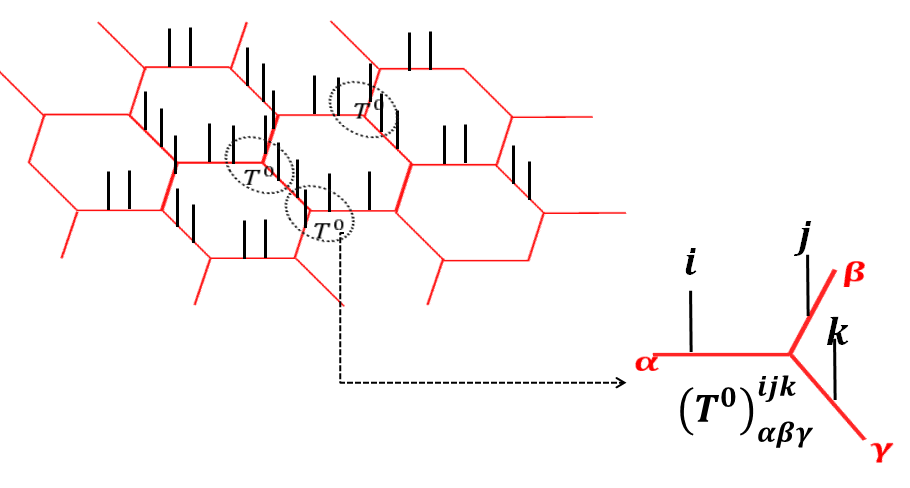}
\caption{ Single-line TNR of the toric code state: we double the local Hilbert space on each edge and take $|i \rangle \rightarrow | ii \rangle$, $i=0,1$. So the state has a $Z_2$ topological order but on a bigger Hilbert space. We associate to each vertex the tensor product of 3 Hilbert spaces closest to it. We can now place a tensor, $T^0$ on each vertex with 3 out of plane physical legs, $ i,j,k$, and 3 in-plane virtual legs, $\alpha,\beta,\gamma$ with values as given in Eq.~\eqref{SLTNReq}. }
\label{fig:SLTNR}
\end{figure}
This is the simplest TNR of the toric code state. We first copy each computational basis into two, as shown in the Fig.~\ref{fig:SLTNR}. That is, the labels $0$ and $1$ on every edge become $00$ and $11$ on the same edge. Now the local Hilbert space neighbouring each vertex is made out of three qubits. We associate a tensor with three physical indices/legs (throughout the paper we will use ``indices" and ``legs" interchangeably), and three virtual indices/legs to each vertex, represented algebraically as $(T^0)^{ijk}_{\alpha\beta\gamma}$ where $i,j,k$ are the three physical indices and $\alpha,\beta,\gamma$ are the three virtual indices, as shown in Fig.~\ref{fig:SLTNR}.  The components of the tensor are 
\begin{eqnarray} \label{SLTNReq}
(T^0)^{ijk}_{\alpha\beta\gamma} =  \begin{cases} \delta_{i\alpha}\delta_{j\beta}\delta_{k\gamma}  & \text{ if } \alpha+\beta+\gamma= \text{even} \\
0 & \text{ otherwise} \end{cases}
\end{eqnarray}
where $\delta$ is the kronecker delta function. So, physical and virtual legs are identified and an even number of indices carry label $1$ out of every three edges neighbouring a vertex, i.e., we satisfy the vertex condition. The plaquette condition is also satisfied since every configuration is of equal weight. 
Therefore, the tensor network state constructed using the above local tensor leads to the toric code ground state given in Eq.~\eqref{TCgs}.

It was shown by \citet{Chen10} that single-line TNR of the toric code state is not stable in certain directions of variation. Before we explain what these unstable directions of variation are, we first note that the single-line TNR explained above has a \textit{virtual symmetry}. If an operation on the virtual indices leaves the tensor invariant, we will call it a virtual symmetry of the tensor. Because the single-line tensor is non-zero only when virtual legs have even number of 1s, it has a natural $Z\otimes Z\otimes Z$ virtual symmetry (see \citet{Schuch2010} for TNR virtual symmetries of the quantum double models). That is, the tensor in \eqref{SLTNReq} satisfies the relation,
\begin{eqnarray}\label{SLTCsym}
\centering
\includegraphics[scale=0.3]{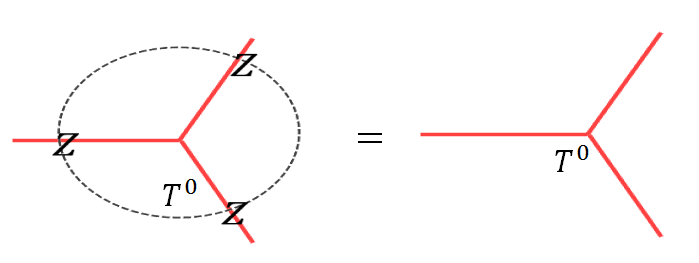}.
\end{eqnarray}
where we have omitted the physical legs for visual clarity. (We would often omit physical legs from tensor diagrams throughout the paper when we are mostly concerned with the virtual space/indices.)  It is a $Z_2$ symmetry with group elements $I \otimes I \otimes I$ and $Z\otimes Z\otimes Z$ acting on the virtual legs of the local tensor.  \citet{Chen10} showed that topological order is stable with \textit{any }$Z_2$ respecting variations and unstable with any $Z_2$ violating variation. To illustrate this, we can consider two different directions of variation in single-line TNR. We can add an $X$ or $Z$ variation on one of the virtual indices of the tensor,
\begin{eqnarray}\label{SLpert}
\includegraphics[scale=0.3]{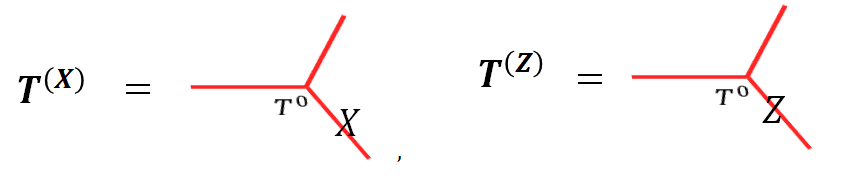}.
\end{eqnarray}
More explicitly, these tensor components are given by,
\begin{eqnarray}
T^{(X)}_{ \alpha\beta \gamma} =  \sum_{\gamma'} X_{\gamma,\gamma'} T^0_{\alpha \beta \gamma}, \\
T^{(Z)}_{ \alpha\beta \gamma} =  \sum_{\gamma'} Z_{\gamma,\gamma'} T^0_{\alpha \beta \gamma} .
\end{eqnarray}
$T^{(X)}$ variation violates the $Z_2$ symmetry while $T^{(Z)}$ does not. That is,
\begin{eqnarray}
\includegraphics[scale=0.3]{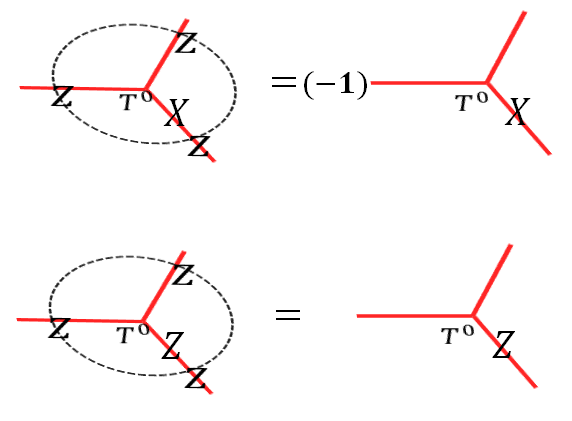},
\end{eqnarray}
and it was shown that $T^{(X)}$ type variations cause an instability and while $T^{(Z)}$ type variations do not. Note that, though we chose variations only on the virtual indices for simple illustration, the same conclusion applies for any random variation including those on the physical indices. However, if a variation  acts \textit{only} on the physical indices, it cannot break the $Z_2$ virtual symmetry, and hence would \textit{always} be stable.  \par 
We reproduce this known result with a new algorithm for calculating $S_{\text{topo}}$. This algorithm allows us to calculate $S_{\text{topo}}$ in more complicated examples to be dealt with later. Before we move on to other TNRs, we would like to explain the algorithm used here. Readers can skip this section if they are not interested in the details of the algorithm.  

\subsection{Algorithm for calculating topological entanglement entropy}
\label{alg}
Here we explain the algorithm we use to calculate the topological entanglement entropy of any translation invariant tensor network state. We use the idea presented by \citet{Cirac11} to calculate reduced density matrix on a region and hence its entanglement entropy. We consider honeycomb lattice, though it can easily be extended to other lattices. By translation invariant we mean that all vertices on the sublattice A and sublattice B are attached with the same tensors, $T_A$ and $T_B$, respectively. First we define certain notations for convenience of later discussion. The starting objects are given tensors $T^{I}_{\alpha}$, where $I$ and $\alpha$ denote the set of physical and virtual indices, respectively: $I= (i_1,i_2,..), \alpha = (\alpha_1, \alpha_2, .. )$. The state represented  by these tensors can be written as 
\begin{eqnarray}
|\Psi \rangle= \sum_{I_1,I_2,..} \text{Tr}(T^{I_1} T^{I_2}\ldots ) |I_1,I_2,...\rangle.
\end{eqnarray}
We denote the tensor resulting from contracting the virtual indices of tensors $T$  on a region $R$ as $T(R)$.
$\mathbb{T}$ denotes the `double tensor' resulting from contracting the physical indices of $T$ with those of $T^{\dagger}$, that is, $\mathbb{T}=TT^{\dagger}= \sum_I T^I_{\alpha } \left(T^I_{\alpha'}\right)^*$. Similar to $T(R)$, we denote the double tensor contracted on a region $R$ as $\mathbb{T}(R)$.
\begin{figure}
\includegraphics[scale=0.4]{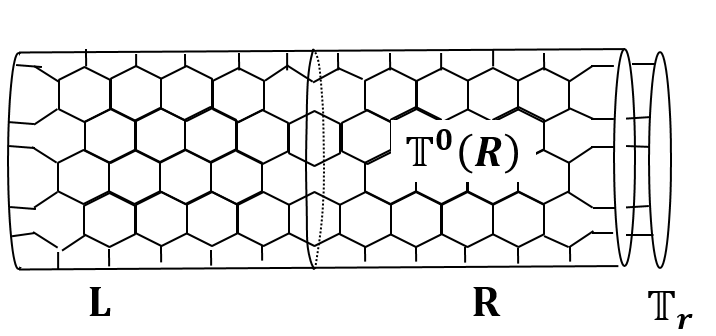}
\caption{The honeycomb lattice is put on a cylinder with some boundary tensors, $T_r$. We calculate the topological entanglement entropy by calculating the entanglement entropy of the right half of the cylinder. }
\label{algo0}
\end{figure}
Now let us consider putting this tensor network state on a cylinder. We denote the left half of the cylinder as $L$ and the right half as $R$. The honeycomb lattice is placed in a way so that $L$ and $R$ divide it into exact halves. So the line between the two halves goes through the middle of the plaqeuttes as shown in the Fig.~\ref{algo0}. We denote the tensors on the left and right boundaries  as $T_l$ and $T_r$. \par 
When we contract bulk double tensors with the boundary double tensors, we get a density matrix operator on the virtual indices,
\begin{eqnarray}
\sigma_L =\mathbb{T}_l(\partial L) \mathbb{T}(L), \quad \sigma_R = \mathbb{T}(R)\mathbb{T}_r(\partial R).
\end{eqnarray}
\citet{Cirac11} showed that the physical reduced density matrix on one of these halves, let's say the left one, is related to the density operator on the virtual indices as,
\begin{eqnarray}
\rho_L = U \sqrt[]{\sigma_L^T} \sigma_R \sqrt{\sigma_L^T} U^{\dagger }
\end{eqnarray}
where $U$ is an isometry. Hence $\rho_L$ and $ \sqrt[]{\sigma_L^T} \sigma_R \sqrt{\sigma_L^T} $ have the same spectrum. In addition, under right symmetry conditions, $\sigma_L^T=\sigma_R=\sigma_b$. When this is true, up to change of basis, we find that $\rho_L \propto \sigma_b^2$. The normalized reduced density matrix is 
\begin{eqnarray}
\rho_L &=& \frac{\sigma_b^2}{\textrm{Tr}(\sigma_b^2)}.
\end{eqnarray}
It is known that the R\'enyi entropy with any R\'enyi index gives the same topological entanglement entropy\cite{Flammia09}. So we calculate R\'enyi entropy with R\'enyi index $\frac{1}{2}$,
\begin{eqnarray}
S_{1/2}(\rho_L) &=& \frac{1}{1-1/2} \log \textrm{Tr}(\rho_L^{1/2}) \nonumber \\
 &=& 2\log \textrm{Tr}(\sigma_b ) -\log\textrm{Tr}(\sigma_b^2).
\end{eqnarray}
In the limit of large cylinder, it should behave like 
\begin{eqnarray}
S_{1/2}(\rho_L) &=& \alpha_0|C| - S_{\text{topo}} \label{Stopomethod}
\end{eqnarray}
where $|C|$ is the circumference of the cylinder. This is how we calculate $S_{\text{topo}}$ starting with a tensor network state. \par 
Before we move on to the next step, we would like to mention an important subtlety regarding computation of $S_{\text{topo}}$ on a cylinder. In Ref.\onlinecite{DongFradkinLeighNowling,Zhang12} it has been shown that $S_{\text{topo}}$ calculated this way on a cylinder, in general, might depend on the boundary conditions. We choose a particular boundary condition for all our calculations and examine the dependence of $S_{\text{topo}}$ on boundary condition in the appendix \ref{boundary issue}. Our findings are consistent with the conclusion in Ref.\onlinecite{Zhang12}. \par  
We first have to calculate $\mathbb{T}(R)\mathbb{T}_r(\partial R)$ for the above setup. The problem is, the computational complexity of exact tensor contraction grows exponentially with the size of $R$, so we need to use some approximate renormalization algorithm. We use an algorithm which is a slight modification of known tensor renormalization algorithms \citep{Gu08,Vidal2003,Vidal09}. 
Consider double tensors contracted along a thin strip on the cylinder giving us a \textit{transfer matrix operator}, $ \mathbb{S}$. If $R$ includes $n$ of such strips, we have $\mathbb{T}(R)=\mathbb{S}^n$. Since the tensor network state under consideration are short range correlated along the cylinder, the spectrum of $\mathbb{S}$ is gapped. Consequently, for large $n$, only the highest eigenvalue and the corresponding eigenvector of $\mathbb{S}$ dominates.  That is, in thermodynamic limit, $\mathbb{T}(R)$ only depends on the highest eigenvalue/eigenvector of the transfer matrix operator, $\mathbb{S}$. Moreover, we expect to approximate the eigenvector of highest eigenvalue with a \textit{Matrix Product State (MPS)} with finite bond dimensions, since the tensor network state is short range correlated along the circumference of the cylinder. So we can start with a boundary MPS, apply the transfer matrix operator, and approximate the resulting state as an MPS with a fixed, finite bond dimensions. With each step, approximation to the eigenvector with highest eigenvalue improves and we do this recursively until we reach the fixed point giving us the desired eigenvector. Note that we require transfer matrix operators to be reflection symmetric for the condition $\sigma_L^T=\sigma_R =\sigma_b \Rightarrow \rho_L \propto \sigma_b^2 $ to hold true. \par 
 The recursive algorithm is as following: 
\begin{enumerate}
\item Initiate the boundary double tensor $\mathbb{T}_{A'}= \mathbb{T}_{r,A'}$ and $\mathbb{T}_{B'}= \mathbb{T}_{r,B'}$. The tensor network to be contracted looks as  
\begin{equation}
\includegraphics[scale=0.3]{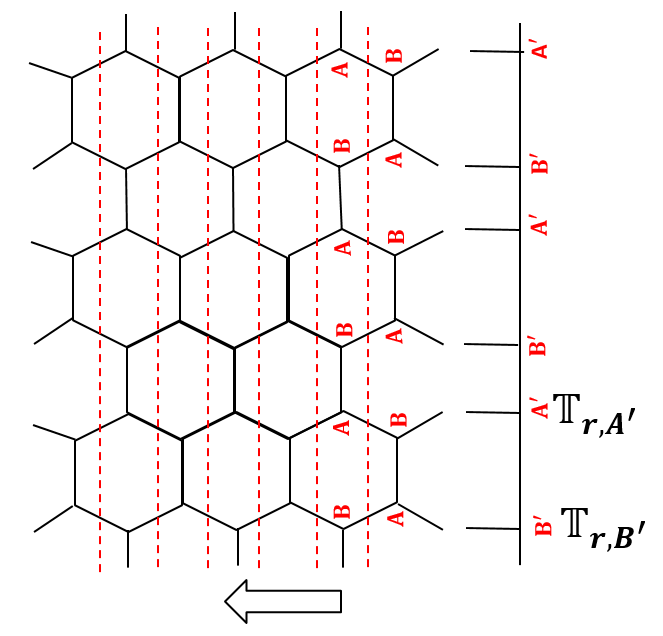}
\end{equation}
\item Contract the bulk double tensors, $\mathbb{T}_{A}$ and $ \mathbb{T}_B$ with the boundary tensors $\mathbb{T}_{A'}$ and $ \mathbb{T}_{B'}$ in the following way to make the 4 leg tensor $\mathbb{T}_{AB'BA'}$, 
\begin{eqnarray}
\includegraphics[scale=0.4]{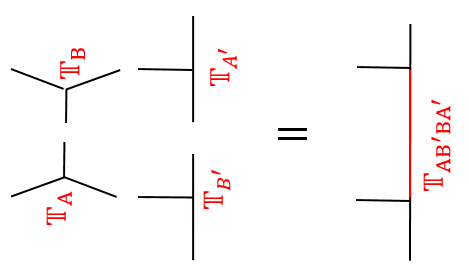}
\end{eqnarray}

\item Reshape the tensor $\mathbb{T}_{AB'BA'}$ into a matrix $M$ where $M_{\alpha\beta', \beta \alpha'} = (\mathbb{T}_{AB'BA'})_{\alpha \beta' \beta \alpha'}$~\cite{Vidal2003}. Now we perform an SVD decomposition of $M$, $M= U\Lambda V^{\dagger}$ and the approximation step: we keep only the highest $D_{\textrm{cut}}$ singular values, and define the new tensors $\mathbb{S}_{A'}$ and $\mathbb{S}_{B'}$ as $(S_{A'})_{\alpha \beta'\gamma }  = U_{\alpha\beta',\gamma}  \sqrt[]{\Lambda_{\gamma,\gamma } }$ and $(S_{B'})_{\gamma \beta \alpha' }  = \sqrt[]{\Lambda_{\gamma,\gamma}}V^{\dagger}_{\gamma, \beta \alpha'}$ where $\gamma$ takes values $1,2,\ldots,D_{\textrm{cut}}$. $\mathbb{S}_{A'}$ and $\mathbb{S}_{B'}$ form an approximate decomposition of  $\mathbb{T}_{AB'BA'}$,
\begin{eqnarray}
\sum_{\gamma=1 }^{D_{\text{cut}}} (S_{A'})_{\alpha \beta' \gamma}  (S_{B'}^{\dagger})_{\gamma \beta \alpha'} \approx  (\mathbb{T}_{AB'BA'})_{\alpha \beta' \beta \alpha'}
\end{eqnarray} 
\begin{eqnarray}
\includegraphics[scale=0.4]{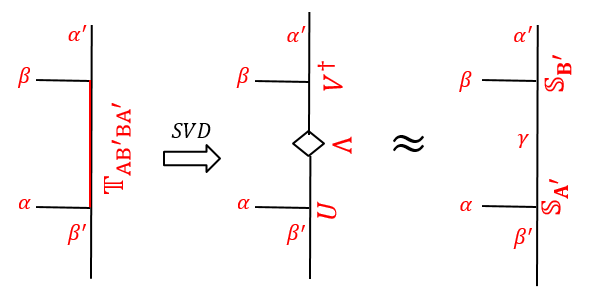}
\end{eqnarray}

\item Check convergence of $\Lambda$.  $\eta \ll 1 $ is the precison tolerance. Let $n$ denote the $n$th recursion step. If $|| \Lambda_n - \Lambda_{n-1 }||_1 < \eta $ exit algorithm.   
\item Put $\mathbb{T}_{A'} =\mathbb{S}_{A'}$ and  $\mathbb{T}_{B'} =\mathbb{S}_{B'}$ and go to step 2. 
\end{enumerate}

\subsection{Numerical result for single-line TNR with random variations}
\label{sec:numerical SLTNR}

\begin{figure}
\includegraphics[width=0.9\columnwidth]{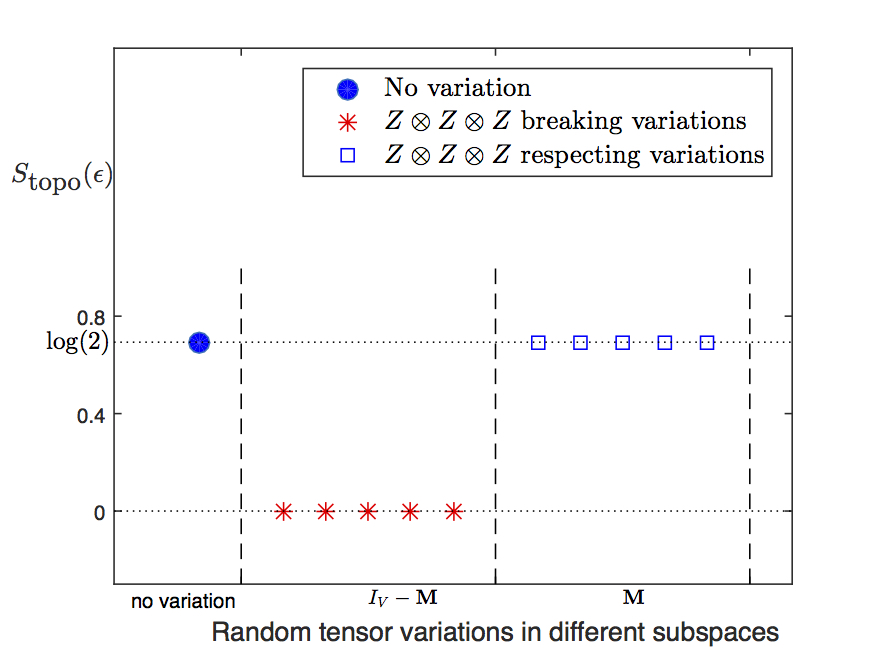}
 \caption{ Numerical calculation of topological entanglement entropy $S_{\text{topo}} (\epsilon)$ of states represented by toric code fixed point single-line tensors, $T^0$, varied with an infinitesimal random tensor in different subspaces. $\epsilon$ value is kept fixed at $\epsilon=0.01$. Blue dot corresponds to $S_{\text{topo}}$ with no variation. $I_V=I^{\otimes 3}$ is the projector on to the full virtual space. $\mathbb{M}=\frac{1}{2}(I^{\otimes 3 }+Z^{\otimes 3})$ is the projector on to the space of variations that respect the $Z^{\otimes 3}$ symmetries. So, $I_V-\mathbb{M}$ is a projector on to the space of variations that break $Z^{\otimes 3}$ symmetries. We see that variations in $I_V-\mathbb{M}$ subspace are unstable while variations in $\mathbb{M}$ are stable. Details of this numerical calculation are given in the appendix \ref{NDSLTC}}.  \label{fig:SLTCrand}
\end{figure}

We use the algorithm described in the previous section to calculate $S_{\text{topo}}$ of the tensor network state constructed by a local tensor with random variations added to the fixed point tensor given in Eq.~\eqref{SLTNReq}.  
$I_V=I^{\otimes 3}$ is projector onto the full virtual space. $\mathbb{M}=\frac{1}{2}(I^{\otimes 3}+Z^{\otimes 3})$ is a projector on to the space of variations that respect the $Z^{\otimes 3}$ symmetries. So, $I_V-\mathbb{M}$ is a projector on to the space of variations that break $Z^{\otimes 3}$ symmetries. We first calculate $S_{\text{topo}}$ in the state constructed by the fixed point tensor, $T^{0}$. 
Then we generate a random tensor $T^r$ on the full space, project it on to the subspace $I_V-\mathbb{M}$, add it to the fixed point value, $T^0 \rightarrow T^0 + \epsilon (I_V-\mathbb{M}) T^r$ and calculate $S_{\textrm{topo}}(\epsilon)$. 
Similarly, we generate a random tensor $T^r$ on the full space, project it on to $Z^{\otimes 3}$ respecting subspace $\mathbb{M}$, add it to the fixed point value, $T^0 \rightarrow T^0 + \epsilon \mathbb{M} T^r$ and calculate $S_{\textrm{topo}}(\epsilon)$. 
We keep the value of variation strength $\epsilon=0.01$ (low enough) to make sure it is not near any phase transition point. The results are shown in Fig.~\ref{fig:SLTCrand}. 
\par  
We see that $Z^{\otimes 3 }$ respecting variations lead to the same topological entanglement entropy as the fixed point state, while $Z^{\otimes 3 }$  violating variations lead to zero topological entanglement entropy. This reproduces the result by  \citet{Chen10}. 

\subsection{Double-line TNR of the toric code state and its instablities}
\label{double-line section}

\begin{figure} 
\includegraphics[width=\columnwidth]{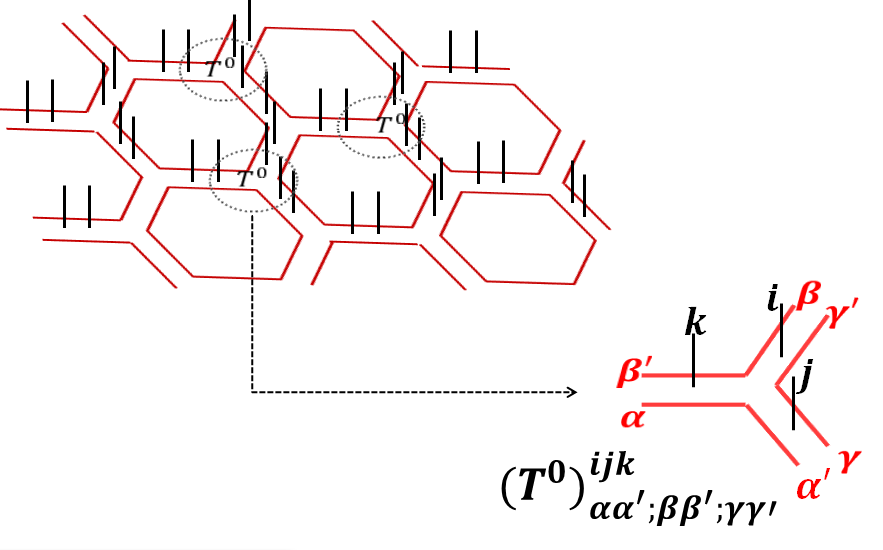}
 \caption{Double-line TNR of the toric code state. We double the local Hilbert spaces in the same way as that for single-line TNR (Fig.~\ref{fig:SLTNR}). We associate to each vertex a tensor $T^{i,j,k}_{\alpha,\alpha';\beta,\beta';\gamma,\gamma'}$, where out of plane legs, $i,j,k$, correspond to the 3 physical indices, and in-plane legs $\alpha,\alpha',\beta,\beta',\gamma,\gamma'$ are the virtual indices. Virtual indices of the tensors contract along the shared edges to produce the toric code state on the physical indices.   }
\label{fig:DLTNR}
\end{figure}

In the double-line TNR of the toric code state, we associate with each vertex a tensor with 3 physical legs and 6 virtual legs, $T^{ijk}_{\alpha\alpha';\beta\beta';\gamma\gamma'}$, (see Fig. \ref{fig:DLTNR}). We will refer to these virtual indices as `plaquette indices' or `plaquette legs' sometimes, because they carry the plaquette degree of freedom that comes from the local Hamiltonian term. All indices take values 0 and 1. We denote the TNR corresponding to the RG fixed point state as $T^0$. (We use the same notation for different fixed point tensors, but it should be clear from the context which fixed point tensor we are discussing.) First property of $T^0$ is that $(T^0)^{ijk}_{\alpha\alpha';\beta\beta';\gamma\gamma'}  \propto \delta_{\alpha\alpha'}\delta_{\beta\beta'}\delta_{\gamma\gamma'}$, that is, indices on the same plaquette assume the same values. The second property is that the physical indices can be considered as labeling the domain wall between the virtual indices. If the two virtual indices in the same direction have the same values (both either $00$ or $11$) then the physical index in the middle has value $0$, otherwise it is $1$.  That is, $i = \beta+\gamma, j=\gamma+\alpha, k=\alpha+\beta$ (all additions are modulo 2). So we can write $T^0$  as
\begin{align*}
(T^0)^{ijk}_{\alpha\alpha';\beta\beta';\gamma\gamma'} = S^{ijk}_{\alpha\beta\gamma} \delta_{\alpha\alpha'}\delta_{\beta\beta'}\delta_{\gamma\gamma'}, 
\end{align*}

\begin{align*}
S^{ijk}_{\alpha\beta\gamma} &=& \begin{cases} 1 & \text{  if }  i = \beta+\gamma, j=\gamma+\alpha, k=\alpha+\beta \\
0 & \text{ otherwise}  \end{cases}.
\end{align*}
We can write all non-zero components explicitly, 
\begin{align}
T^{000}_{00;00;00}= T^{000}_{11;11;11}=1, \quad & 
T^{011}_{00;11;11} =T^{011}_{11;00;00} = 1, \nonumber \\
T^{101}_{11;00;11}= T^{101}_{00;11;00} =1, \quad &
T^{110}_{11;11;00} = T^{110}_{00;00;11}= 1. \nonumber \\
\label{DLTNReq}
\end{align}
\par

Is double-line TNR unstable too? We find that it is unstable as well in certain directions of variations. Similar to the single-line case, a variation is stable or unstable depending on whether or not it violates certain virtual symmetries. So let's first look at the symmetries of the double-line TNR. It has 6 virtual indices, so the virtual space dimension is $2^6=64$, while the physical space dimension is again 4. So we need a symmetry group with $ |G|= 64/4=2^4$. Indeed the tensor has a $Z_2 \times Z_2 \times Z_2 \times Z_2$ virtual symmetries. First it has a $X^6$ symmetry. That is, if we flip all the six virtual indices, the tensor remains the same. Second, it has three $Z\otimes Z$ symmetries, where $Z\otimes Z$ are applied to the two virtual indices on the same plaquette. So the double-line tensor in \eqref{DLTNReq} satisfies these symmetry equations:
\begin{eqnarray}\label{DLTCsym}
\includegraphics[width=0.8\columnwidth]{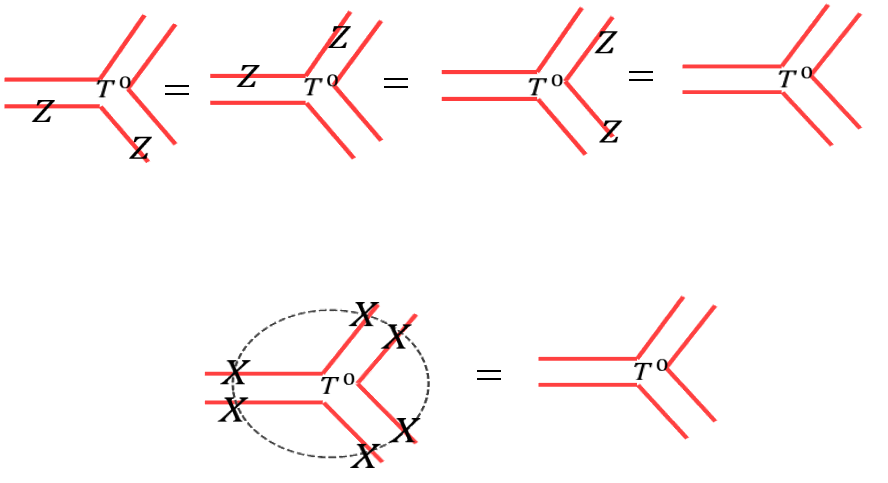}.
\end{eqnarray}
Single-line TNR had only one such $Z_2$ symmetry and it turned out that breaking it results in phase transition. For double-line we have four $Z_2$ symmetries. So the question is, are all of them important? That is, is it the case that breaking any of them with a variation leads to instability? Indeed many different possible kinds of variations are possible:
\begin{eqnarray}\label{DLTCpert}
\includegraphics[width=0.7\columnwidth]{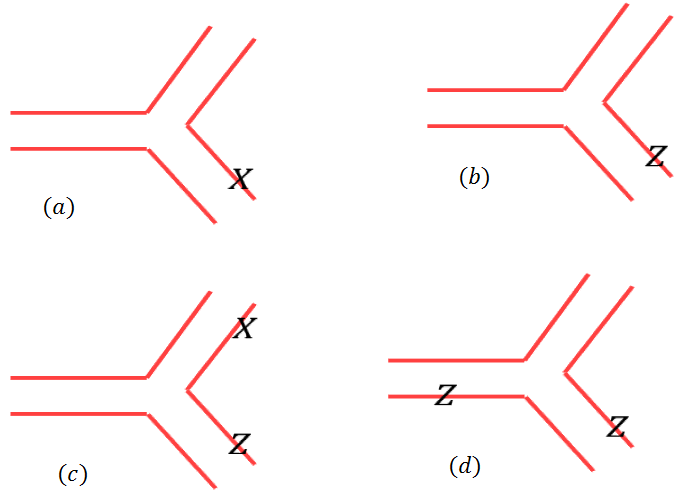}.
\end{eqnarray}
A variation can violate $Z\otimes Z$ but not $X^{\otimes 6}$ (for example, \ref{DLTCpert}(a) ), or it can violate $X^{\otimes 6}$  but not $Z\otimes Z$ (for example, \ref{DLTCpert}(b) ), or it can violate  both (for example, \ref{DLTCpert}(c) ), or it can violate neither(for example, \ref{DLTCpert}(d) ), etc.  So  to find out, we need to look at the unstable directions of variations of the fixed point tensor. \par
Our numerical calculation reveals an interesting result. We find that (see Fig. \ref{DLTCrand} )
\begin{enumerate}
\item{ If a variation violates any of the $Z\otimes Z$ symmetries then it is stable. }
\item{If a variation respects all $Z\otimes Z$  symmetres then there are two subcases \begin{enumerate}
\item If it respects  $X^{\otimes 6}$ symmetry then it is stable. 
\item If it breaks $X^{\otimes 6}$ symmetries then it is unstable. 
\end{enumerate}}

\end{enumerate}
\begin{figure}
\includegraphics[width=0.9\columnwidth]{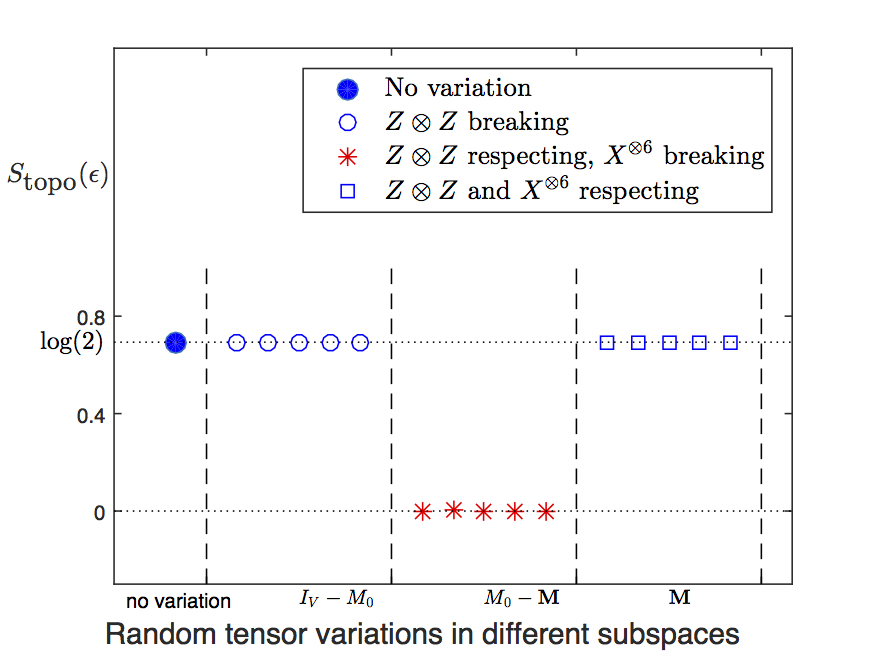}
\caption{ Numerical calculation of topological entanglement entropy $S_{\text{topo}} (\epsilon)$ of the states represented by toric code fixed point double-line tensors, $T^0$, varied with an infinitesimal random tensor in different subspaces. $\epsilon$ value is kept fixed at $\epsilon=0.01$. Blue dot corresponds to $S_{\text{topo}}$ with no variation. $I_V$ is projector onto the full virtual space. $M_0$ is the projector on the stand-alone subspace. $\mathbb{M}$ is the MPO-injective subspace projector.  We take a random tensor and apply the projectors to generate random tensors in respective subspaces. Variations in $I_V-M_0$ violate $Z \otimes Z$ symmetry. Variations in $M_0-\mathbb{M}$ violate $X^{\otimes 6}$ but not $Z\otimes Z$. Variations in $\mathbb{M}$ violate no virtual symmetry. The details of this numerical calculation are given in appendix A.2.}
\label{DLTCrand}
\end{figure}
So we see that the relation between unstable variations and virtual symmetries of the double-line TNR is more complicated than that for single-line TNR. The $X^{\otimes 6}$ symmetry looks like the $Z^{\otimes 3}$ symmetry of the single-line TNR, as they both operate as a loop operators, and unstable variations in both TNR violates these loop symmetries. However, the crucial difference in double-line is that then there are extra symmetries (the 3 $Z^{\otimes 2}$ symmetries) which have an exact opposite relation with the unstable variations: a variation is actually \textit{ stable} when it violates these symmetries (irrespective of whether or not it violated the loop symmetry). It indicates that the sources of these two kinds of symmetries must be different.

How can we understand this phenomena? We will now show that the sources of these symmetries are indeed different, and it is the interplay between these  two symmetries that determines the tensor instability phenomena. The symmetries whose violation causes instability comes from the so-called \textit{MPO-injective subspace}\cite{csahinouglu2014characterizing} of the virtual space, while the symmetries whose violation causes stability comes from what we define to be \textit{stand-alone subspace.}  We now define these subspaces precisely and put forward the conjecture regarding their relationship to TNR instability. 
\section{Virtual symmetries/subspaces of a TNR and Tensor-Instability Conjecture}
\label{sec: virtual subspaces}
We first give a constructive definition of stand-alone subspace of a TNR.
\subsection{Stand-alone subspace}
\label{stand-alone section}
A generic tensor can be thought of as a linear map from virtual vector space to the physical vector space. Consider a generic tensor, $T^I_{\alpha}$ where $I$ is the collection of all physical indices and $\alpha$ is the collection of all virtual indices.\textit{ Double tensor} $\mathbb{T}$ of a tensor $T$ is defined as $\mathbb{T}=\sum_I T^I_{\alpha } (T^{*})^I_{\alpha'}$. For example, for single-line TNR the double tensor can be represented diagrammatically as follows
\begin{eqnarray}
\includegraphics[scale=0.3]{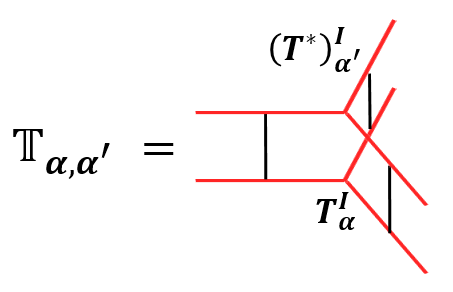}
\end{eqnarray}
In general we would think of double tensor as having two layers of virtual indices, lower and upper. Double tensor can be interpreted as a density matrix on the virtual space. Now consider the double tensor of a RG fixed point TNR, $\mathbb{T}^0$, contracted over some large region $R$. Let's say we remove $\mathbb{T}^0$ from one site and replace it with some other double tensor, $\mathbb{T}$ as follows:
\begin{eqnarray}
\includegraphics[scale=0.4]{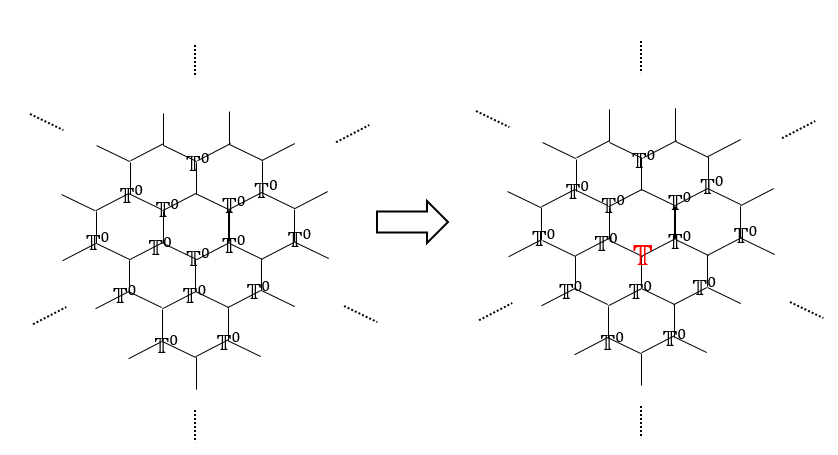}.
\end{eqnarray}

What do we get? In particular, are there tensors $\mathbb{T}$ such that this replacement \textit{collapses} the whole tensor network? By `collapse', we mean that we simply get zero upon contraction. The answer turns out to be yes for tensors that represent topological order. In fact, as we will see later, \textit{most} tensors $ \mathbb{T}$ will collapse the fixed point tensor network upon replacement. It turns out that only the tensors supported on a particular subspace of the full virtual space can replace the fix point tensor without collapsing the whole tensor network. We will call this space \textit{the stand-alone subspace of the TNR.} Now we will give a systematic way of calculating this subspace for a given fixed-point TNR. \par 
\begin{figure}
\includegraphics[ scale=0.3]{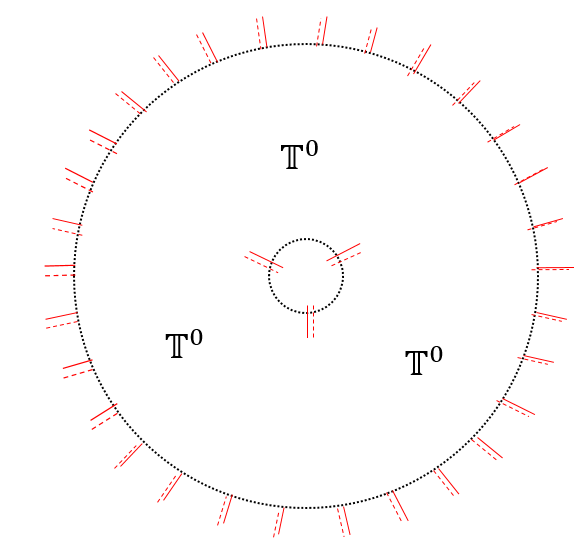}
\caption{Calculation of stand-alone space. We put the fixed point double tensor network on a large disc with a hole at the origin (one double tensor removed). This tensor network has dangling virtual indices (red legs) at the outer and inner boundaries. We trace out the virtual indices at the outer boundary, and the support space of the remaining tensor at the inner boundary gives us the stand-alone space.}
\label{standalone2}
\end{figure}
Consider contracting the fixed-point double tensors $\mathbb{T}^0 $ on a large disc with an open boundary. Now we remove the tensor at the origin. We get a tensor network with dangling virtual indices at the origin and at the boundary of the disc. We want to find out the space of tensors that can be put on the origin without collapsing the tensor network. We do not care what tensor at the boundary we get. So we trace out the indices at the outer boundary (i.e. contract upper and lower virtual indices with each other). This leaves us with a tensor at the origin. \textit{The support space of this tensor will be precisely the stand-alone space.} Any tensor supported on this subspace can stand alone with the surrounding tensor being the fixed-point tensors. \par 
Note that definition of stand-alone space implies that it is a vector space. If $T_1$ and $T_2$ doesn't collapse the tensor network then $aT_1+bT_2$ will also not collapse the network. 
\par 
\textbf{Stand-alone subspace of double-line TNR:} Now let's first calculate the stand-alone subspace of double-line TNR of toric code as it is more interesting than that of single-line TNR. The double tensor of $T^0$ in \eqref{DLTNReq} can be written as (ignoring an overall normalization factor)
\begin{eqnarray}\label{DoubleTensorDL}
\mathbb{T}^0 &=& \sum_{I} (T^0)^I_{\alpha } (T^{0;*})^I_{\alpha'} \nonumber \\
&=& (I^{\otimes 2}+Z^{\otimes 2})^{\otimes 3}(I^{\otimes }+X^{\otimes 6}),
\end{eqnarray}
where the double tensor is written as an operator between the lower virtual indices and upper virtual indices. The $Z^{\otimes 2}$ and $X^{\otimes 6}$ act in the way it is shown in Fig.~\ref{DLTCsym}. We need to contract this tensor on a disc with a hole at the origin. To contract  two tensors given in an operator form, we need to multiply them and take a trace on the shared indices. A cumbersome but straight-forward calculation shows that double tensor contracted on a region $R$ give (ignoring an overall normalization factor)
\begin{eqnarray}
\mathbb{T}^0(R) &=& \left(I^{\otimes 2}+Z^{\otimes 2}\right)^{\otimes m}(\partial R)\nonumber \\ 
& &\left( I^{\otimes 2m}(\partial R)+X^{\otimes 2m}(\partial R)\right),
\end{eqnarray}
where $\partial R$ denote the boundary of $R$, and $m=|\partial R|$ is the length of the boundary. $O(\partial R)$ means the operator $O$ is applied on the virtual legs along the boundary $\partial R$. We will omit this when it is clear from the context which leg the operator is being applied on. The region we want is a disc with a vertex removed, $R=D_{2m}-D_6$. $D_n$ denotes the disc with $n$ virtual legs at the boundary. It has two disconnected boundaries, one the boundary of $D_{2m}$ and other the boundary of $D_6$ (with opposite orientation).
\begin{figure}[t]
\begin{center} 
\includegraphics[trim=10mm 0mm 30mm 20mm, scale=0.3]{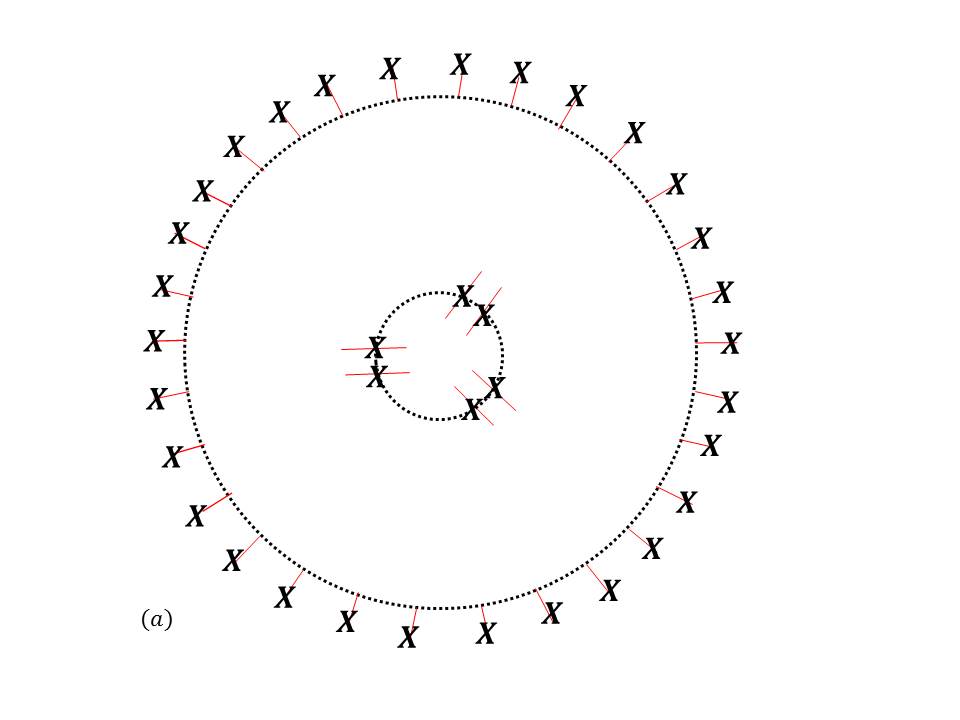}
\vfill
\includegraphics[trim=10mm 10mm 30mm 20mm, scale=0.3]{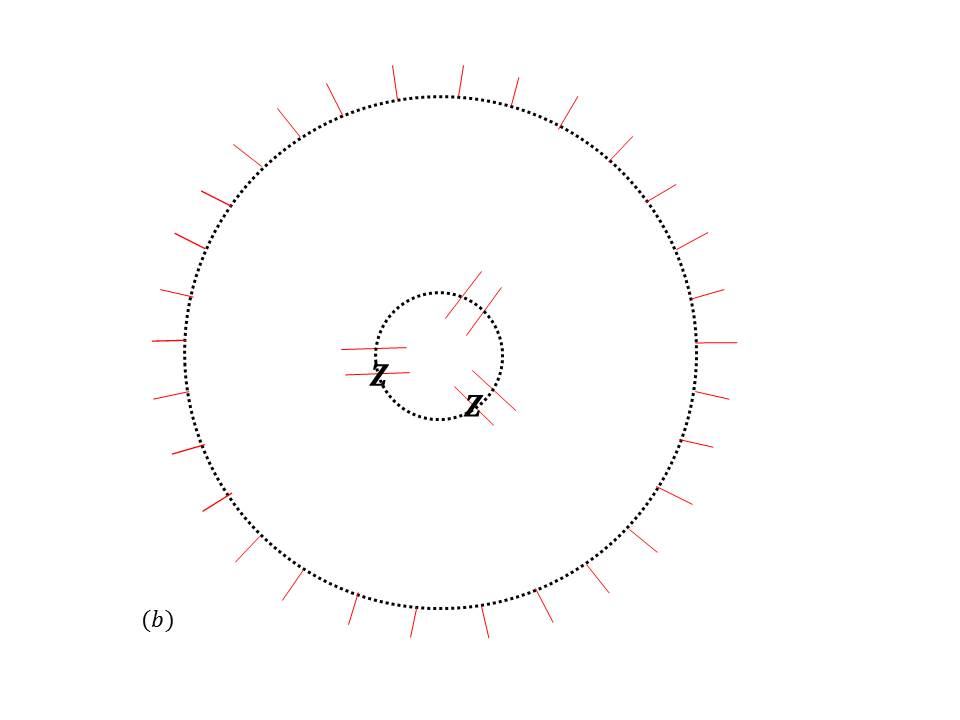}

\caption{$X$ and $Z$ operators appear differently in the toric code double-line double tensor contracted on a disc with a hole. $X$ operators on inner boundary only appears with $X$ operators on the outer boundary, which vanishes upon taking the trace. But $Z$ operators on the inner boundary appear with identity on the outer boundary. So these terms survive the trace. That is why $Z^{\otimes 2}$ symmetry is imposed on the stand-alone space but not the $X^{\otimes 6}$ symmetry. }
\label{standalone X Z}
\end{center} 
\end{figure}

\begin{eqnarray}
\mathbb{T}(D_{2m}-D_6) &=& \left(I^{\otimes 2}+Z^{\otimes 2}\right)^{\otimes m}\otimes\left(I^{\otimes 2}+Z^{\otimes 2}\right)^{\otimes 3}\nonumber \\ & &
\left( I^{\otimes 2m}\otimes I^{\otimes 6}+ X^{\otimes 2m}\otimes X^{\otimes 6}\right).
\end{eqnarray}
As explained in Fig.~\ref{standalone X Z}, $X$ operators act on the two boundaries simultaneously, but $Z$ operators act independently.   Now to get the stand-alone space at the origin, we need to trace out the virtual legs at the boundary of $D_{2m}$. If we expand the expression for $\mathbb{T}(D_{2m}-D_6)$ above and apply trace on the operators on the outer boundary, only the terms with identity on the outer boundary survive. $X$ operator does not have such a term, but $Z$ does. So finally, tracing out the outer boundary leaves only $Z^{\otimes 2}$ on the inner boundary. That is, we get the following tensor on the 6 virtual indices incident on a singe vertex 
\begin{eqnarray}
B_0 = \left(I^{\otimes 2}+Z^{\otimes 2}\right)^{\otimes 3}.
\end{eqnarray}
$B_0^2=8B_0$, so 
\begin{eqnarray}
M_0= \frac{1}{2}B_0=\frac{1}{8}\left(I^{\otimes 2}+Z^{\otimes 2}\right)^{\otimes 3}
\end{eqnarray}
is a projector on to the support space of $B_0$. $M_0$ defines the stand-alone space of double-line TNR of toric code. Any tensor $T$ that satisfies $M_0 T \neq 0$ can `stand alone'. $M_0$ will be used to denote the projector on to stand-alone space throughout the paper. So we see that\textit{ only } the tensors that respect the $Z\otimes Z$ symmetry can stand alone. The $X^{\otimes 6} $ symmetry, however, is not required to define the stand-alone space.  \par 
\textbf{Stand-alone subspace of single-line TNR:} We can also calculate the stand-alone subspace of single-line TNR of toric code. One can calculate the double tensor of single-line TNR in Fig. \ref{SLTNReq} to be
\begin{eqnarray}\label{DoubleTensorSLTNR}
\mathbb{T}^0 &=&\frac{1}{2} (I^{\otimes 3} + Z^{\otimes 3}).
\end{eqnarray}
This double tensor upon contraction on the disc with a hole at the origin ($D_{m}-D_3$) gives (up to an overall normalization)
\begin{eqnarray}
\mathbb{T}^0(D_{m}-D_3) &=& I^{\otimes m} \otimes I^{\otimes 3}+Z^{\otimes m}\otimes Z^{\otimes 3}.
\end{eqnarray}
Now contracting the outer circle gives us 
\begin{eqnarray}
B_0 = M_0 = I^{\otimes 3}.
\end{eqnarray}
So we see that, for single-line TNR, the stand-alone subspace is actually \textit{all} of the virtual space. That is, there are no tensors that cannot stand alone. \par 

\subsection{MPO-injective subspace}
\label{MPO subspace}
Here we repeat the definition of MPO-injective subspace given in Ref. \onlinecite{csahinouglu2014characterizing} for convenience.  \par  
As explained above, stand-alone subspace is the maximal virtual subspace such that any tensor supported on this subspace can be inserted into the tensor network without collapsing it. Therefore, the virtual space of the RG fixed point tensor, $T^0$ itself must be inside the stand-alone subspace. This virtual space of $T^0$, which is by definition a subspace of the stand-alone space, is what we would call the \textit{MPO-injective subspace}. The reason for calling it an MPO-injective space is that, as it turns out, this space is protected by symmetry operators which are Matrix Product Operators (MPO). Let's make the notion of virtual space of $(T^0)^I_{\alpha}$ precise. We can think of it as a matrix with indices $\alpha$ and $I$ and perform an SVD decomposition,
\begin{eqnarray}\label{TSVD}
(T^0)^I_{\alpha} =\sum_{\alpha',I'} V_{\alpha,\alpha'} \Lambda_{\alpha',I'} P_{I',I},
\end{eqnarray}
where $V$ and $P$ are unitary matrices in virtual and physical space, and $\Lambda$ is the diagonal matrix containing the singular values. 
\begin{mydef} \label{MPOdef1}
The MPO-injective space, defined as the virtual support space of $T^0$,is the virtual subspace spanned by columns of $V$ for which corresponding singular value is non-zero.
\end{mydef}
 Another way to think about this is to again consider the double tensor $\mathbb{T}^0_{\alpha,\alpha'}= \sum_I (T^0)^I_{\alpha} (T^{0;*}))^I_{\alpha}$ which is a matrix in the virtual space. An equivalent but more useful definition is, 
 \begin{mydef} \label{MPOdef2}
 The MPO-injective space is the space spanned by eigenvectors of $\mathbb{T}^0$ with nonzero eigenvalues.
 \end{mydef}
Using Eq.~\eqref{TSVD} we can write 
\begin{eqnarray}
T^0 &=& \sum_{j;\lambda_j \neq 0} \lambda_j |v_j\rangle \langle p_j| \\
\Rightarrow \mathbb{T}^0 &=&\sum_{j;\lambda_j \neq 0} \lambda_j^2 |v_j\rangle\langle v_j|,
\end{eqnarray}
where $\lambda_j$ are the singular values and $v_j$ and $p_j$ are the corresponding vectors in virtual and physical space. MPO-injective space is the space spanned by $v_j$, so the projector on this space is 
\begin{eqnarray}
\mathbb{M}= \sum_{j;\lambda_j \neq 0} |v_j\rangle \langle v_j|,
\end{eqnarray}
The mathematical understanding of the MPO-injective space is that it is the virtual subspace which is isomorphic to the ground state physical subspace and $T^0$ is the isomorphism. MPO-injective subspace by definition nested inside the stand-alone subspace, which by definition is nested inside the full virtual space. Similarly, the ground-state physical space is by definition a subspace of the full physical space. These spaces are represented visually in Fig.~\ref{venndiag} for clarity. \par 
\textbf{MPO-injective subspace of single-line TNR:} We can write the single-line tensor in Eq.~\eqref{SLTNReq} in the Eq.~\eqref{TSVD} form,
\begin{eqnarray}
T^0 = |000\rangle\langle 000|+|011\rangle\langle 011|+|101\rangle \langle 101 | +|110\rangle \langle 110|. \nonumber \\
\end{eqnarray}
Of course, this happened to be already written in SVD decomposed form. So the MPO-injective space is spanned by vectors $\lbrace |000\rangle, |011\rangle,|101\rangle, |110\rangle \rbrace$. So the MPO projector is
\begin{eqnarray}\label{SLMPO}
\mathbb{M} &=& |000\rangle\langle 000|+|011\rangle\langle 011|+|101\rangle \langle 101 | +|110\rangle \langle 110|   \nonumber \\
&=& \frac{1}{2}\left( I^{\otimes 3} +Z^{\otimes 3} \right).
\end{eqnarray}
We see that this projector can be written as a translation invariant superposition of tensor product of matrices. That's why we call it the MPO-injective subspace. Remember that stand-alone space of single-line TNR was determined to be all of virtual space, $M_0=I_V$. So as expected, $\mathbb{M} \subset M_0$. 
\par
\textbf{MPO-injective subspace of double-line TNR:} For calculation of the MPO-injective subspace of double-line tensor in Eq.~\eqref{DLTNReq} we use the second definition in \ref{MPOdef2} above to avoid a cumbersome but straight forward calculation. We already calculated the double tensor of double-line TNR in Eq.~\eqref{DoubleTensorDL}. We ignored the normalization factor there. If we use a normalization factor of $\frac{1}{16}$  and write 
\begin{eqnarray}\label{DLMPO}
\mathbb{M}=\frac{1}{8}\mathbb{T}^0 &=& \frac{1}{16}\sum_{I} (T^0)^I_{\alpha } (T^{0;*})^I_{\alpha'} \nonumber \\
&=& \frac{1}{16} (I^{\otimes 2}+Z^{\otimes 2})^{\otimes 3}(I^{\otimes }+X^{\otimes 6}).
\end{eqnarray}
Then $\mathbb{M}$ is a projector, that is, it satisfies $\mathbb{M}^2=\mathbb{M}$, and has the same support as $\mathbb{T}^0$ hence this is the desired MPO projector. Remember that  stand-alone projector was calculated to be $M_0 =  \frac{1}{8} (I^{\otimes 2}+Z^{\otimes 2})^{\otimes 3}$, and hence $\mathbb{M} = M_0 \frac{1}{2} (I^{\otimes }+X^{\otimes 6}) \subset M_0 $, as expected. 

\subsection{TNR instability conjecture}
\label{TNR instability conjecture}
\begin{figure}
\begin{center}
\includegraphics[width=0.5\columnwidth]{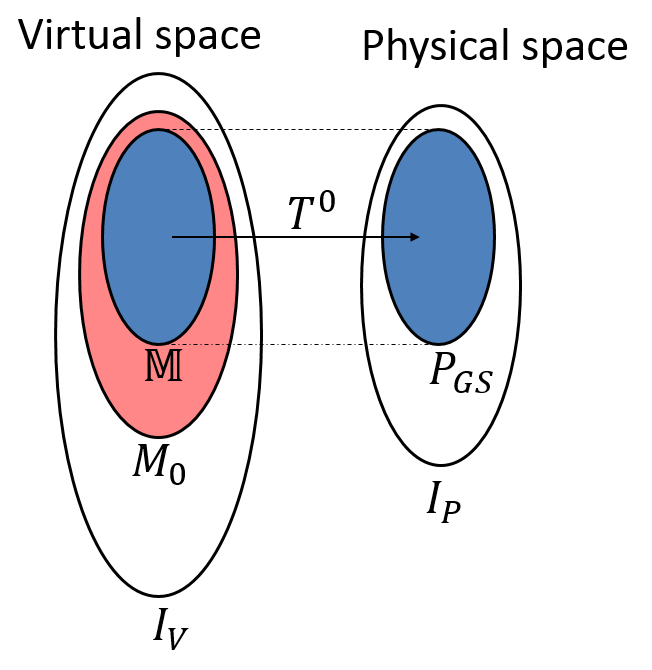}
\caption{A pictorial representation of relevant vector spaces. A tensor, $T^0$, is a linear map from the virtual space to the physical space. We denote  the full virtual space as $I_V$ and the full physical space as $I_P$. $M_0$ (region in red and blue) is the stand-alone subspace of the virtual space. $\mathbb{M}$ (region in blue) is the MPO-injective subspace of the stand-alone space. MPO-injective subspace is isomorphic to the local ground-state physical subspace (also in blue), denoted as $P_{GS}$, which is a subspace of the full physical space. }
\label{venndiag}
\end{center}
\end{figure}
Now that we have defined the stand-alone and MPO-injective subspaces precisely, we are ready to state the central conjecture of this work.
\begin{conjecture} 
If, for a given RG fixed point TNR, $T^0$, of a topological state, $M_0$ and $\mathbb{M}$ are the projectors onto the stand-alone and MPO-injective subspaces as defined above, then an infinitesimal tensor variation $ T^0 \rightarrow T^0+ \epsilon T $ changes the topological phase of the state if and only if $(M_0 -\mathbb{M})T \neq 0 $. 
\end{conjecture}
It implies that the projector onto the stable space is $P_S=I_V-(M_0-\mathbb{M})$. 
\begin{mycorol}
An infinitesimal variation $T$ does not change the topological phase if and only if $P_ST=T$.
\end{mycorol}
\begin{mycorol}
For any tensors $T$, the variation $  T^0 \rightarrow T^0+ \epsilon P_S T $, with $\epsilon \ll 1$ does not change the topological phase. 
\end{mycorol}

Or in simple words, \textit{ a variation is unstable if and only if it has a component in the stand-alone space that is outside the MPO-injective subspace.} A pictorial representation of the decomposition of the virtual space through these projectors is shown in Fig.~\ref{venndiag}. We will denote $(M_0 -\mathbb{M})$ as $P_U$ for convenience. Note that $P_U$ shouldn't be thought of as `the projector onto unstable subspace' because unstable variations do not form a vector space, as opposed to stable variations that do form a vector space. It is because $P_UT^1\neq 0$ and $P_U T^2 \neq 0$ does not imply $P_U(T^1+T^2)\neq 0$. \par

 Let's first see how this conjecture is true for the single-line and double-line TNR of the toric code. For single-line we have already calculated the stand-alone and MPO-injective subspaces and found $M_0= I_V$ and $\mathbb{M} = \frac{1}{2}\left( I^{\otimes 3} +Z^{\otimes 3} \right) $. So,
 \begin{eqnarray}
 P_U=M_0 - \mathbb{M} = \frac{1}{2}\left( I^{\otimes 3} -Z^{\otimes 3} \right) .
 \end{eqnarray}
 So for a tensor $P_UT  \neq 0$ if and only if it violates the $Z^{\otimes 3}$ symmetry. Indeed, this is exactly what we saw numerically in Fig~\ref{fig:SLTNR}. All variations can be summarized visually using Fig.~\ref{venndiag} as follows:
\begin{eqnarray}
\includegraphics[scale=0.3]{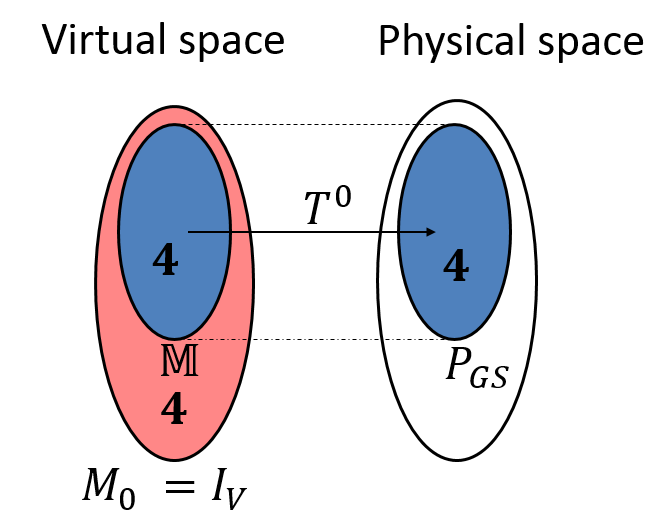}
\end{eqnarray}
Variations supported on the red region are unstable, while those on blue and white are stable. The dimension of each space is indicated. So we see that the space  $(M_0 -\mathbb{M})$ is 4 dimensional space spanned by basis 
\begin{eqnarray}
\includegraphics[width=0.9\columnwidth]{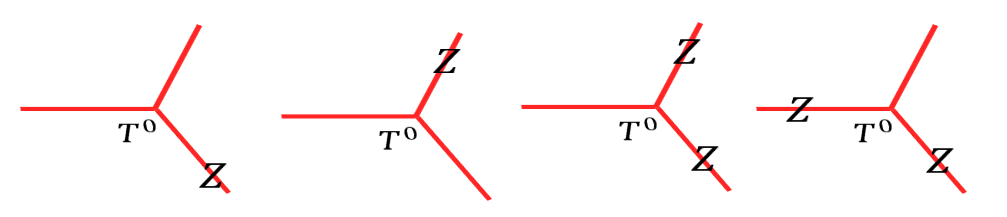}
\end{eqnarray}
It is wrong to think that these basis set span the space of unstable variations, because unstable variations do not form a vector space. All we can say is if a variations has overlap with any of these basis, it would cause instabilty. 

For double-line TNR we found, 
\begin{eqnarray}
M_0 &=& \frac{1}{8}\left(I^{\otimes 2}+Z^{\otimes 2}\right)^{\otimes 3} \nonumber \\
\mathbb{M} &=&  M_0 \frac{1}{2} (I^{\otimes }+X^{\otimes 6})
\end{eqnarray}
 So,
\begin{eqnarray}
 P_U &= & M_0-\mathbb{M} \nonumber \\
 &= &\frac{1}{8}\left(I^{\otimes 2}+Z^{\otimes 2}\right)^{\otimes 3}\frac{1}{2} (I^{\otimes }-X^{\otimes 6}) \nonumber
\end{eqnarray}
 So $P_UT \neq 0$ if and only if $T$ satisfies the three $Z^{\otimes 2}$ symmetries but violates the $X^{\otimes 6}$ symmetry. Indeed this is precisely what we found numerically as shown in Fig.~\ref{fig:DLTNR}. All variations can be summarized visually as follows
\begin{eqnarray}
\includegraphics[scale=0.3]{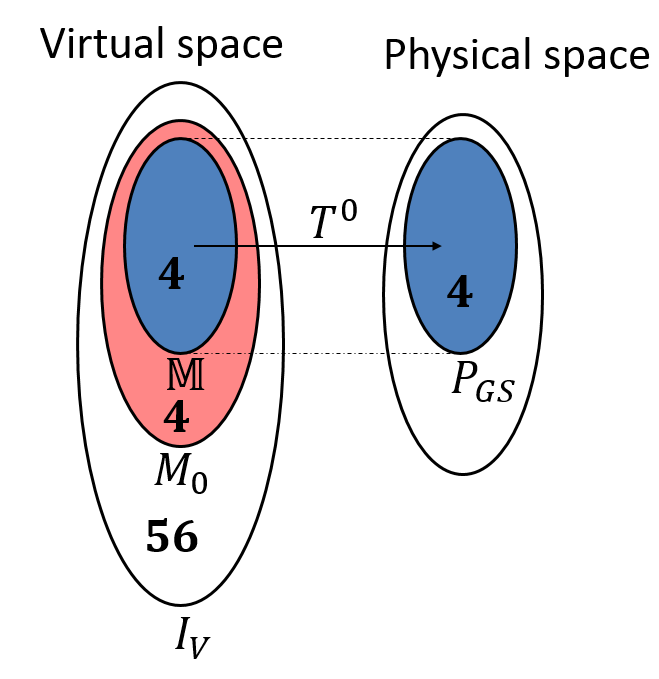}
\end{eqnarray}
Variations supported on the red region are unstable, while those on blue and white are stable. The dimension of each space is indicated. So we see that the space  $(M_0 -\mathbb{M})$ is 4 dimensional space spanned by basis
\begin{eqnarray}\label{DLTCspan}
\includegraphics[width=0.9\columnwidth]{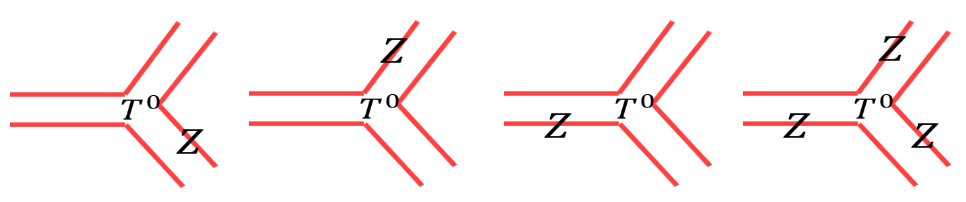}
\end{eqnarray}
\par 

\section{Physical understanding of TNR instability}
\label{sec: physical understanding}
TNR-instablity conjecture would predict mathematicall exactly which variations would be unstable. But what is the physical reason behind such instabilities? To answer, we put forward the following physical conjecture, which we would justify and explain in detail in the rest of this section. 

\begin{conjecture}\label{physical conjecture}
Variations in stand-alone subspace $M_0$ correspond to `bosonic excitations' that can proliferate/condense in the given TNR. Variations in the MPO-injective subspace $\mathbb{M} $ are the subset of these condensable `bosons' that are trivial (belong to trivial superselection sector). Hence the variations in $M_0-\mathbb{M}$ are the non-trivial condensable bosons. So such a variation results in topological boson condensation and causes a topological phase transition of the state.
\end{conjecture}
By `excitation' we mean any point-like variation to the ground state, or its TNR. It should be carefully noted that the word `boson' here refers to any point like excitation (not necessary an irreducible excitation) that has trivial topological spin.  For example, if $a$ is an anyon of the given model, then composite particle $a\bar{a}$ where $a$ and $\bar{a}$ are sitting next to each other is included in this definition of boson. Of course it is a topologically trivial boson. Similarly, if we apply any local operation on the topological state, we would say that the resulting state contains a boson. Though, of course, it is again a topologically trivial boson. \par 
Now we turn to the first part of the claim, which is basically the physical significance of stand-alone space 
\subsection{Physical understanding of stand-alone space $M_0$ }
\label{physical significance stand-alone}
As claimed, the physical significance of stand-alone space is that it contains proliferatable bosonic excitations. 
\subsubsection{Proliferatable variations of a TNR}
First we explain what we mean by `Proliferatable variations/excitations'. (We use the term 'variation' for any mathematical variation to the ground state tensors. `Excitation' should be used for a quasi-particle excitation. But in slight abuse of  the nomenclature we would often use them interchangeably. It is justified as we are only working with the wave functions and not Hamiltonians.)  Let's say $T^0$ is the RG fixed point tensor of some topological ground state wave function $|\Psi_0\rangle$. Let's say we add a variation, $T^0 \rightarrow T^0+\epsilon T$ and the resulting wave function is $|\Psi\rangle$. 
\begin{eqnarray}
|\Psi_0\rangle &=& \sum_{\lbrace i_j \rbrace } (T^0)^{i_1} (T^0)^{i_2}\ldots (T^0)^{i_n} | i_1 i_2 \ldots i_n\rangle,  \nonumber \\ 
|\Psi\rangle &= & \sum_{\lbrace i_j \rbrace } (T^0+\epsilon T)^{i_1} (T^0+\epsilon T)^{i_2}\ldots \times\nonumber  \\ & & (T^0+\epsilon T)^{i_n} | i_1 i_2 \ldots i_n\rangle\nonumber  \\ 
&=& |\Psi_0\rangle +\epsilon \sum_{s_1}|\Psi_{s_1}\rangle+ \epsilon^2\sum_{s_1,s_2} |\Psi_{s_1,s_2}\rangle+ \ldots,   \label{TCv1v2}
\end{eqnarray}
where $|\psi_{s_1}\rangle$ denotes the tensor network state similar to $|\Psi_0\rangle$ except $T^0$ has been replaced with $T$ at site $s_1$. Similarly, $|\Psi_{s_1,s_2}\rangle$ denotes the tensor network state similar to $|\Psi_0\rangle$ except $T^0$ is replaced with $T$ at site $s_1$ and $s_2$. Higher order terms can be understood in a similar manner. Physically, $|\psi_{s_1}\rangle$ can be interpreted as  `excitation' $T$ (which may be trivial) sitting at site $s_1$ with probability $\epsilon^2$. Similarly,  $|\Psi_{s_1,s_2}\rangle$ can be interpreted as excitation $T$ sitting at sites $s_1$  and $s_2$  with probability $\propto \epsilon^4$. Higher order terms can be interpreted in a similar fashion. Though $\epsilon^2$ looks small compared to the weight of $|\Psi_0\rangle$, one has to bear in mind there are $\sim N $ such terms in the expansion, where $N$ is the number of sites. So after normalization they can have comparable weights. \par
When $T$ is in the stand-alone space then it can appear anywhere in the tensor network state, independent of each other, even at large scales. However when $T$ is outside of the stand-alone space, then it can at most appear next to other $T$s. But then the distance between excitations is exponentially suppressed since each $T$ appears with an $\epsilon$ weight. So  such excitations do not appear at large scale and would vanish under RG process. Tensors within the stand-alone space, on the other hand, can appear at any scale and would not vanish under RG process. So we can call the new wave function as a `proliferation/condensate of $T$', since the variation/excitation $T$ proliferates and each site is in superposition of $T$ appearing and not appearing at all length scales. (We caution that we use the term `proliferation' to denote the mathematical fact that the wave function is a superposition of a variation appearing everywhere. While the term `condensate' in physics means something more specific. But, again, we would use these terms interchangeably. It is justified as we are not dealing with the Hamiltonians, rather looking at the changes in the wave functions as we vary the tensors. So the `condensation of variations'  doesn't necessarily mean a phase transition. It just means a particular mathematical variations, which can be interpreted as an excitation, proliferates and the resulting wave function is a superposition of this variation appearing everywhere.) \par
A key point here is that $v_1$ and $v_2$ can be at arbitrary distance from each other but the contribution of this term in the superposition remains $\epsilon^2$. Let's compare this with how the ground state changes with respect to a perturbation on the Hamiltonian level. Let's perturb the toric code Hamiltonian in  \eqref{TCHamiltonian} with $X$ perturbations on every link,
\begin{eqnarray}
H=H_0 + \epsilon \sum_l X_l.
\end{eqnarray}
The ground state of this perturbed Hamiltonian is also a superposition of $|\Psi_0\rangle$ and terms like $|\Psi_{v_1,v_2}\rangle$. But the weight that appears with $|\Psi_{v_1,v_2}\rangle$ is of the order of $\epsilon^{\text{distance}(v_1,v_2)}$, that is, the separation between two $e$ particles is exponentially suppressed. So, in thermodynamic limit, these excitations disappear. But this is not the case with state in Eq.~ \eqref{TCv1v2}. That is why the state in Eq.~\eqref{TCv1v2} cannot be produced by infinitesimal  small local perturbation of the parent Hamiltonian. \par 

So we have argued that stand-alone space, by definition, is the space of variations that can condense. But how do we know they are `bosonic excitations', that is, they have a trivial topological spin? We will show it now. \par 
\subsubsection{Condensable excitations are `bosons'}
Consider the tensor network state which has the fixed point tensor $T^0$ everywhere except at sites $s_1$ and $s_2$, where $T^0$ has been replaced by stand-alone tensors $T$. We denote this wave function as $|\Psi_{s_1,s_2}\rangle $, as above. Topological spins of quasi-particles in topological models are calculated using the string-operators that create them. So we need to first define a string-operators that create these variations. The anyonic string operators in topological models have the property that they commute with the Hamiltonian everywhere except possibly at  its ends. But we are working directly with the quantum wave function and are not really concerned with the underlying Hamiltonian, whose form can change going away from the RG fixed point. We see that we can define an appropriate string-operator for tensor network states without referring to a Hamiltonian. To do that, first notice that every tensor network state has underlying \textit{gauge symmetries} at the virtual level. That is, if we apply  operators $A$ and $B$ on the two contracting virtual legs, such that $AB=I$, the tensor network state doesn't change (though the individual tensors may change). That is, 
\begin{eqnarray}
\includegraphics[scale=0.3]{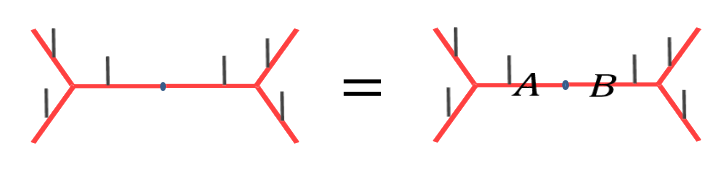}.
\end{eqnarray}
It means that if we apply a string of $A$, $B$ on virtual levels along a path, the tensor network state would not change along the path but only at the ends. For example, on the double-line TNR we can create a stand-alone excitation $A$ in the following way, 
\begin{eqnarray}
\includegraphics[width=0.9\columnwidth]{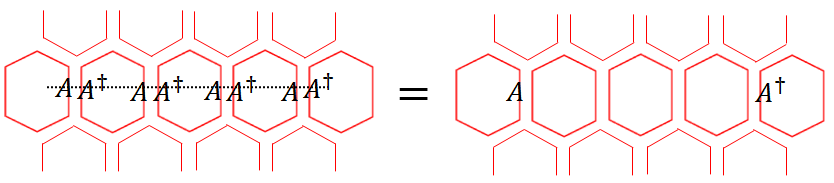}.
\end{eqnarray}
The $A$ and $A^{\dagger}$ cancel each other on each plaquette as all the 6 virtual legs are contracted. We chose double-line tensor network for illustration but of course it can be done for any tensor network. So wave functions like  $|\Psi_{s_1,s_2}\rangle $ can be created by such string operators. Note that since the tensor network didn't change along the path, $|\Psi_{s_1,s_2}\rangle $ is still in the ground state along the path. So this string operator can only possibly create excitations at the ends, which is what we wanted. We will call such string operators \textit{gauge-string-operators} to distinguish them from the usual string operators on the physical level. Note that gauge-string-operators can only create stand-alone variations/excitations, and they are deformabale on the ground state subspace, like the physical string operators. We know that physical string-operators might not be deformed through a site that has an anyonic excitation present. Gauge-string operators also may not be be deformed through excitations. For example, there may be another operator $C$ present at the virtual legs such that $ACB\neq C$. But the interesting thing to note is that they can always be deformed through a stand-alone excitation. The reason for this is simple. A stand-alone tensor is surrounded by fixed point tensor $T^0$. So if we consider a Wilson-loop of gauge-string operator around it, $AB=I$ is still true, so $A$ and $B$ will simply cancel each other. So the Wilson-loop will simply disappear irrespective of what stand-alone excitation was there. So not only gauge-string operators create stand-alone excitations, they also always commute with the other stand-alone excitations. This suggests that all excitations in the stand-alone space have trivial mutual and self stastics. But to prove they are bosons, we need to do the topological spin calculation. Though again, using the same reasoning, it can readily be seen that the topological spin of stand-alone excitation is 1, as explained in Fig. \ref{bosoargu}. 
 \begin{figure}[t]
\begin{center}
\includegraphics[width=0.8\columnwidth]{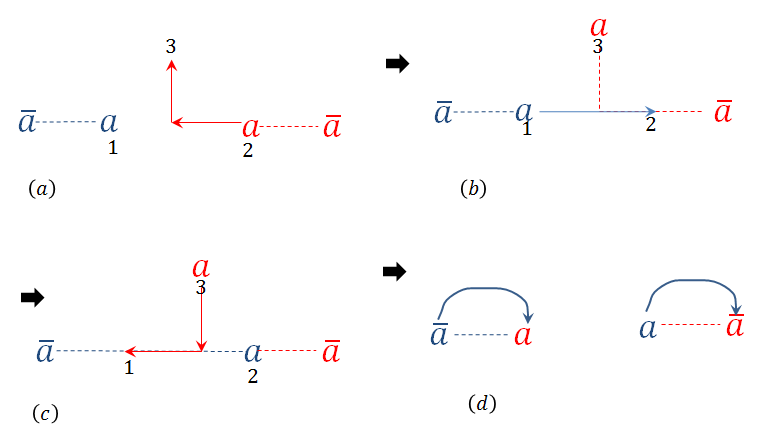}
\caption{Calculation of topological spin. We create two pairs (shown as red and blue) of particle, anti-particle pairs $a-\bar{a}$, with $a$ situated at site 1 and 2. We apply the following procedure in this order: (a)  Move first $a$ (red) from 2 to 3, (b) move second $a$ (blue) from 1 to 2,  (c) move first $a$ (red) from 3 to 1.  Finally, (d) we annihilate each $a$ with the anti-particles of the other anyon (i.e. red $a$ with blue $\bar{a}$ and \textit{vice versa}). When the propagation of $a$ happens through a gauge-string operator, which disappears along the path, this order of process becomes irrelevant, as the second string-operator does not interact with the first one, and the whole process is equivalent to creating and annihilating two pairs of $a-\bar{a}$, which has amplitude 1. It implies $a$ has a trivial topological spin.}
\label{bosoargu}
\end{center}
\end{figure}

\par 
We have determined that the variations in the stand-alone space are condensable bosons. So any such variation results in a wave function which is a condensate of the boson the variation corresponds to. But this alone does not necessarily cause a phase transition, because if the boson was topologically trivial, there should be not topological phase transition. Or, mathematically speaking, the stand-alone projector projects out variations that cannot proliferate, but it doesn't project out those stable variations that can proliferate. For example, the double-line stand-alone projector does not project out variation $X\otimes X$ though it is not unstable. So to find the unstable variations, we need an additional projector to project out  condensable but stable variations. We will argue now MPO-injective subspace is precisely this projector. 
\par 
To find out whether a virtual variation would cause the phase transition we need to first determine what this virtual variation corresponds to on the physical level. That is, we need to `lift' the variation from the virtual level to the physical level. When we do that we discover that there are two kinds of variations: The first kind is where a local virtual variation is lifted to a local physical variation, and the second kind is where the local virtual variation is lifted to a \textit{non-local} physical variation. We know that a local physical variation can only correspond to a topologically trivial boson since it can be removed by a local operation. This distinction between variations further decomposes the stand-alone into two subspaces: the MPO-injective space $\mathbb{M}$, which corresponds to the first kind of variation, and the unstable subspace $M_0-\mathbb{M}$, which corresponds to the second kind of variation. Let's first focus on the first kind of variations. 
\subsection{Physical understanding of MPO-injective subspace $\mathbb{M}$}
\label{subsec:PUM}
As claimed above, the physical significance of MPO-injective subspace is that the variations in this subspace are lifted to local physical variations, which have to be topologically trivial bosons since they can be removed by a local operation. Hence the physical significance of MPO-injective subspace is that  it contains all the topologically trivial excitations. \par 
To understand it better, let us look at concrete examples of variations that are lifted to local physical variation. Consider a $Z$ variation on the virtual leg of the fixed point single-line TNR. If we lift it to the physical level, what do we get? Since the virtual legs are just copies of  the physical legs, we get
\begin{eqnarray}\label{SLZVZP}
\includegraphics[scale=0.3]{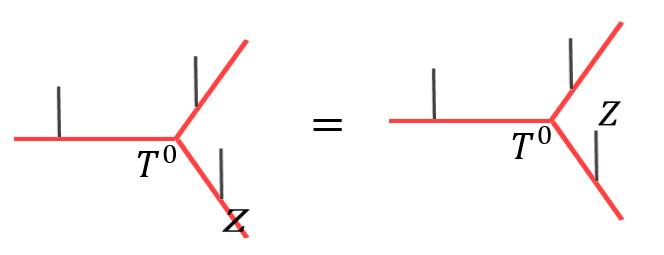}.
\end{eqnarray}
So $Z$ virtual variation is lifted to a $Z$ physical variation, which is local. According to our claim, it should be in the MPO-injective subspace. And indeed it is, since it respects the MPO symmetry of the single-line TNR, $Z^{\otimes 3}$. Also note that a $Z$ physical variation corresponds to a pair of $m$ particles sitting next to each other, not a single $m$ particle. It is a trivial excitation and can be removed by applying one $Z$ operation on the state, so it matches our claim. Contrast this with the $X$ variation on the virtual level. Can we find any local physical operator $O$ such that 
\begin{eqnarray}
\includegraphics[scale=0.3]{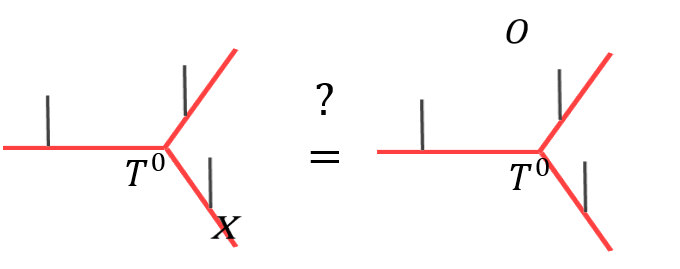}?
\end{eqnarray}
One can try and see that there is no such local operator $O$ for which this equation holds. (We will later show that $X$ can be lifted to the physical level, but it results in a non-local operator. ) \par 
A similar phenomena occurs in double-line TNR. The $X\otimes X$ variation can be lifted to a local physical operator,
\begin{eqnarray} \label{DLXVXP}
\includegraphics[scale=0.3]{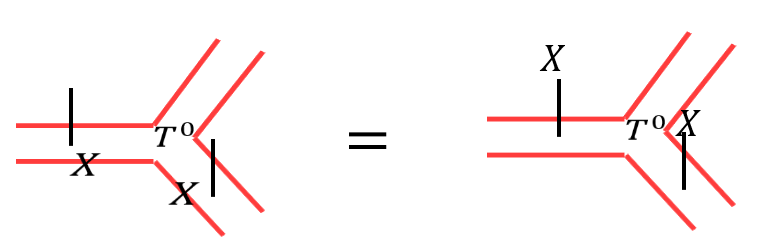},
\end{eqnarray}
but a $Z$ variation cannot be. (We will later show that $Z$ can be lifted to the physical level, but it results in a non-local operator.) It is again consistent with the claim as $X\otimes X$ variation respects the double-line MPO symmetry $(X^{\otimes 6})$ but $Z$ variation breaks it. Note that $X\otimes X$  variation on the physical level corresponds to a pair of $e$ particles sitting across a plaquette. It is a topologically trivial excitation and can be removed by an $X \otimes X$ operation on the state. So again, this matches our claim. \par 
Now we prove that these examples are no coincidence, and in fact any variation in the MPO-injective subspace is a local physical variation. 

Let us repeat the definition of the MPO-injective subspace here for convenience. We SVD decomps the fixed point RG tensor $T^0$ as a matrix between virtual and physical legs
\begin{eqnarray}
T^0 = \sum_j \lambda_j |v_j\rangle \langle p_j|,
\end{eqnarray}
where $\lambda_j$ are the singular values, and $v_j$ and $p_j$ are orthonormal vectors in the virtual and ground-state physical spaces respectively.   Then the MPO-injective subspace is the virtual subspace spanned by vectors $v_j$ such that corresponding singular values $\lambda \neq 0$. So the MPO projector is 
\begin{eqnarray}
\mathbb{M} = \sum_{j;\lambda_j\neq 0} |v_j\rangle\langle v_j|.
\end{eqnarray}
A mathematically inclined reader would note that MPO-injective subspace is nothing but the virtual subspace which is isomorphic to the image of the tensor as a map from virtual to ground state physical space. That is, if we restrict the domain of the tensor to this subspace, then tensor is an injective map from the virtual to the physical space, and a bijective map from MPO-injective subspace to ground-state physical subspace.  \textit{ Since these spaces are isomorphic, any operator in MPO-injective subspace can be mapped to an operator in the ground-state physical space and vice-versa, and this mapping would be bijective (one-to-one) as well}. Let us make it precise,
\begin{mylemm} \label{MPOinj}
If $A$ is any operator on the virtual space completely supported on $\mathbb{M}$ ($\mathbb{M}A=A\mathbb{M}=A $) then there exists an operator $B$ on the ground state physical space such that $AT^0=T^0B$ and vice-versa. That is, any variation in the subspace  $\mathbb{M}$ is equivalently a variation on the ground state physical space and vice-versa.
\end{mylemm}
 \begin{proof}
 Let us say an virtual operator $A$ is given which is completely supported on subspace $\mathbb{M}$. 
 Define pseudo-inverse of $T^0$ as 
 \begin{eqnarray}
 T^{0;+} = \sum_{j;\lambda_j\neq 0} \frac{1}{\lambda_j} |p_j\rangle\langle v_j|.
 \end{eqnarray}
 It is a pseudo-inverse since 
 \begin{eqnarray}
 T^0T^{0;+} &=& \mathbb{M}\\
 T^{0;+}T^0 &=& P_{GS},
 \end{eqnarray}
 where $P_{GS} \subset I_P $ denotes projector on the ground-state physical subspace of the full physical space. $I_P$ is the projector onto the full physical space. Now define a physical operator $B$ as 
 \begin{eqnarray}
 B= T^{0;+} A T^0, 
 \end{eqnarray}
 then
 \begin{eqnarray}
 T^0B =  T^0T^{0;+} A T^0 =\mathbb{M}A T^0= AT^0.
 \end{eqnarray}
 The last equality follows from the assumption that $A$ is completely supported on MPO-injective subspace. 
Similarly, given physical operator $B$ on the ground-state physical space, define $A=T^0BT^{0;+}$. So we have $AT^0=T^0BT^{0;+}T^0 = T^0BI_P=T^0B$. And of course these maps are injective. So $A$ and $B$ have a one-to-one correspondence. 
 \end{proof}
 With this lemma, we  see why in general variations in the MPO-injective subspace are trivial excitations. They are nothing but a local variation on the physical level, which is a local physical operator and can be removed by another local operator. In fact notice that if $A$ is unitary (within the space $\mathbb{M}$ ) then so is $B$ and vice-versa. Since all trivial excitations are obtained by local unitaries, or their linear combinations, we conclude that MPO-injective subspace should contain all trivial virtual excitations as well. \par 
 This completes the study of first kind of variations (those that are lifted to local physical variations) mentioned above. Now we study the second kind of variations.
 \subsection{Physical understanding of subspace $M_0-\mathbb{M}$ }
 \label{subsec:PUM0M}
 The physical significance of subspace $M_0-\mathbb{M}$ is that it contains the second kind of variations: the virtual variations that are lifted to a non-local physical operator. So these variations cannot be removed by a local physical operation on the state, and hence represent a topologically non-trivial excitation. And this excitation has to be a boson, as all excitations in stand-alone space are. So it means that the physical significance of $M_0-\mathbb{M}$ space is that it contains condensable excitations that are topologically non-trivial bosons, and that's why these variations cause a topological phase transition. \par 
 Let us first look at some concrete examples to understand this phenomena. We saw how a virtual $X$ variation on the single-line TNR couldn't be lifted to a local physical variation. But, the question is, can it be lifted to a non-local physical variation? The answer is, yes. To see it, first note that although a single $X$ variation cannot be lifted locally, two $X$ variations can be. That is, the fixed point single-line tensor satisfies
 \begin{eqnarray}\label{SLXXVP}
 \includegraphics[scale=0.3]{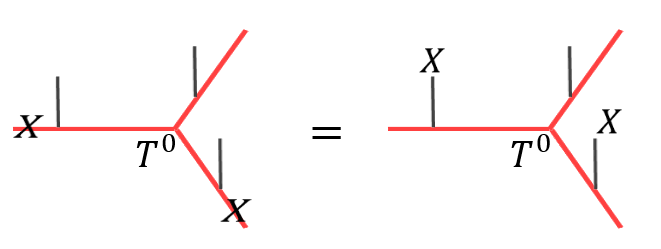}.
 \end{eqnarray}
 And also, we have the usual guage symmetry
 \begin{eqnarray}\label{SLinsertXX}
 \includegraphics[scale=0.4]{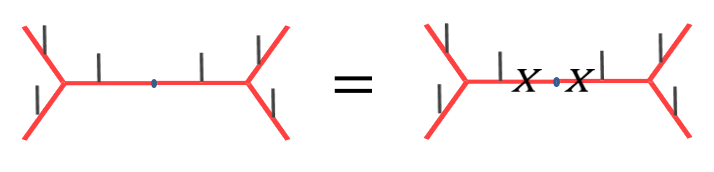}.
 \end{eqnarray}
 Using these two relations, we see that a single $X$ virtual variation on can be moved to another tensor on the same sublattice, and this transfer produces an $X\otimes X$ operation on the physical level, 
\begin{eqnarray}
 \includegraphics[width=0.8\columnwidth]{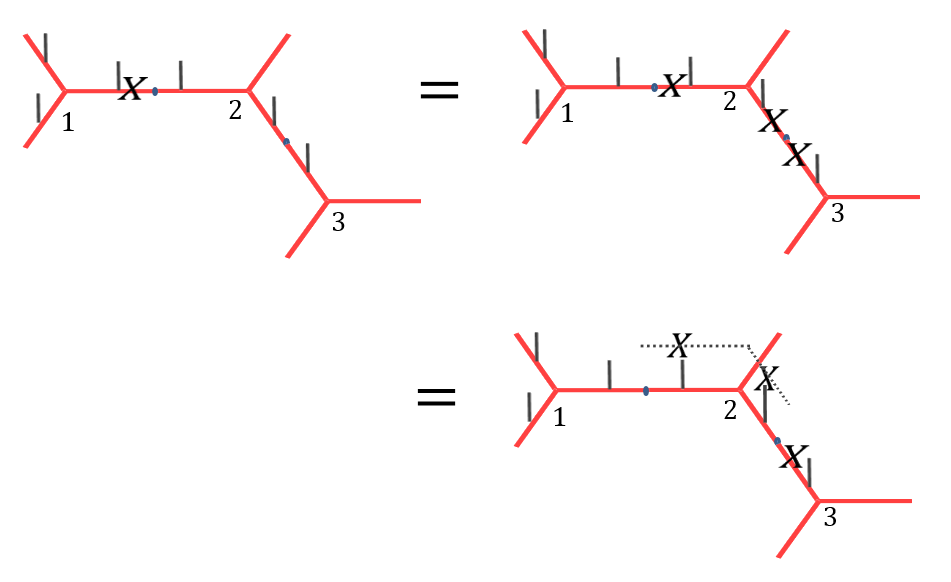}.
\end{eqnarray} 
In the first equality, Eq.~\eqref{SLinsertXX} is used while in the second equality Eq.~\eqref{SLXXVP} is used. We see that the $X$ variation moved from site 1 to site 3 while leaving operator $X\otimes X$ along the path (on site 2). We can repeat this process and move $X$ to the next tensor and so on. After $X$ is moved from site 1 to $n$ there will be an $X$-string operator applied on the physical level along the path. Finally, if there is already an $X$ variation present at site $n$, the two will cancel and we will be left with an $X$-string operator only,
\begin{eqnarray}\label{econd}
\includegraphics[width=\columnwidth]{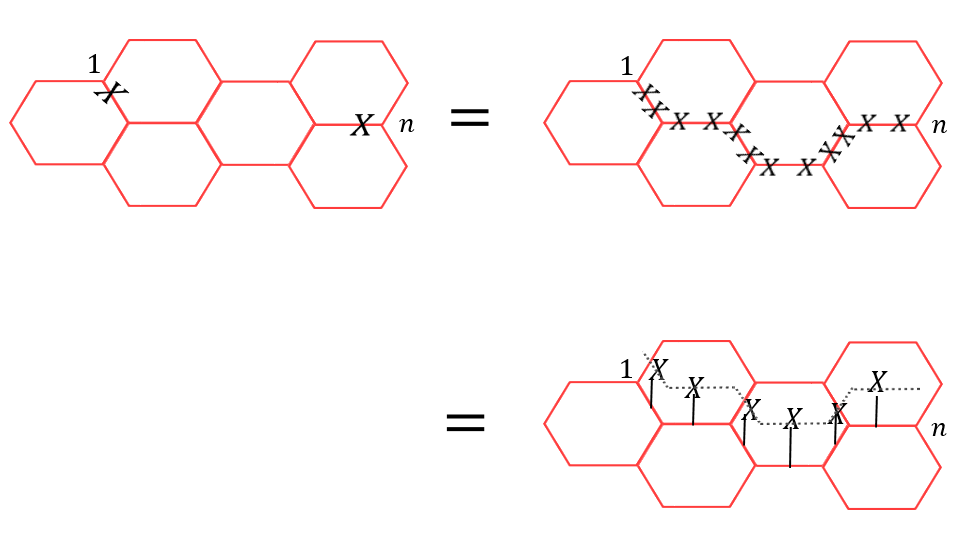}
\end{eqnarray}
Of course the particular path between site 1 and $n$ chosen is completely arbitrary. We can choose any path between them as we like. So we have successfully shown that though a single $X$ variation cannot be completely lifted to  the physical level, two such variations sitting far apart can be, and they are lifted to a non-local physical operator between them. It implies that a $X$ variation on the single-line TNR cannot be removed locally on a physical level. Only two of them can be removed by applying a non-local operator between them. In fact, it is easy to recognize what this excitation is. Since $X$-string operators correspond to creation or annihilation of $e$-particles in the toric code, it is clear that the $X$ virtual variation actually is an $e$ particle excitation. It is topologically non-trivial, which is in line with our claim. Condensation of $X$ variations is actually the condensation of $e$ particles, and that is why it leads to topologically phase transition.  \par 
A similar analysis can be carried out for the $Z$ variation in the double-line tensor. We noted that it cannot be lifted to a local physical operator. But two $Z$ operators can be lifted to a non-local physical operator,
\begin{eqnarray}\label{mcond}
\includegraphics[width=\columnwidth]{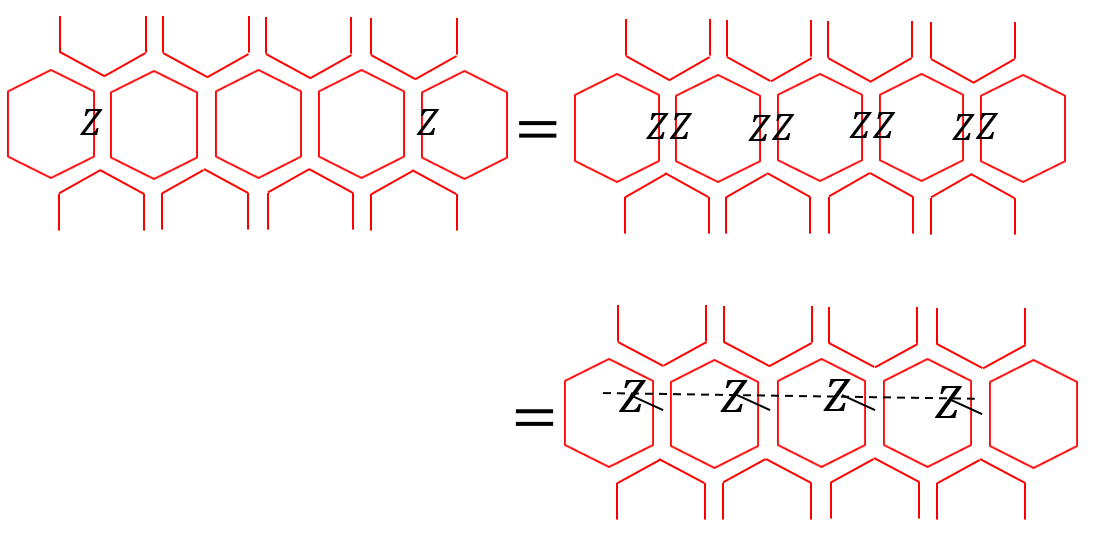}
\end{eqnarray}
where in the first equality, the following relation (similar to Eq.~\eqref{SLinsertXX}) has been used,
\begin{eqnarray}
\includegraphics[width=0.4\columnwidth]{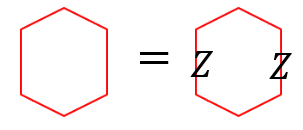}.
\end{eqnarray}
And in the second equality, the following property of the fixed point double-line tensor is used.
\begin{eqnarray}\label{DLZZVP}
\includegraphics[width=0.5\columnwidth]{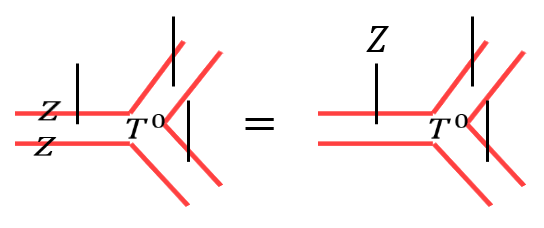}.
\end{eqnarray}
So we see that two $Z$ variations on the virtual level, sitting far apart cannot be removed by local operations on the physical level. They can only be removed by a non-local operator, the $Z$-string operator. This suggests that the $Z$ variation is a topologically non-trivial excitation. Indeed, it is easy to see that it is nothing but the $m$ particle excitation, since $Z$-string operator creates and annihilates $m$-particles. This is in line with all our claims: $Z$ variation is in the $M_0-\mathbb{M}$ space; it cannot be removed locally, and that it is a topological boson.  

 \subsection{Non-trivial gauge string-operators: Zero-string operators}
It may look a little puzzling that a local virtual operator in $M_0-\mathbb{M}$ can create a non-trivial physical excitation. We analyzed what these variations correspond to by lifting them up to the physical level. To understand the phenomena better we can ask the opposite question: what happens when we 'bring down' a non-trivial quasi-particle excitation on the physical level to the virtual level? Since such an excitation is created by a physical string-operator, the equivalent question is, what happens to string-operators of the model when we bring them down to the virtual level? We can look at the specific examples considered above. For example, if we look at Eq.~\eqref{econd} in the opposite way, we see that the physical $X$-string operator, which creates $e$-particles, becomes a gauge-string operator on the virtual level, which subsequently creates a variation in the $M_0-\mathbb{M}$ space. Contrast this with $Z$-string operator, that creates $m$-particles. This operator does not map to gauge-string operator,
\begin{eqnarray}
\includegraphics[scale=0.3]{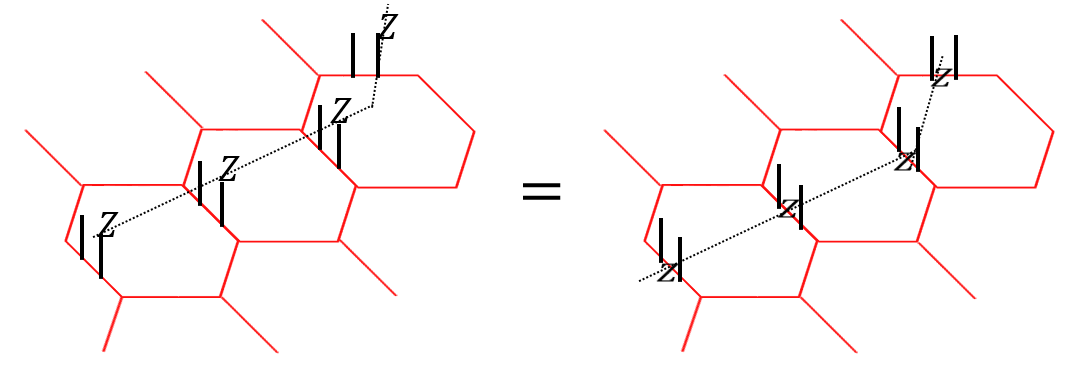}.
\end{eqnarray}
Similarly, Eq.~\eqref{mcond} shows that that the physical $Z$-string operator, which creates $m$-particles, becomes a gauge-string operator on the virtual level. So we see that if a physical anyonic string operator maps to a gauge-string operator on the virtual level, it creates an excitation in the  $M_0-\mathbb{M}$ space. This property of the tensor network state in general is the reason why a local virtual variation can actually correspond to non-local variation on the physical level. In other words, certain gauge-string operators are non-trivial because they come from a non-trivial string operator on the physical level. We would call such physical string operators that map  to gauge-string operator on the virtual level a \textit{zero-string operator}. The reason behind this terminology will become clearer in the next chapter. So we conclude that the $m$-particle operator is the zero-string operator of the double-line TNR, while the $e$-particle string operator is the zero-string operator of  the single-line TNR.\par 
Since variations in the unstable space $M_0-\mathbb{M}$ are created by zero-string operators, it implies that, if, in a given TNR, none of the physical string operators map to a gauge-string operator then it will have no variation in $M_0-\mathbb{M}$. In that case, we would simply have $M_0=\mathbb{M}$. Such a TNR will have no instabilities.  \par 
Before going to the physical explanation of instabilities, we would like to mention that there is one more way of decomposing the stand-alone space in trivial and non-trivial excitations: using Wilson-loops of anyonic string-operators. We can use these operators to detect whether a nontrivial excitation is sitting at a site, hence can potentially differentiate between $M_0-\mathbf{M}$ and $\mathbf{M}$. This has been explored in appendix \ref{MPO symmetry come from anyon} for readers who are interested in this perspective.   
\subsection{Physical reason of instability: topological boson condensation}

Now we put together the physical understanding of all the relevant subspaces ($M_0, \mathbb{M}$ and $M_0-\mathbb{M}$) to make the coherent picture of why variations in the subspace $M_0-\mathbb{M}$ are unstable, and, in particular, explain the numerical results shown in Fig.~\ref{fig:SLTCrand} and Fig.~\ref{DLTCrand}. The general explanation has already been stated in form of conjecture \ref{physical conjecture} but we repeat it again informally going through all possible variations one by one.
\begin{itemize}
    \item The variations on the physical indices are of course not stable because they are topologically trivial and can be removed with local operations.
     \item  Variations outside the stand-alone space, $I_V-M_0$ are not unstable because they cannot proliferate.
     \item Every variation inside $M_0$ is `bosonic' and does proliferate and the varied wave function is a condensate of that `boson'. But when the variations  is inside $\mathbb{M}$, it was a topologically trivial boson and hence does not cause a topological phase transition. Or, equivalently, every variation inside MPO-injective subspace was nothing but a variation on the physical level hence stable. 
     \item  Finally, when the variation was inside stand-alone, but outside MPO-injective subspace then it can condense and is a topologically non-trivial boson. Hence it causes a topological phase transition, resulting in a TNR instability. 
\end{itemize}
Now we explain this boson condensation specifically for the single-line and double-line TNR considering specific variation.
\subsubsection{$e$-particle condensation in single-line TNR}
To guide the discussion, consider two illustrative variations to single-line TNR as before 
\begin{eqnarray}
\includegraphics[scale=0.3]{SLTCpert}
\end{eqnarray}
\begin{itemize}
    \item The $T^Z$ variation exemplifies variations that can condense but correspond to a local physical variation, hence are  trivial/elementary excitations. Such variations result in a proliferation of elementary excitation which does not cause a topological phase transition.
    \item The $T^X$ variation exemplifies variations that can condense but do not correspond to local physical variations. In fact, such variations correspond to $e$-particle excitation. Hence such variation results in an $e$-particle condensation and destroys the topological order of the tensor network state.
    
\end{itemize}

\subsubsection{$m$-particle condensation in double-line TNR}
Let's consider the different variations in double-line TNR as before 
\begin{eqnarray}
\includegraphics[scale=0.3]{DLTCpert}
\end{eqnarray}
\begin{enumerate}
\item Variations in (a) and (c) exemplify variations that break the stand-alone symmetry $Z^{\otimes 2}$, hence they cannot proliferate and, therefore, are stable.
\item Variation in (b) exemplifies variations that can stand alone, so can proliferate. But they break the MPO symmetry, so correspond to a non-trivial boson.  In fact it corresponds to an $m$-particle excitation. So this variations causes $m$ particle condensation and results in the loss of topological order.
\item Finally, variation in (d) exemplifies variations that can proliferate. But they also are inside the MPO-injective subspace, so correspond to trivial/elementary excitations. So their proliferation does not cause a topological phase transition. 
\end{enumerate}

\section{Implications for the simulation of phase transitions}
Projected Entangled Pair States (PEPS), one type of Tensor Network States (TNS), are often used as \textit{ansatz} for different numerical simulations of gapped lattice topological models. In particular, TNS can be used to simulate phase transitions between different topological phases\cite{Gu08}. The fixed point Hamiltonian is perturbed with a local Hamiltonian $ H_0 \rightarrow H_0+ \eta H_{\text{local}} $ and  the perturbation strength, $\eta$ is increased slowly. At some finite value of $\eta$ the gap closes and the system goes through a phase transition. For many perturbations, this phase transition consists of boson condensation.  For example, for the toric code Hamiltonian Eq.~\eqref{TCHamiltonian}, two kinds of perturbations can be added 
\begin{eqnarray}
H_1= -U\sum_v \prod_{l\in v} Z_l - g\sum_p \prod_{l\in p} X_l -\eta \sum_l Z_l, \\ 
H_2 = -U\sum_v \prod_{l\in v} Z_l - g\sum_p \prod_{l\in p} X_l - \eta \sum_l X_l .
\end{eqnarray}
Let's first discuss the first kind of perturbation. In the first Hamiltonian, we keep $U=\infty $ and study the ground state as the relative values of $\eta $ and $g$ change.  At $\eta =0$ the ground state is simply the fixed point toric code state given in Eq.~\eqref{TCgs}. That is, it is an equal weight superposition of all closed string configuration. At $g=0$, the state is the vacuum state, that is, all spins are 0. These two states are topologically different, hence there must be a phase transition as we change $\eta/g$ from 0 to $\infty$. This phase transition can be understood as a condensation of $m$ particles. Recall that $\langle \Psi| B_p|\Psi\rangle=1$ corresponds to no $m$ particle and $\langle \Psi| B_p|\Psi\rangle =-1$ corresponds to an $m$ particle excitation at a plaquette $p$, where $B_p=  \prod_{l\in p} X_l $ is the plaquette term of the toric code Hamiltonian. For $\eta=0$ ground state we have $\langle \Psi| B_p|\Psi\rangle =1, \forall p$, while for $g=0$ ground state we have $\langle \Psi| B_p|\Psi \rangle =0, \forall p$. It indicates that as $\eta/g$ is increased, $m$ particles proliferate and at phase transition point, the system goes through a boson ($m$ particle) condensation and the ground state becomes a trivial state. Boson condensation phase transitions are known to be second order phase transitions. That is, ground state energy and its first order derivative  as a function of $\eta/g$ are  smooth functions, but its second order derivative is discontinuous at the phase transition point. \par 
It was shown by \citet{Gu08} that an attempt to numerically simulate  this phase transition point with single-line tensor network state ansatz gives a transition that is wrong both quantitatively and qualitatively. It gives a wrong critical point value of $\eta/g$, and it gives a first order phase transition, not a second order one. But with double-line tensor network state ansatz, it gives the correct second order phase transition with correct critical point. 
\par 
This difference can be easily understood in light of our discussion on single-line and double-line TNR of toric code state. As we showed, double-line TNR is capable of condensing $m$ particles while single-line TNR is not. That is why double-line TNR is suitable for simulating a phase transition that involves $m$ particle condensation. \par 
A similar analysis can be done for the second type of perturbation. We set $g=\infty$ and change relative value of $U$ and $\eta$. For $\eta=0$ the ground state is the toric code ground state in Eq.~\eqref{TCgs}, and for $U=0$ the state is trivial state with all qubits aligned in +x direction. Here the phase transition involves $e$ particle condensation which is again a second order phase transition. Hence, to simulate this phase transition, one should use single-line TNS ansatz and not the double-line TNS ansatz. \par 
This is one of the important point of understanding the unstable direction of variations that a particular TNR possesses. To simulate a boson condensation phase transition, one should choose the TNR that is capable of condensing that particular boson of the model. \par 
Of course, there is also a flip side to this. If one is interested in determining the topological order of a particular TNR by calculating the topological entanglement entropy, one should make sure to keep out of the unstable space, $M_0-\mathbb{M}$, for numerical stability. A small numerical variation in this space will change the state globally and result in wrong results. For example, in calculations involving Tensor Entanglement Renormalization Group (TERG)~\cite{Gu08} and Tensor Network Renormalization (TNR~\footnote{not to be confused with TNR that we use for referring to tensor network representation.})~\cite{EvenblyVidalTNR} steps, we should project the resulting tensor after every RG step back to the stable space, ($I_V-(M_0-\mathbb{M})$), or naturally occurring  numerical errors might gain a component in $M_0-\mathbb{M}$ space and change the topological order of the state radically. \par
Now we will apply what we learned from the toric code example to analyze the TNR of another closely related model, the double semion model.
\section{Double-semion}
\label{ds}
Double-semion model can be understood as a `twisted' $Z_2$ quantum double model\cite{Freedman04,Levin06}. Its Hamiltonian is almost the same as that of toric code, except for the phase factor associated to the plaquette term
\begin{eqnarray}
H_0= -\sum_v \prod_{l\in v} Z_l - \sum_p \prod_{l\in p} X_l \prod_{r\in \textrm{legs of } p} i^{\frac{1-Z_r}{2}}, 
\end{eqnarray}
where `legs of $p$' refers to the six legs attached to a plaquette. Its ground state is 
\begin{align} \label{DSgs}
\Ket{\psi} = \sum_{X \in \text{closed}} (-1)^{n(X)}\Ket{X},
\end{align}
where $X$ again refers to string configurations on the hexagonal lattice. $n(X)$ denotes the number of loops in a given string configuration. The ground state, like that of toric code, is again a superposition of all closed string configurations. But it has a phase factor of $(-1)^{n(X)}$ which is $1$ for even number of loops and $-1$ for odd number of loops. It has 3 quasi-particle excitations: a semion, an anti-semion, and a self-boson. So, unlike the toric code, it has only one boson. There is a known double-line TNR of this state \cite{Gu09,Buerschaper09}, $(T^0)^{ijk}_{\alpha\alpha';\beta\beta';\gamma\gamma'}$, with the same structure  as that of toric code. So,
\begin{eqnarray}
(T^0)^{ijk}_{\alpha\alpha';\beta\beta';\gamma\gamma'} = S_{\alpha\beta\gamma} \delta_{\alpha\alpha'}\delta_{\beta\beta'}\delta_{\gamma\gamma'}\delta_{i,\beta+\gamma}\delta_{j,\alpha+\gamma}\delta_{k,\alpha+\beta }. 
\end{eqnarray}
But now the values are 
\begin{eqnarray} \label{DSTNReq}
S_{\alpha\beta\gamma} = 
\begin{cases} 1 & \text{ if } \alpha+\beta+\gamma=0,3  \\
i & \text{ if }   \alpha+\beta+\gamma=1 \\
-i & \text{ if }  \alpha+\beta+\gamma=2 .
 \end{cases}
\end{eqnarray}
Clearly, it has the same $Z^{\otimes 2}$ symmetry, as the toric code double-line TNR. That is, $T^0$ satisfies
\begin{eqnarray}\label{DLDSZsym}
\includegraphics[scale=0.3]{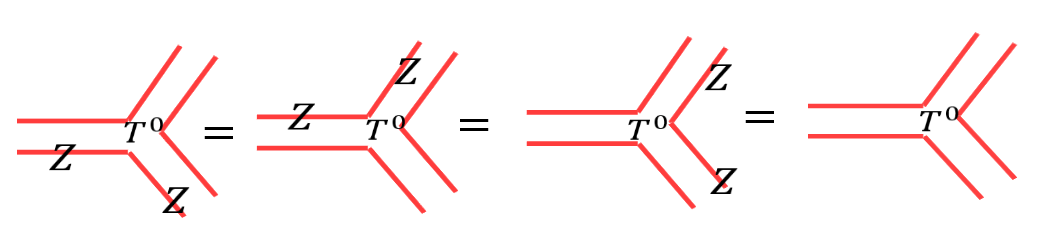}
\end{eqnarray}
But it does not have the exact $X^{\otimes 6} $ symmetry as that of toric code double-line TNR. By looking at the tensor values, it can be seen that it has the $X^{\otimes 6} $ with an additional phase factor $\omega$ between virtual legs,
\begin{eqnarray}\label{DLDSXsym}
\includegraphics[scale=0.3]{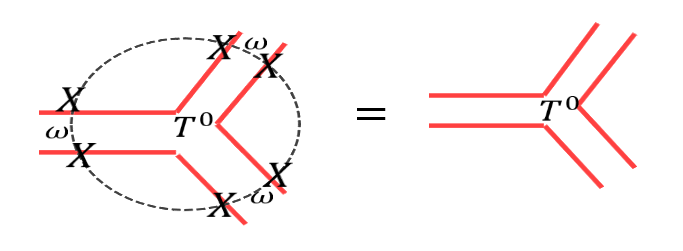}
\end{eqnarray}
where $\omega=i$ if the virtual legs on the two sides of it take different values (that is, there is a domain wall) and $\omega = 1$ otherwise. So $T^0$ has   $Z_2$ symmetry of the form  $(i)^{n(d)} X^{\otimes 6}$ where $n(d)$ is the number of domain walls between $\alpha, \beta$ and $\gamma$. That is, 
\begin{align} \label{nd}
n(d) = \begin{cases} 0 & \text{ if } \alpha+\beta+\gamma=0,3 \\ 
2 & \text{ if } \alpha+\beta+\gamma = 1,2 . \end{cases}
\end{align}
To apply our conjecture we first need to calculate the stand-alone and MPO-injective subspaces of $T^0$. The double tensor is, 
\begin{eqnarray}
\mathbb{T}^0&=& (I^{\otimes 2}+Z^{\otimes 2})^{\otimes 3}(I^{\otimes }+i^{n(d)}X^{\otimes 6}).
\end{eqnarray}
Comparing it with the toric code double tensor in Eq.~\eqref{DoubleTensorDL}, we can immediately guess that the stand-alone projector is given by
\begin{eqnarray}
M_0=\frac{1}{8}\left(I^{\otimes 2}+Z^{\otimes 2}\right)^{\otimes 3},
\end{eqnarray}
which is the same as that of toric code. And the MPO projector is,
\begin{eqnarray}
\mathbb{M}&=& \frac{1}{16} (I^{\otimes 2}+Z^{\otimes 2})^{\otimes 3}(I^{\otimes }+i^n(d)X^{\otimes 6}).
\end{eqnarray}
So we see that the symmetries identified in Eq.~\eqref{DLDSZsym} are actually the stand-alone symmetries, and the symmetry identified in Eq.~\eqref{DLDSXsym} is actually the MPO symmetry. With this information, our mathematical conjecture predicts:
\begin{enumerate}
\item If the variations breaks the $Z\otimes Z$ symmetries, then it is stable.
\item If the variations respects all $Z\otimes Z$ symmetries then there are two subscases 
\begin{enumerate}
    \item If it also respects the $ (i)^{n(d)}X^{\otimes 6} $ symmetry then it is stable.
    \item If it breaks the $ (i)^{n(d)}X^{\otimes 6} $ symmetry then it is stable.
\end{enumerate}
\end{enumerate}
\begin{figure}
\begin{center} 
\includegraphics[trim = 10mm 0mm 10mm 10mm,width=8cm]{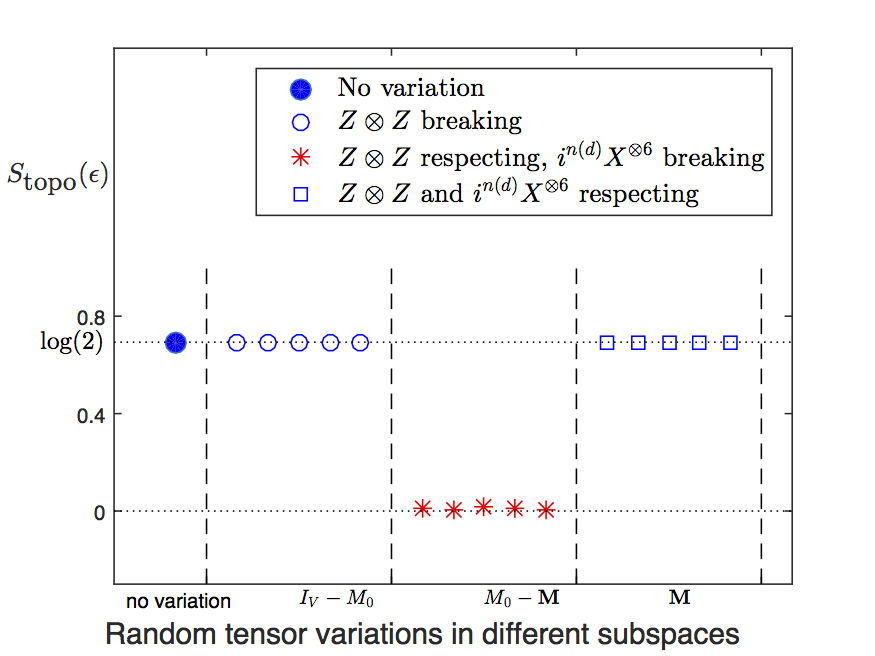}
\caption{Numerical calculation of the topological entanglement entropy $S_{\text{topo}} (\epsilon)$ of the states represented by double-semion fixed point double-line tensors, $T^0$, varied with an infinitesimal random tensor in different subspaces. $\epsilon$ value is kept fixed at $\epsilon=0.01$. Blue dot corresponds to $S_{\text{topo}}$ with no variation. $I_V$ is projector onto the full virtual space. $M_0$ is the projector on the stand-alone subspace. $\mathbb{M}$ is the MPO-injective subspace projector.  We take a random tensor and apply the projectors to generate random tensors in respective subspaces. Variations in $I_V-M_0$ violate $Z \otimes Z$ symmetry. Variations in $M_0-\mathbb{M}$ violate $i^{n(d)}X^{\otimes 6}$ but not $Z\otimes Z$. Variations in $\mathbb{M}$ violate no virtual symmetry.\label{DLDSrand}. The details of this numerical calculations are given in appendix A.2.}
\end{center} 
\end{figure}
We test these predictions numerically. The results are shown in Fig.~\ref{DLDSrand}. We conclude that the conjecture predicts the numerical observation correctly. \par 
What about the physical conjecture?  Is it compatible with the numerical observation? The answer is, yes. the double-semion model has one boson whose string operator is the $Z$-string operator, the same as that of the $m$-particle in toric code. Since both also have the same stand-alone space, it means bringing down this string operator to the virtual level would again give us a gauge-string operator. Hence the string-operator corresponding to the boson in the double-semion model is a zero-strong operator, which implies that the variations in the stand-alone space corresponds to this boson. So the instability we see is due to the condensation of this topological boson. Another way to see it is to notice that the MPO symmetry in Eq.~\eqref{DLDSXsym} actually comes from the Wilson loop operator corresponding to semion (or anti-semion). So variations that break it actually signify the presence of the boson.  \par 
\subsection*{Comparing double-line TNR of toric code and double-semion}
So we see that the space  $(M_0 -\mathbb{M})$ for double-semion is 4-dimensional space spanned by basis
\begin{eqnarray} \label{DLDSspan}
\includegraphics[width=0.9\columnwidth]{DLspan}
\end{eqnarray}
This looks exactly similar to the $M_0-\mathbb{M}$ basis in double-line toric code in Eq.~\eqref{DLTCspan}, which might lead one to believe that they both are unstable for similar variations. But one has to carefully note that the tensor $T^0$ for both models are different, so the basis shown in Eq.~\eqref{DLTCspan} and in Eq.~\eqref{DLDSspan} are actually different. To illustrate this consider the following variation:
\begin{eqnarray} 
\includegraphics[scale=0.4]{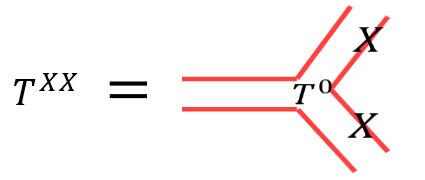}
\end{eqnarray}
This variation is in stand-alone space of both toric code and double-semion, but it respects MPO symmetry of the toric code but violates the MPO symmetry of the double-semion. Indeed this variation causes phase transition in double-semion but not toric code. This variation cannot be lifted to the physical level on the double-semion tensor like it did for toric code (Eq.~\eqref{DLXVXP}). But one can readily see that this variation is not spanned by the  basis in Eq.~\eqref{DLDSspan}. And the reason for this is that this variation has components both in the MPO-injective subspace \textit{and} the unstable subspace, which can be seen by applying the projector of the two spaces. We find that $(M_0-\mathbb{M})T^{XX}\neq 0 $ and $\mathbb{M}T^{XX}\neq 0$. This examples reminds us that for a variation to be unstable, all it needs is to have a non-zero component in the unstable space. So it should not be thought that only variations spanned by the unstable basis are unstable. 

\section{General String-net Models }
\label{sec: sn}
The models discussed so far, the toric code model and the double semion model, are particular examples of a general class of 2D topological models known as \textit{string-net models} \cite{Levin06}. Also, the TNR discussed so far (single-line and double-line) are  reduced versions of the a general \textit{triple-line }TNR of the string-net states\cite{Gu09,Buerschaper09}. \par 
A string-net construction defines a topological model on a honeycomb lattice for any arbitrary \textit{unitary tensor fusion category} \cite{Kitaev03,Levin06}. The local Hilbert space has spins sitting on the edges. These spins can take $i=0,1,..,N-1$ values called \textit{string-types}. $i=0$ corresponds to the vacuum state. In general strings have orientation and for each string-type $i$ we have a unique strng-type $i*$ with opposite orientation. If $i^*=i, $ then the string is called `unoriented'.\textit{ In this paper we would assume all strings are unoriented for simplicity but we believe our results are easily generalizable to oriented links.} A branching rule $\delta_{i,j,k}$ defines what string-types are allowed to meet at a vertex. An $F$-symbol guides how the strings fuse with each other. The $F$-symbol comes from the unitary tensor category data and satisfies the so called \textit{pentagon equations}. A local commuting Hamiltonian is defined, $H = -\sum_v A_v - \sum_p B_p$, where $v$ and $p$ denote the vertices and plaquettes of the honeycomb lattice. The vertex term projects onto the space allowed by the branching rule. The plaquette term acts by creating loops of $s$-type strings which subsequently fuse with the existing string. As for any local commuting Hamiltonian, the ground state can be obtained by applying the projector $P_{gs}=\left(\sum_p B_p)\right) \left( \sum_v A_v \right) $ on the vacuum state. A brief  review of the string-net models has been given in the Appendix \ref{string net review}. Readers can refer to the original papers for more details on the subject \cite{Levin06,Gu09,Buerschaper09,csahinouglu2014characterizing}. \par 
\subsection{Triple-line TNR of RG fixed point string-net state}
\label{triple-line TNR}

As shown by \citet{Gu09,Buerschaper09}, RG fixed point string-net states described above are known to have a triple-line TNR.
We will only briefly discuss the relevant details here. A short derivation of the triple-line TNR is given in the Appendix \ref{string net TNR derivation}. An interested reader may refer to the original papers \cite{Levin05,Gu09,Buerschaper09} for more details. \par 
A general triple-line Tensor is represented diagrammatically as:
\begin{equation}
    \includegraphics[scale=0.4]{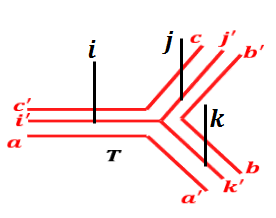}.
\end{equation}
A string-net fixed point state has a triple line TNR with components given by
\begin{eqnarray}\label{eq:snT0}
T^{ijk}_{aa';bb';cc';i'j'k'} &= & S^{ijk}_{abc} \delta_{a,a'}\delta_{b,b'} \delta_{c,c'}\delta_{i,i'}\delta_{j,j'}\delta_{k,k'} \nonumber \\
\text{where }
S^{ijk}_{abc} &=& \sqrt[4]{d_{i}d_{j}d_{k}}G^{ijk}_{abc} \sqrt[6]{d_{a}d_{b}d_{v}}.
\end{eqnarray} 
So it would be represented diagrammatically as:
\begin{equation}\label{snT0}
    \includegraphics[scale=0.4]{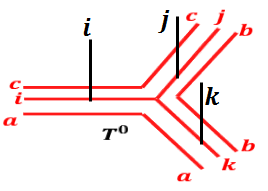}.
\end{equation}

Before we discuss the properties of the triple-line TNR of the general string-net models, we would like to mention that double-line TNR and single-line TNR are actually reduced versions of the triple-line TNR, and as such, many results about the triple-line TNR apply to double-line and single-line as well.  We can discard some of the legs of the triple-line tensor if fewer legs are required to encode the necessary information. For example, for abelian models, the middle leg of the triple-line tensor is redundant; it always assumes value which is a product (fusion) of the two legs on either side of it. That's why for abelian models, double-line tensors suffice and the middle-leg can be discarded.  Non-abelian models, such as the double-Fibonacci model we will study in section \ref{sec: fb}, the middle-leg does carry essential information and cannot discarded. So one cannot have a double-line TNR of non-abelian models.  Furthermore, if the ground state of a model can be written as an \textit{equal} superposition of states allowed by branching rules then the ground state admits a single-line TNR. For example, toric code ground state is an equal superposition of all closes string configurations, and hence admits a single-line TNR.  In fact, \textit{any} quantum double model with an abelian gauge group  can have a single-line TNR. The double-semion model, on the other hand, is not an equal superposition of states allowed by the branching rules (it has a phase factor $i^{n(d)}$), and hence it cannot admit a single-line TNR. \par 

To apply the conjecture to the triple-line TNR of general string-net model, we now first calculate its MPO-injective and stand-alone subspaces. We will do so in the next two sections.

\subsection{Stand-alone space of triple-line TNR string-net}
\label{stand-alone triple-line section}
We find that the stand-alone space of the triple-line TNR is given by the following theorem:
\begin{mythm} \label{thm:string net M0}
The stand alone space, $M_0$, of the triple-line string net TNR is spanned by the orthonormal vectors
\begin{eqnarray} 
  \delta_{i,b,c}\delta_{j,c,a}\delta_{k,a,b}|a,b,c;i,j,k\rangle
\end{eqnarray}
where $\delta$ is the branching rule of the string-net model. $i,j,k$ label the middle legs, and $a,b,c$ label the plaquette legs.
\end{mythm}
The proof of this result is rather involved and is given in Appendix \ref{stand-alone triple-line section app}.
These basis vectors can be represented as string-configurations,
\begin{eqnarray}
\begin{centering}
| \delta_{i,b,c}\delta_{j,c,a}\delta_{k,a,b}\rangle =\raisebox{-10mm}{\includegraphics[scale=0.4]{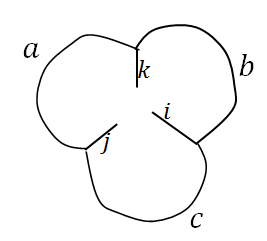}}
\end{centering}.
\label{triple-line stand-alone basis fig}
\end{eqnarray}
So we get,
\begin{eqnarray}
dim(M_0)=\sum_{\substack{a,b,c;i,j,k }} \delta_{i,b,c}\delta_{j,c,a}\delta_{k,a,b}.
\label{triple-line M0 dim }
\end{eqnarray}
\subsection{ MPO-injective  subspace of String-net triple-line TNR }
\label{triple line MPO subspace}
we will use definition \ref{MPOdef2} to find the MPO-injective subspace of triple-line TNR. 
Using the triple-line TNR $T^0$ of string-net states given in Eq.~\eqref{eq:snT0}, the virtual density matrix is found to be
\begin{eqnarray}\label{string net sigma}
\sigma &=& \sum_I(T^0)^I_{\alpha }(T^{0;*})^I_{\alpha }\nonumber \\ 
&=&\sum_{\lbrace a_k,b_k;i_{k,k+1} \rbrace} G^{i_{23}i_{31}i_{12}}_{a_1a_2a_3}G^{i_{23}i_{31}i_{12}}_{b_1b_2b_3}\prod_j(d^{\frac{1}{6}}_{a_j}d^{\frac{1}{6}}_{b_j}d^{\frac{1}{2}}_{i_{j,j+1}} ) \nonumber \\ & &  | \lbrace a_k;i_{k,k+1} \rbrace\rangle \langle \lbrace b_k;i_{k,k+1} \rbrace |.
\end{eqnarray}
Clearly, this density matrix can simply be written as 
\begin{eqnarray}
\sigma&=&\sum_{i,j,k} |v_{i,j,k}\rangle \langle v_{i,j,k}| \\
\text{where} |v_{i,j,k}\rangle &=& (d_id_jd_k)^{\frac{1}{4}} \sum_{a_1,a_2,a_3} G^{i,j,k}_{a_1,a_2,a_3}(d_{a_1}d_{a_2}d_{a_3})^{\frac{1}{6}} \nonumber \\ & & 
|a_1,a_2,a_3;i,j,k\rangle. \label{vijk}
\end{eqnarray}
So $\sigma$ has a diagonal form in terms of vectors $v_{i,j,k}$. To get the projector on to its support space,we simply need to use the unit vectos $\frac{1}{N_{i,j,k}} |v_{i,j,k}\rangle $, where $N_{i,j,k}=\sqrt{\langle v_{i,j,k}|v_{i,j,k}}\rangle $ is the norm of vector $v_{i,j,k}$. So the string-net MPO projector is 
\begin{eqnarray}
\mathbb{M} = \sum_{i,j,k} \frac{1}{N^2_{i,j,k}} |v_{i,j,k}\rangle\langle v_{i,j,k}|.
\end{eqnarray}
Note that $|v_{i,j,k}\rangle=0$ if $\delta_{i,j,k}=0$. It means that $\mathbb{M}$ projects on to the physical states allowed by the branching rules, and
\begin{eqnarray}\label{triple-line M dim}
dim(\mathbb{M})=\sum_{i,j,k}\delta_{i,j,k}
\end{eqnarray}
Comparing Eq.~\eqref{triple-line M0 dim } with Eq.~\eqref{triple-line M dim}, we can see that $dim(M_0) > dim(\mathbb{M})$. So according  to our conjecture,\textit{ there must always be unstable directions of variations in the triple-line TNR of any string-net model!} Indeed, we give examples of such unstable directions and will prove in section \ref{sninstability proof} that triple-line TNR of string-net always have instabilities. 

\subsection{Tensors in the unstable space $M_0-\mathbb{M}$}
\label{tensors in unstable space}
We have determined both the stand-alone space, $M_0$ and the MPO-injective space $\mathbb{M}$. $M_0$ space is spanned by vectors,
\begin{eqnarray}
   \delta_{i,b,c}\delta_{j,c,a}\delta_{k,a,b}|a,b,c;i,j,k\rangle.
\end{eqnarray}
And the MPO-injective space $\mathbb{M}$ space is spanned by $|v_{i,j,k}\rangle: \delta_{i,j,k}=1$, where
\begin{eqnarray}
|v_{i,j,k}\rangle&=&\sum_{a,b,c} G^{i,j,k}_{a,b,c}(d_{a}d_{b}d_{c})^{\frac{1}{6}}|a,b,c;i,j,k\rangle.
\end{eqnarray}
The tensors supported on $M_0-M$ are precisely the tensors that cause instability. To determine the orthogonal basis of this space we simply need to find vectors orthogonal to $v_{i,j,k}$ which are within the stand-alone space. First note that $M_0$ space decomposes in orthogonal subspaces $M_0= \bigoplus_{i,j,k} \mathbb{V}_{i,j,k}$ where the subspace $\mathbb{V}_{i,j,k}$ is spanned by  $\delta_{i,b,c}\delta_{j,c,a}\delta_{k,a,b}|a,b,c;i,j,k\rangle $, that is, $a,b,c$ for which  $ \delta_{i,b,c}\delta_{j,c,a}\delta_{k,a,b}$ is non-zero.   $M_0-\mathbb{M}$ space can be decomposed into two subspaces, 
\begin{enumerate}
    \item $\delta_{i,j,k} =0$: This consists of all the string-configurations in Fig.~\ref{triple-line stand-alone basis fig} for which $\delta_{i,j,k} =0$. They are obviously orthogonal to all $v_{i,j,k}$ since $v_{i,j,k}=0$ if $\delta_{i,j,k}=0$. Since these vectors violate the vertex term of the Hamiltonian we will refer to them as `vertex variations'.
    \item  $\delta_{i,j,k} =1$: This is the subspace spanned by string configurations for which $\delta_{i,j,k} =1$. We need to find other vectors in  $\mathbb{V}_{i,j,k} $ that are orthogonal to $v_{i,j,k}$. $\text{dim}(\mathbb{V}_{i,j,k})=\sum_{a,b,c} \delta_{i,b,c}\delta_{j,c,a}\delta_{k,a,b} = \sum_{a,b,c}  [G^{i,j,k}_{c,a,b}] $ where  $ [G^{i,j,k}_{c,a,b}]=1$ if $G^{i,j,k}_{c,a,b} \neq 0$ and $0$ otherwise. Note that since $\mathbb{V}_{i,j,k}$ are orthogonal for different values of $i,j,k$, we just need to find vectors in individual $\mathbb{V}$ subspaces. To find these we will use the orthogonality of $G$ \eqref{orthogonal}
    \begin{eqnarray}
\sum_c G^{i,j,k}_{a,b,c}G^{i,j,k}_{a,b,c}d_c=\frac{1}{d_k}\delta_{a,b,k}.
\end{eqnarray}
And the fact that matrices $N^k$ defined by $N^k_{a,b} = \delta_{a,b,k}$ can be simultaneously diagonalized $\forall k$. Let's say $|s_q\rangle= s_{q;a}|a\rangle $ is its qth such simultaneous  eigenvector. As discussed in Appendix \ref{algebra}, $s_{0;a}=d_a$, that is, the vector formed by quantum dimensions is an eigenvector to $N^k$. These vectors are orthogonal, $\langle s_q|s_{q'}\rangle =\delta_{q,q'}$, which also implies that $\langle s_q|N^k|s_{q'}\rangle =\sum_{a,b}s_{q;a}\delta_{k,a,b}s_{q';b} \propto \delta_{q,q'} $. Now we are ready to write down the vectors spanning $\mathbb{V}_{i,j,k}$.
\end{enumerate}

Consider vectors 
\begin{eqnarray}
|v^{q;a}_{i,j,k}\rangle=\sum_{a,b,c}\frac{s_{q;a}}{d_a} G^{i,j,k}_{a,b,c}(d_{a}d_{b}d_c)^{\frac{5}{6}}|a,b,c;i,j,k\rangle,
\end{eqnarray}
where superscript $(q;a)$ indicates that the $q$th eigenvector is used on leg $a$. Using the orthogonality relation, we get 
\begin{eqnarray}
\langle v_{i,j,k}|v^{q;a}_{i,j,k}\rangle &=&\sum_{a,b,c}s_{q;a}d_bd_c G^{i,j,k}_{a,b,c} G^{i,j,k}_{a,b,c}\\
&=&\sum_{a,b} s_{q;a}\delta_{a,b,k}d_{b}\\
&=&\sum_{a,b} s_{q;a}\delta_{a,b,k}s_{0;b}\\
&\propto & \delta_{q,0}.
\end{eqnarray}
So we see that the vector $v^{q;a}_{i,j,k}$ is orthogonal to $v_{i,j,k}$ if $q\neq 0$. Since $q$ takes $N-1$ non-zero values and it can be put on leg $a,b$ or $c$ we seem to have $3(N-1)$ such vectors. However not all of them will be independent, but they span the full vector space $\mathbb{V}_{i,j,k}$. Since these kinds of variations change the plaquette leg factors, hence violating the plaquette term, we will refer to these variations as `plaquette  variations'. \par 
\subsection{Instability of triple-line TNR}
\begin{mythm}\label{thm:sninstability}
Let $T^0$ be the fixed point triple-line TNR of a string net ground state. There exist tensors $T^q$ in the space $M_0-\mathbb{M} $ (that is,$ (M_0-\mathbb{M})T^q \neq 0$) that for the variation $T^0\rightarrow T^0+\epsilon T^q  $, $ \lim_{\epsilon\rightarrow 0 }S_{\text{topo}}(\epsilon)  \neq S_{\text{topo}}(0)$.
\end{mythm}
The proof of this theorem is rather involved and has been included in the appendx \ref{sninstability proof}. \par 
With this we have concluded the analysis of general string-net mdodels and their triple-line TNR. Now we turn to some concrete examples to understand how the conjecture in \ref{TNR instability conjecture} explains the instabilities in string-nets. 

\subsection{Examples: Triple-line TNR of the toric code and double semion states}
\label{triple-line of TC and DS}
\begin{figure}
\begin{center}
\includegraphics[width=9cm]{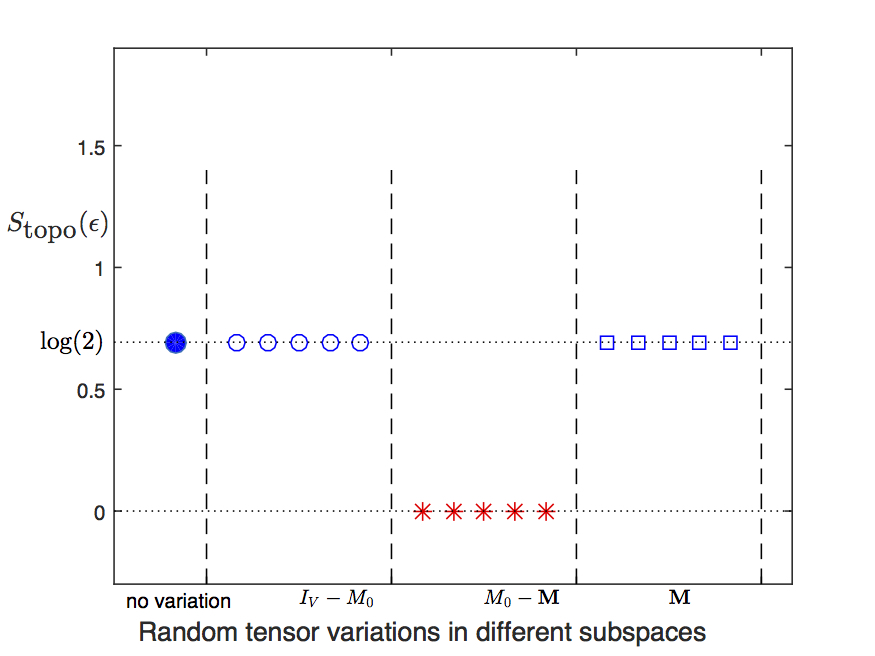}
\caption{Numerical calculation of topological entanglement entropy $S_{\text{topo}} (\epsilon)$ of states represented by toric code fixed point triple-line tensors, $T^0$, varied with an infinitesimal random tensor in different subspaces. $\epsilon$ value is kept fixed at $\epsilon=0.1$. Blue dot corresponds to $S_{\text{topo}}$ with no variation. $I_V$ is projector onto the full virtual space. $M_0$ is the projector on the stand-alone subspace. $\mathbb{M}$ is the MPO-injective subspace projector.  We take a random tensor and apply the projectors to generate random tensors in respective subspaces. Details of this numerical calculation are given in Appendix A.2.}
\label{TLTCrand}
\end{center}
\end{figure}
let's first examine how models covered in the previous chapter, toric code and double semion, fit into triple-line TNR. One can get the triple-line TNR for them by plugging in the relevant string-net data into Eq.~\eqref{eq:snT0}. We will apply the results about general-string net models developed in previous sections to the two cases. \par 
Toric code string-net data is, 
\begin{eqnarray}
N=1, \,d_0=1, \, d_1=1\nonumber; \\
 \delta_{000}=\delta_{110}=\delta_{101}=\delta_{110}=1; \nonumber \\  
G^{000}_{000}=G^{000}_{111}=1; \nonumber \\ G^{011}_{011}=G^{011}_{100}=G^{101}_{101}=G^{101}_{010}=G^{110}_{110}=G^{110}_{001}=1. \nonumber\\
\end{eqnarray}
The triple-line TNR of toric code can be built by plugging in this data into the general expression in Eq.~\eqref{eq:snT0}. This tensor has 9 virtual indices, each of which takes 2 values. So the full virtual space is $\textrm{rank}(I_V)=2^9=512$ dimensional. The dimension of the stand-alone space is 
\begin{eqnarray}
\textrm{rank}(M_0) &=& \sum_{a,b,c;i,j,k}\delta_{i,b,c}\delta_{j,c,a}\delta_{k,a,b} =8, 
\end{eqnarray}
and the dimension of the MPO-injective subspace is 
\begin{eqnarray}
\textrm{rank}(\mathbb{M} ) &=& \sum_{i,j,k}\delta_{i,j,k}=4.
\end{eqnarray}
These imply that $\textrm{rank}(I_V-M_0)=512-8=504$ and $\textrm{rank}(M_0 -\mathbb{M})=8-4=4$. So we reach the conclusion that out of 512 possible variations, 504 are stable since they are outside the stand-alone space. In the remaining 8 dimensional subspace, perturbations in a 4 dimensional subspace are in stable whereas the ones in the other 4 dimensional subspace are unstable. Using Fig.~\ref{venndiag}, the classification of all variations can be represented as follows
\begin{eqnarray}\label{TLTC venn}
\includegraphics[scale=0.3]{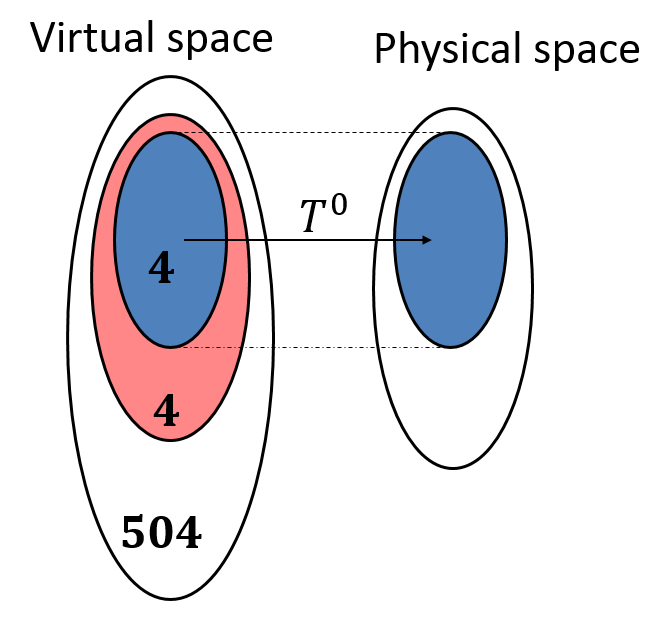}
\end{eqnarray}
The numerical calculation supporting this conclusion is shown in Fig.~\ref{TLTCrand}. Also note that all unstable variations are flux variations, that is, it happens through the condensation of $m$-particle. It is not possible for the $e$-particle to condense in this way. 

For the double semion model, the string-net data is
\begin{eqnarray}
N=1, \, d_0=1, \, d_1=1; \nonumber \\
 \delta_{000}=\delta_{110}=\delta_{101}=\delta_{110}=1; \nonumber \\ 
G^{000}_{000}=1; \nonumber \\ 
G^{011}_{011}=G^{101}_{101}=G^{110}_{110}=-1 ; \nonumber \\ 
G^{011}_{100}=G^{101}_{010}=G^{110}_{001}=G^{000}_{111}= -i. \nonumber\\
\end{eqnarray}
The triple-line TNR of the double semion model can be built by plugging in this data into the general expression in Eq.~\eqref{eq:snT0}. This tensor has 9 virtual indices, each of which takes 2 values. So the full virtual space is $\textrm{rank}(I_V)= 2^9=512$ dimensional. Dimension of the stand-alone space is 
\begin{eqnarray}
\textrm{rank}(M_0) &=& \sum_{a,b,c;i,j,k}\delta_{i,b,c}\delta_{j,c,a}\delta_{k,a,b} =8, 
\end{eqnarray}
and the dimension of the MPO-injective subspace is 
\begin{eqnarray}
\textrm{rank}(\mathbb{M} ) &=& \sum_{i,j,k}\delta_{i,j,k}=4.
\end{eqnarray}
These imply that $\textrm{rank}(I_V-M_0)=512-8=504$ and $\textrm{rank}(M_0 -\mathbb{M})=8-4=4$. So we reach the conclusion that out of 512 possible variations, 504 are stable since they are outside the stand-alone space. In the remaining 8, 4 are in stable and 4 are unstable. The numerical calculation supporting this conclusion is shown in Fig.~\ref{TLDSrand}. 
\begin{figure}
\begin{center}
\includegraphics[width=9cm]{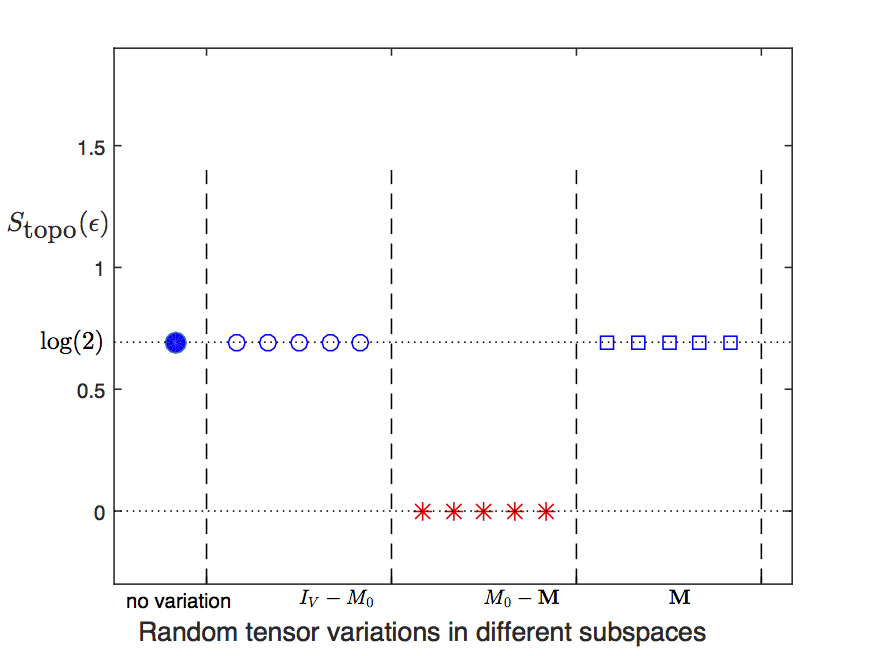}
\caption{Numerical calculation of topologiccal entanglement entropy $S_{\text{topo}} (\epsilon)$ of states represented by double semion model fixed point triple-line tensors, $T^0$, varied with an infinitesimal random tensor in different subspaces. $\epsilon$ value is kept fixed at $\epsilon=0.1$. Blue dot corresponds to $S_{\text{topo}}$ with no variation. $I_V$ is projector onto the full virtual space. $M_0$ is the projector on the stand-alone subspace. $\mathbb{M}$ is the MPO  projector.  We take a random tensor and apply the projectors to generate random tensors in respective subspaces. Details of this numerical calculation are given in the Appendix A.2.}
\label{TLDSrand}
\end{center}
\end{figure}
Also note that all unstable variations are plaquette variations, that is, it happens through condensation of the boson of the double-semion model. The classification of all variations is the same as that for toric code shown in \eqref{TLTC venn} above. 

Now we are ready to discuss a concrete example of the string-net triple line TNR and its instabilities. We choose double-Fibonacci model for two main reasons: 1- Unlike the toric code and the double-semion model, it is a non-abelian model, so the general triple-line TNR, as far as we know, cannot be reduced to a double-line or single-line TNR. So it serves as a good example to test our conjecture for the general string-net TNR. 2- Unlike toric code and double-semion, its bosonic string operator is not a zero string operator, so it does not disappear along the path. 
\section{A non-abelian example: Double-Fibonacci Model}
\label{sec: fb}

Toric code and the double-semion models are abelian models. Now we will discuss a non-abelian model: the double-Fibonacci model. The Ground state of non-abelian string net models cannot be described by a single-line or  the double-line TNR; it only accepts a triple line TNR (tensor in Eq.~\eqref{eq:snT0}). Let's first describe the model briefly. The data for this can be found in section IV.B of \citet{Levin05}. There is one type of string ($N=1$). Its quantum dimension is, $d_1=\gamma = \frac{1+\sqrt[]{5}}{2}$. Its branching rules are,
\begin{align*}
\delta_{ijk} = \begin{cases} 
0 & \text{ if } i+j+k =1; \\ 
1 & \text{otherwise}.	\end{cases}
\end{align*}
\begin{align} \label{fibG}
 d_0&= 1, \, d_1 = \gamma,  \quad \text{where }\, \gamma^2=\gamma+1 \\
 G^{111}_{111}&=-\frac{1}{\gamma^2}; \, G^{110}_{111}=\frac{1}{\gamma};  \, G^{110}_{110}=\frac{1}{\gamma}; \,G^{000}_{111}=\frac{1}{\sqrt{\gamma}}; \, G^{000}_{000}=1.	
\end{align}
The branching rules tells us that one string is allowed to branch into two, unlike the abelian models we have studied until now. First, let's apply our conjecture to find out how many unstable directions we should expect. 
The triple-line TNR of the Fibonacci model can be built by plugging in this data into the general expression in Eq.~\eqref{eq:snT0}. This tensor has 9 virtual indices, each of which takes 2 values. So the full virtual space is $\textrm{rank}(I_V)= 2^9=512$ dimensional. The dimension of the stand-alone space is 
\begin{eqnarray}
\textrm{rank}(M_0) &=& \sum_{a,b,c;i,j,k}\delta_{i,b,c}\delta_{j,c,a}\delta_{k,a,b} =18, 
\end{eqnarray}
which is  bigger than that of the toric code and the double-semion models. The dimension of the MPO-injective subspace is 
\begin{eqnarray}
\textrm{rank}(\mathbb{M} ) &=& \sum_{i,j,k}\delta_{i,j,k}=5
\end{eqnarray}
which implies that $\textrm{rank}(I_V-M_0)=512-18=494$ and $\textrm{rank}(M_0 -\mathbb{M})=18-5=13$. So we reach the conclusion that out of 512 possible (virtual) variations, 494 are stable since they are outside the stand-alone space. In the remaining 18, 5 are in stable as they are in the MPO-injective subspace and remaining 13 are unstable. The classification of all variations can be represented pictorially as follows: 
\begin{eqnarray}
\includegraphics[scale=0.3]{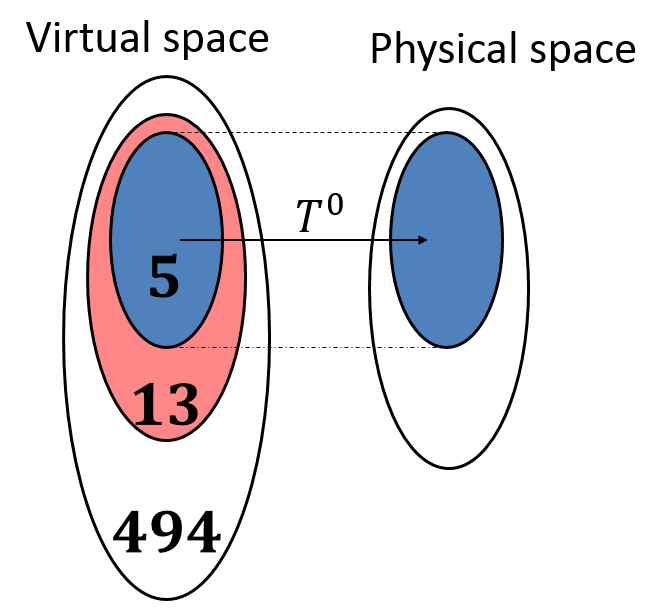}
\label{TLFIB venn}
\end{eqnarray}
The numerical calculation supporting this conclusion is shown in Fig.~\ref{TLFIBrand}. \par 
Comparing it to the toric code and the double-semion models we see that the Fibonacci triple-line TNR is significantly more unstable. Another difference is that the stand-alone space does have vertex unstable variations in addition to plaquette ones. Out of 13 unstable variations in $M_0-\mathbb{M}$ the following 3 are vertex variations and the rest 10 are plaquette variations:
\begin{eqnarray}
|a,b,c;i,j,k\rangle&=&|1,1,1;1,0,0\rangle,|1,1,1;0,1,0\rangle,\nonumber \\
& &|1,1,1;0,0,1\rangle. 
\end{eqnarray}
That is, the following 3 tensor components are allowed in the stand-alone space but not in the physical space:
\begin{equation}
    \includegraphics[scale=0.4]{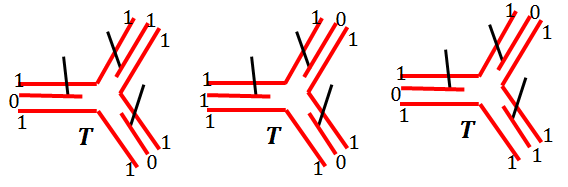}.
\end{equation}

Since $\delta_{i,j,k}=\delta_{1,0,0}=\delta_{0,1,0}=\delta_{0,0,1}=0$ these 3 vectors are not in the MPO-injective subspace $\mathbb{M}$.\par 
To understand the physics behind this, we need to look at the quasi-particles of the Fibonacci model. There are 3 quasi-particles excitations: $\tau, \bar{\tau}$, and $\tau\bar{\tau}$.  The $T$ and $S$ matrices of the particles are as follows:
\begin{align}
T = \begin{bmatrix}
1 & 0 & 0 & 0 \\ 0 & e^{-\frac{4}{5}\pi i} & 0 & 0 \\ 0 & 0 & e^{\frac{4}{5}\pi i} & 0 \\ 0 & 0& 0& 1
\end{bmatrix}, 
S= \frac{1}{1+\gamma^2}\begin{bmatrix}
1 & \gamma & \gamma	 & \gamma^2 \\ \gamma  & -1 & \gamma^2 & -\gamma \\ \gamma & \gamma^2 & -1 & -\gamma \\ \gamma^2 & -\gamma& -\gamma & 1
\end{bmatrix}.
\end{align}
It is best seen as two layers of Fibonacci model with opposite chiralities. $\tau$ and $\bar{\tau }$ are particles in the two respective layers. They have non-trivial self statistics. But, because they are in different layers, they have a trivial statistics with one another. And the boson, $\tau \bar{\tau} $ is the composition of the Fibonacci particles in the two layers. The string operator for these quasi-particles are given in equation (51) of \citet{Levin06}. We are most interested in the boson of the model, so let us write its string operator ($\Omega$ matrices) explicitly: 
\begin{eqnarray} \label{4thSO}
n_{4,0}=1, n_{4,1}=1, \Omega^0_{4,000}=1, \Omega^1_{4,110}=1, \nonumber \\ 
\Omega^1_{4,001}= -\gamma^{-2}, \Omega^0_{4,111}=\gamma^{-1}, \Omega^{1}_{4,111}=\gamma^{-5/2}, \nonumber \\ 
\Omega^1_{4,101}= \Omega^{*1}_{4,011}= \gamma^{-11/4}(2- e^{3\pi i/5}+\gamma e^{-3 \pi i/5}) .
\end{eqnarray}

\begin{figure}[t]
\begin{center}
\includegraphics[trim=10mm 0mm 10mm 10mm, width=8cm]{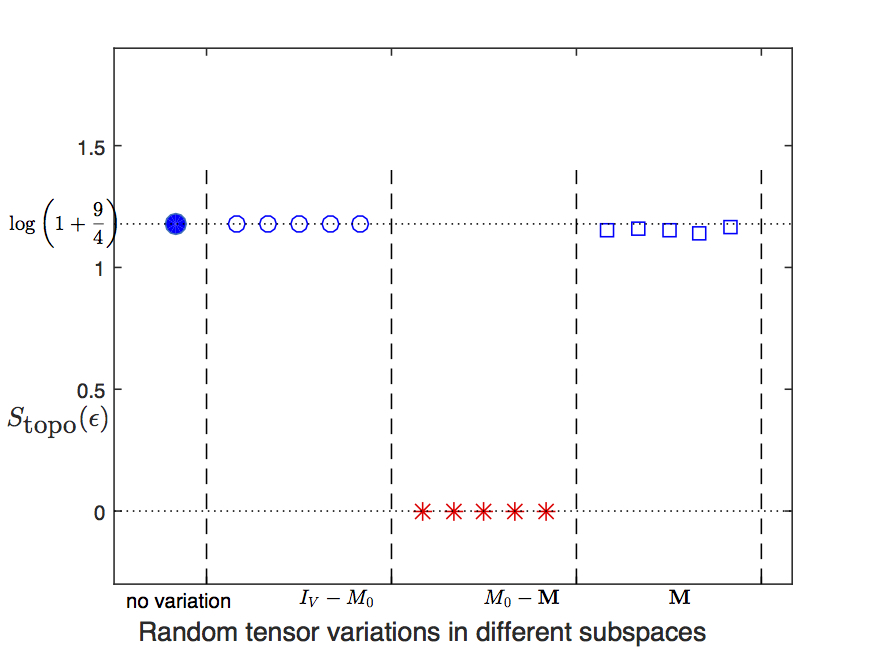}
\caption{ Numerical calculation of topological entanglement entropy $S_{\text{topo}} (\epsilon)$ of states represented by Fibonacci model fixed point triple-line tensors, $T^0$, varied with an infinitesimal random tensor in different subspaces. $\epsilon$ value is kept fixed at $\epsilon=0.1$. Blue dot corresponds to $S_{\text{topo}}$ with no variation. $I_V$ is projector onto the full virtual space. $M_0$ is the projector on the stand-alone subspace. $\mathbb{M}$ is the MPO-injective subspace projector.  We take a random tensor and apply the projectors to generate random tensors in respective subspaces. The exact numerical values on this plot can be found in Appendix A.2. }
\label{TLFIBrand}
\end{center}
\end{figure}
One can see that it is not a simple-string operator: when applied on the vacuum, it creates both 0-type and 1-type strings. So we see that the double-Fibonacci model is different from the above two examples in one crucial aspect: the boson string operators in the toric code and the double-semion models were zero-string operators for the given TNRs. That is, the string operator `disappeared' along the path (Figs.~\ref{econd},~\ref{mcond}), not changing tensors along the path. This is why a single variations standing alone could be thought of as an operator sitting at the ends of an invisible string operator. But the same is not true for the double-Fibonacci model. The string operator corresponding to the boson  $\tau \bar{\tau} $  does not disappear in the middle.\par 
Because the bosons don't have a zero string operator, one might conclude that there would be no unstable directions as bosons cannot condense. However, numerical calculations find that there actually are unstable directions. How can we understand that? \par 
We look at how the boson string-operator changes the tensors along the path. In Fig.~\ref{open string operator},  one can see that a wave function corresponding to the boson sitting at two places, $v_1$ and $v_2$, is actually a superposition of many wave functions:
\begin{eqnarray}
|\Psi_{\text{boson}}\rangle &=& \sum_{t_1,s,t_2}n_s \Phi_{t_1,s,t_2} |\Psi_{\text{gs}}\rangle \nonumber \\
&= & |\Psi_{0,0,0}\rangle+|\Psi_{1,0,0}\rangle+|\Psi_{0,0,1}\rangle+|\Psi_{1,0,1}\rangle \nonumber \\
&+& |\Psi_{0,1,0}\rangle+|\Psi_{1,1,0}\rangle+|\Psi_{0,1,1}\rangle+|\Psi_{1,1,1}\rangle, \nonumber \\ \label{4state}
\end{eqnarray}
\begin{figure}
\includegraphics[width=\columnwidth]{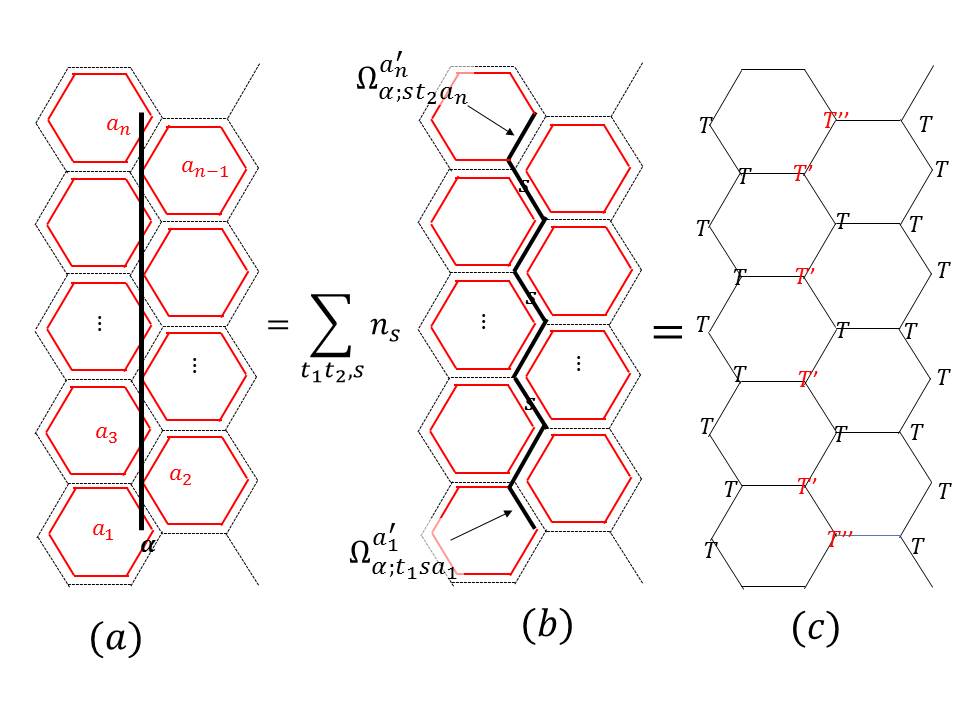}
\caption{Action of a generic (simple and non-simple) open-end string operators corresponding to anyon $\alpha$ on tensors can be calculated in a similar fashion as that of simple-string operator Wilson loops. (a) We start with applying the string operator on the 'loop state' on the fattened lattice. (b) The string operator becomes a superposition of operations $\sum_s n_s \Phi_{t_1st_2}$. $\Phi_{t_1st_2}$  acts as follows: at the ends, the string operator acts as $\Omega_{\alpha;t_1 s}$ and $\omega_{\alpha;st_2}$ matrices on the plaquette-loops, while in the middle, it is simply a $s$-type string to be fused with the nearby plaquette loops. (c) We fuse all strings in the previous step to get the physical state. The effect of the string operator can be absorbed into redefining the tensors along the path. A generic string operator changes the tensors along its path. The only case where it doesn't change the tensors is for simple-string operators of type 0. \label{open string operator}}
\end{figure}
where the operator $\Phi_{t_1,s,t_2}$ is explained in Fig.~\ref{open string operator}. $\Phi_{t_1,s,t_2}$ is equivalent to applying $\Omega^{s'_1}_{4; t_1,s,s_1}$ and $\bar{\Omega}^{s'_n}_{4;s,t_2,s_n}$ on the loops at the ends of the string operator, and creating a $s$ type string along the path. Fusing the loops with each other and with the $s$ string along path $P$ gives the final state. The important thing to note is that, though a TNR of the full state $|\Psi_{\text{boson}}\rangle $ involves changing tensors along the path, the TNR of $|\Psi_{t_1,0,t_2}\rangle, \,  t_1,t_2=0,1 $, have tensors changed only on the ends. Simply putting, the zero-string component of the string operator does not change the tensors $T^0$ in the middle, as expected. So the boson state has a finite overlap with the state where tensors are changed only at the ends. So when the variations corresponding to the ends of this zero-string component of the boson operator proliferate, it effectively condenses the bosons, as they have finite overlap with the resulting state.\par

\par 
So in conclusion, we see that although the boson string operator is not a zero-string operator, that is, it does not disappear in the middle for the triple-line TNR, its zero-string component still causes an instability because the resulting state has a finite overlap with the boson-condensed state. 
\par 
Now we have looked through important examples of string-net TNR and their instabilities. Finally, we will give a proof of instability in the generic case. 


\section{Conclusions and Discussion}
\label{conclusion}

In this paper, we try to answer the following question: are the tensor network representation of string-net states stable? That is, if we start from the tensor network representation of a string-net state and add arbitrarily small variations to the local tensor, does the topological order of the represented state always remain the same? This is an important question because if the answer is no, then the task of determining topological order of a tensor network state may be numerically `ill-posed'. That is, arbitrarily small numerical error in the process may change our conclusion in a qualitative way. Previous work\cite{Chen10} has shown that this is indeed the case for the single line representation of the toric code state. While this may seem to seriously limit the applicability of tensor network methods to the study of the toric code type topological order, Ref.\onlinecite{Chen10} also identified an inner $Z_2$ symmetry by preserving which the numerical task becomes `well-posed' again. 

We want to know if similar problems happen for general string-net states. In particular, we asked
\begin{enumerate}
\item{Does the tensor network representation of other string-net states also have unstable directions of variation?}
\item{If so, can they be avoided by preserving certain symmetries in the tensor?}
\item{What is the physical reason behind such instabilities and their prevention?}
\end{enumerate}
We found that
\begin{enumerate}

\item{All string-net tensors have unstable directions of variation.}

\item{To avoid such instabilities, we need to avoid `stand-alone' variations that break the Matrix-Product-Operator(MPO) symmetry introduced in Ref.\onlinecite{csahinouglu2014characterizing, Buerschaper2014}.}

\item{The physical reason for the instability is that `stand-alone' variations which violate these symmetries induce condensation of bosonic quasi-particles and hence destroys (totally or partially) the topological order.}

\end{enumerate}
 \par 
 We demonstrated the case explicitly for the tensor network representation of the toric code (single, double, triple line), the double semion, and the double Fibonacci model, by calculating the topological entanglement entropy $S_{\text{topo}}$ of tensors with random variations. We observe that MPO symmetry preserving variations keep $S_{\text{topo}}$ invariant and MPO symmetry breaking variations lower $S_{\text{topo}}$ (to zero). While for general string-net models, we cannot prove the above claim analytically, we are able to show that 1. the fixed point tensor of any string-net has unstable directions (which break the MPO symmetry) 2. MPO breaking variations induce the condensation of bosons in the state, and therefore destroy (at least partially) the topological order. Moreover, we point out that to correctly simulate the local properties of a phase transition induced by such boson condensation, these MPO breaking variations must be \textit{allowed} in the variational calculation; otherwise, one may reach the wrong conclusion about the phase transition (e.g. regarding the order of the transition). This has been observed in the case of toric code in Ref.\onlinecite{Gu08}.

Given this result, we can ask, how to properly design the tensor network algorithm so that it can correctly simulate topological phases and phase transitions? In particular, if we want to determine whether the ground state of some Hamiltonian has topological order by calculating topological entanglement entropy in the thermodynamic limit, we need to use a variational ansatz with the proper MPO symmetry. How to do that in an efficient and unbiased way is an interesting open question. 

On the other hand, if we want to properly simulate a topological phase transition induced by boson condensation, we need to put in the proper variational parameter. However, as we have seen in the case of the toric code, different representations (single line, double line, triple line) contain parameters corresponding to the condensation of different bosons ($e$ or $m$). In fact, none of the representations contain parameters which correspond to the condensation of both bosons. Therefore, it is not possible to use any of them to correctly obtain the full phase diagram. It implies that, if we want to study a topological phase transition whose nature is unknown, we need to try different ans\"atze. How to do that in an efficient and unbiased way is again an interesting open problem. We leave these problems to future study.

\begin{acknowledgments}
Sujeet Shukla would like to thank Pinaky Bhattacharyya for help with the numerical calculations. 
X.C. is supported by the Caltech Institute for Quantum Information
and Matter and the Walter Burke Institute for Theoretical Physics.

\end{acknowledgments}
\par

\appendix
\section{Decomposing stand-alone space using Wilson-loops: MPO symmetries}
\label{MPO symmetry come from anyon}
In sections \ref{subsec:PUM} and \ref{subsec:PUM0M}, we argued how the stand alone space $M_0$ decomposes further into two subspaces, $\mathbb{M}$ and $\mathbb{M_0}-\mathbb{M}$ on the basis whether a stand-alone variation can be lifted to the physical level locally or non-locally. In doing so we used the fact about topological models: anyonic excitation cannot be removed by a local operation but an elementary excitation can be. \par 
There is another way to distinguish between trivial and non-trivial excitations. 
 Consider the tensor network state made out of $T^0$, except at site $s_0$, $T^0$ has been replaces by some stand-alone tensor $T$. Now we want to find out whether this variation/excitation is a topologically non-trivial excitation. In topological models the way to detect the presence of anyon is by measuring Wilson-loop operators around it. We will do the same here, but on the virtual level. Doing so will reveal another interpretation of the MPO subspace/symmetries: these symmetries come from Wilson-loops of anyons of the model. \par 
Consider the following physical process. We generate an anyon $a$, anti-anyon $\bar{a}$ pair, move $a$ around the site $s_0$ where $T$ is sitting and finally fuse it with $\bar{a}$. Mathematically, this is equivalent to applying a Wilson loop operator $W_a(C)$ corresponding to particle $a$. $C$ represents the closed curve/loop around the site. If there was another anyonic excitation $b$ present at $s_0$ and if $a$ and $b$ have a non-trivial braiding statistics with each other, then this process produces a phase factor. Hence application of $W_a(C)$, where $C$ is a loop around a site can be used to detect if there is a topologically non-trivial excitation present at the site. Of course $W_a(C)$ are  symmetries of the ground state for all anyons $a$. But more than that, it would be a symmetry of any state with a trivial local excitation sitting at $s_0$.  \par 
$W_a(C)$ is an operator on the physical degrees of freedom, which induces an operator, $M_a(C)$, on the virtual degrees of freedom. $W_a(C)$ is guaranteed to have a representation $M_a(C)$ on the virtual level because $W_a(C)$ is an operator supported on the ground-state physical space of local tensors, and as we noted in Lemma \ref{MPOinj}, such an operator can be mapped to an operator on the virtual level. Hence, just as $W_a(C)$ is a symmetry on the physical level, $M_a(C)$ should be a symmetry of the ground state tensor $T^0$ on the virtual level. But, in fact, any stand-alone variation $T$ that is topologically trivial excitation would be symmetric under $M_a(C)$ for all $a$.  A tensor variation that breaks this symmetry for some $a$ would imply the presence of a non-trivial excitation. So the space of stand-alone tensors $T$ that satisfy $M_a(C)$ symmetries for all $a$ has to the space of topologically trivial excitations. This precisely is the source of MPO symmetries, and $\mathbb{M}$ is nothing but the projector onto the $M_a(C)$ symmetric subspace for all $a$. In fact this is why the MPO projector for both double-line (Eq.~\eqref{DLMPO} and single-line (Eq.~\eqref{SLMPO}) could be written in terms of loop operators on the virtual level. These loop operators are nothing but the Wilson loop operators on the virtual level. 

Let's illustrate the above discussion with the single-line TNR of toric code state. Let's say the stand-alone tensor $T$ is surrounded by $T^0$. We apply an $m$-particle Wilson-loop around this stand-alone tensor. This Wilson-loop applies $Z$ operators on the physical legs of the surrounding $T^0$ tensors. We have already seen that this operation can be brought down to the virtual level (Eq.~\eqref{SLZVZP} in the opposite direction),
\begin{eqnarray}\label{SLZZ}
\includegraphics[scale=0.3]{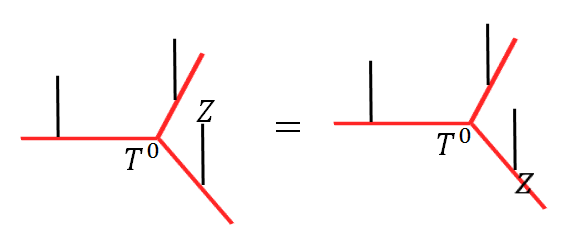}
\end{eqnarray}
Keeping in mind that $T^0$ also satisfies the $Z^{\otimes 3}$ symmetry of Eq.~\eqref{SLTCsym}, we see that the $m$-particle Wilson-loop finally reduces to a $Z^{\otimes 3}$ operator on the stand-alone tensor $T$. That is,
\begin{eqnarray}
\includegraphics[width=0.9\columnwidth]{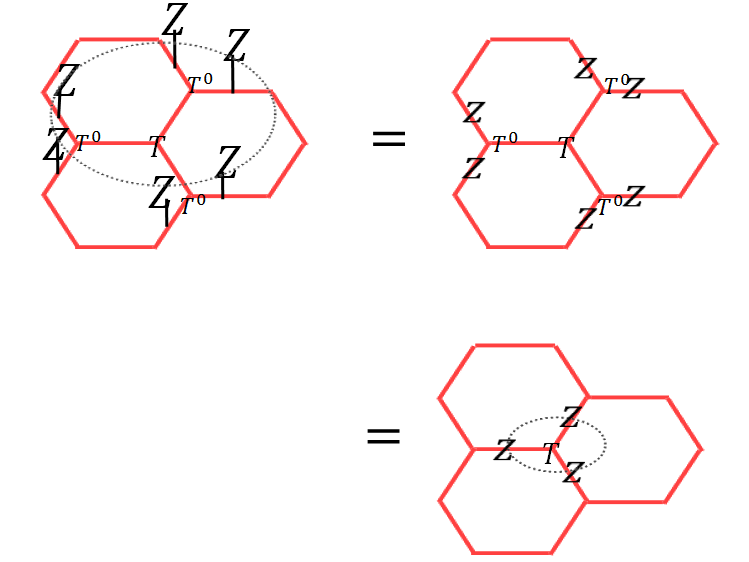}
\end{eqnarray}
In first equality, we have used relation \eqref{SLZZ} and in second equality we have used the $Z^2$ symmetry of the single-line Tensor. So we find that the representation of $m$-particle Wilson loop, $W_m(C)$ on the stand-alone space is, $M_m(C)= Z^{\otimes 3}$. So we have shown that the presence of $Z^{\otimes 3}$ symmetry constraint inside the stand-alone space actually comes from the $m$-particle Wilson loop. 
\par
Now a natural question arises: why isn't there an analogous symmetry constraint on the tensor corresponding to an $e$-string operator Wilson loop? Let's apply the $e$-particle Wilson-loop, which is a loop of $X$ operators on the single-line TNR, and then bring it down to the virual level. We find,
\begin{eqnarray}
\includegraphics[scale=0.4]{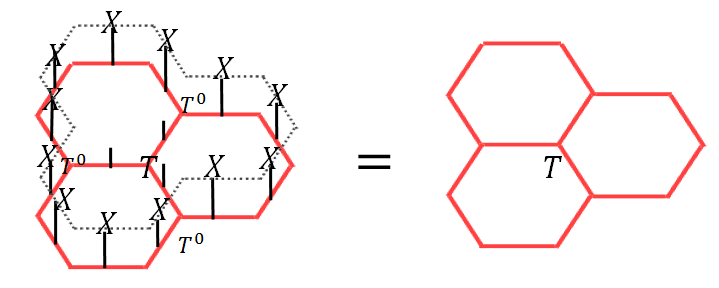}
\end{eqnarray}
where we have used the fact that $X$ operators on the nearby virtual legs simply cancel each other. This was already noted in the discussion of zero-string operators  and in Eq.~\eqref{econd}. So we see that the $e$-particle Wilson-loop poses no extra symmetry constraint on the stand-alone tensors. \par 
Now let's see if the MPO symmetry of double-line TNR also comes from a Wilson-loop. Double-line case is more interesting than the single-line case because, as we have already discussed, the double-line has a stand-alone space smaller than the full virtual space. We first look at the $e$-particle Wilson-loop, which is a loop of $X$ operators on the physical level. We have already seen that this operation can be brought down to the virtual level (Eq.~\eqref{DLXVXP} in the opposite direction)
\begin{eqnarray}\label{DLXPXV}
\includegraphics[scale=0.4]{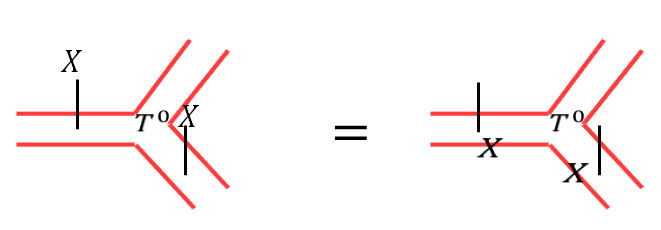}.
\end{eqnarray}
So we find
\begin{eqnarray}\label{DLTCXWilson}
\includegraphics[scale=0.4]{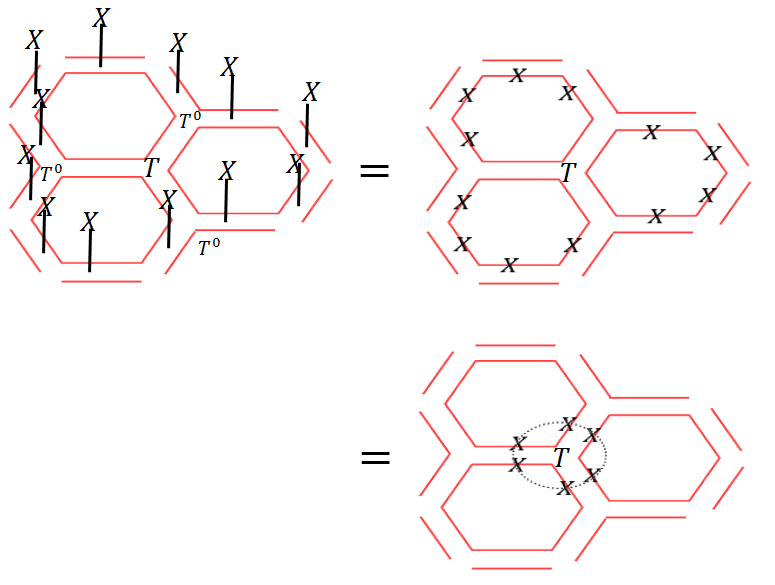},
\end{eqnarray}
where in first equality we used Eq.~\ref{DLXPXV} and in the second equality we simply used the relation
\begin{eqnarray}\label{4X2X}
\includegraphics[scale=0.3]{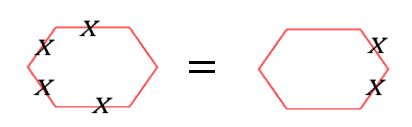}.
\end{eqnarray}
 So we have shown that the MPO symmetry, $X^{\otimes 6}$, of double-line TNR is actually a representation of the  $e$-particle Wilson-loop on the stand-alone space. At this point, it is important to note that  relation \eqref{4X2X} holds only when $T$ is in the stand-alone space, so it has the $Z^{\otimes 2}$ symmetry. If $T$ was outside the stand-alone space, this would not be true. \textit{This is why we say that MPO symmetries, $X^{\otimes 6}$ in this particular case, are  representations of the Wilson-loops on the stand-alone space, not on the full virtual space. } \par 
 Now we analyze the $m$-particle Wilson-loop, which is a loop of $Z$ operators. Eq.~\eqref{DLZZVP} tells us how to bring down the $Z$ operators on double-line fixed point tensor, $T^0$. Using this and other obvious properties of $T^0$ and $T$, we find
 \begin{eqnarray}
 \includegraphics[scale=0.4]{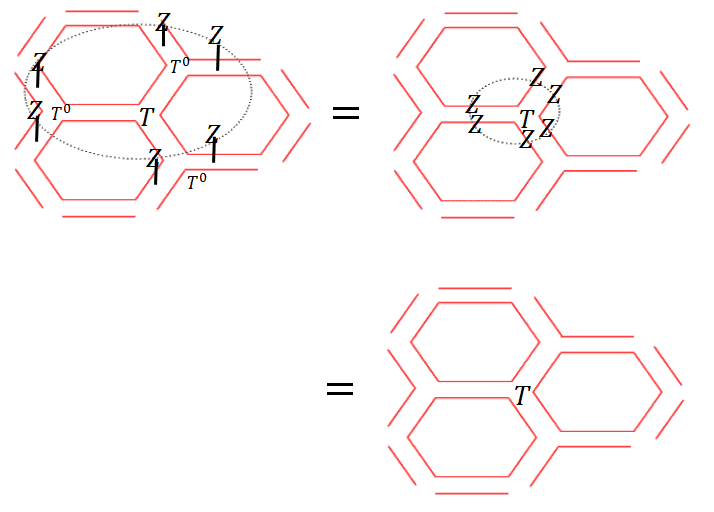}.
 \end{eqnarray}
 The first equality follows from Eq.~\eqref{DLZZVP} and the fact that $Z$ operators can be slid along contracted virtual legs. The last equality follows from the fact that $T$ is a stand-alone tensor, so it satisfies the $Z^{\otimes 2}$ symmetries by definition. Or, in other words, the representation of the Wilson-loop operator on the stand-alone space is $M_m(C) = I^{\otimes 3}$. That is, it is represented trivially.  So we see that all stand-alone tensors satisfy the $m$-particle Wilson-loop symmetry. Hence  this symmetry poses no extra constraint within stand-alone space, and that is why the MPO-injective subspace had only one $Z_2$ symmetry. In fact, this analysis has shown what we already knew from Eq.~\eqref{mcond}: $m$-string operator is a zero-string operator of the double-line TNR.   \par 
At this point, we can notice the similarity between double-line  $m$-particle relation and single-line $e$-particle relation. But there is a crucial difference. $W_e(C)$ has trivial representation on \textit{all} of the virtual space of single-line TNR, but $W_m(C)$ has trivial representation \textit{only} in a subspace of the virtual space  of double-line TNR. \par 
 This analysis points toward a representation theoretic way of understanding tensor instabilities. $T^0 : V\rightarrow P$ is a linear map from virtual vector space to the physical vector space. This map induces a representation of operators on the physical space in the virtual space. In particular, it induces the representation of Wilson-loop operators, $W_a(C)\rightarrow M_a(C) $. Such a representation is always possibly as is guaranteed by the MPO-injectivity (lemma \ref{MPOinj} ). In fact, this representation would be faithful on individual tensors. But there is no guarantee that it would be \textit{faithful} on the whole tensor network, because $M_a(C)$ can be a gauge-string operator, as we have already discussed for $W_e(C)$ in single-line and $W_m(C)$ in double-line. So the string-operator algebra is not faithfully represented on the virtual level. It is this unfaithful representation of anyonic algebra that causes tensor instability.  \par

\section{A brief review of string-net models}
\label{string net review}
String-net models, which are Hamiltonian realizations of Turaev-Viro TQFTs, are introduced by ~\citet{Levin05} as RG fixed point models that describe topological order in $2+1$ spacetime dimensions. Following are the defining data of the string-net states: \\ 
1- \textbf{Local Hilbert space: }String-nets are lattice spin models. Spins sit on the links of hexagonal lattice. Each spin $s$ can be in $N+1$ state, $s=0,1,2, \ldots,N$. $s=j$ at a link can be understood as a string of `type $j$' present on the link. Strings are oriented and $i^{\star }$ denotes string type $i$ with the opposite orientation. If $i=i^{\star}$ the strings are called `unoriented'. We have assumed the strings to be unoriented in the present paper for the sake of simplicity, though our results can easily be generalized to the oriented case. \\ 
2- \textbf{Branching rules:} There are branching rules denoted by $\delta_{ijk}$. $\delta_{ijk}=1$ if string type $i,j,k$ are allowed to meet at a point, and $\delta_{ijk}=0$ otherwise. \\ 
3- \textbf{Quantum dimensions:} For every string type $s$, there is a value $d_s$ associated to it, called its quantum dimensions. $D=\sum_s d_s^2$ is called the 
`total quantum dimension'. \\ 
4- \textbf{String-net condensed state:} If we assign a particular string to each link, it forms a string-net configuration on the lattice. A string-net condensed quantum state is a superposition of these different string-net configurations on the lattice. Let's denote the string-net configurations with $X$. So a string-net condensed state is, 
\begin{eqnarray}
|\Psi\rangle = \sum_X \Phi_X |X\rangle,
\end{eqnarray}
where $\Phi_X$ is the amplitude  with which a configuration $X$ appears in the description of the state. In general, $\Phi_X$ can be complicated and states belonging to the same topological phase might have different wave functions. However, if we perform an RG process, then all states in the same phase would end at the same fixed point state, which is to say that they should look the same at large distances. $\Phi_X$ can be described for this fixed point state. Though their absolute values are again complicated, we can give their relative values by describing local constraints on how amplitude $\Phi_X$ changes as we deform a configuration $X$ locally. These deformations involve rebranching, removing bubbles, fusing two strings together, etc.  These constraint equations are given in equation (4)-(7) of \citet{Levin05}. The most significant of these local constraint is the so called '$F$'-move.\\ 
5- \textbf{$F$-symbols:} A local constraint involving rebranching of 5 strings is the following: 
\begin{eqnarray}
\includegraphics[width=\columnwidth]{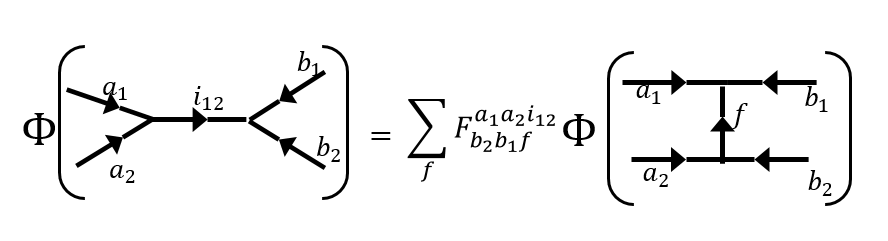}
\end{eqnarray}
$F$-symbol is a six indexed object and it satisfies the following properties: 
\begin{eqnarray}
F^{ijk}_{j^{\star}i^{\star}0} &= & \frac{\sqrt[]{d_k}}{\sqrt[]{d_i}\sqrt[]{d_j}} \delta_{ijk}, \label{F0} \\ 
F^{ijm}_{kln}= F^{lkm^{\star}}_{jin} &=& F^{jim}_{lkn^{\star}}=F^{imj}_{k^{\star}nl}\frac{\sqrt[]{d_md_n}}{\sqrt[]{d_jd_l}}. 
\end{eqnarray}
Properties of the $F$-symbol under index permutations can be best captured by defining a new object called $G$-symbol by $G^{ijk}_{klm}= \frac{F^{ijk}_{klm}}{\sqrt[]{d_kd_m}} $. $G$-symbol can be considered as a value associated to a tetrahedron and the six indices sit on the six edges of tetrahedron. Then it is invariant under all tetrahedron symmetries. It satisfies an important equation, the so-called `Pentagon Identity':
\begin{eqnarray}
\sum_f d_f G^{b_1b_2i_{12}}_{a_2a_1f} G^{b_2b_3i_{23}}_{a_3a_2f}G^{b_3b_1i_{31}}_{a_1a_3f} = G^{i_{23}i_{31}i_{12}}_{a_1a_2a_3}G^{i_{23}i_{31}i_{12}}_{b_1b_2b_3}. \nonumber \\ \label{pentagon1}
\end{eqnarray}
Finally we describe the exactly solvable Hamiltonian such that the RG fixed point state defined as above is one of the ground states,
\begin{eqnarray}
H = -\sum_v A_v - \sum_p B_p, 
\end{eqnarray}
where $v$ and $p$ denote the vertices and plaquette of the lattice. The vertex term is 
\begin{eqnarray}
A_v &=& \sum_{i,j,k}\delta_{ijk} |ijk\rangle\langle ijk|. 
\end{eqnarray}
So, the vertex term simply projects configurations to only the ones that contain the allowed branchings. The plaquette term is more involved:
\begin{eqnarray}
B_p = \sum_s \frac{d_s}{D} B_p^s,
\end{eqnarray}
where $B_p^s$ is an operator that creates an $s$-type string that fuses with the strings on the plaquette. Two strings can be fused together by assuming a 0-string between them and then using $F$-moves. \par 
Finally putting all of it together, we see that the data $(N, d_i, \delta_{ijk}, F^{ijk}_{klm})$ describes a string-net model.

\subsection{Algebraic Identities}
\label{algebra}
Here we enlist multiple algebraic relations regarding string-net data that are used throughout the paper. For rotational convenience, cyclic products will be simply denoted by $\prod_{j=1}^n$ with a cyclic $j=n+1=1$. 
One of the most important identities is the `Pentagon Identity', 
\begin{align} \tag{I.1}
\sum_f d_f \prod_{j=1}^3 ( G^{b_jb_{j+1}i_{j,j+1}}_{a_{j+1}a_jf}) = G^{i_{23}i_{31}i_{12}}_{a_1a_2a_3}G^{i_{23}i_{31}i_{12}}_{b_1b_2b_3}. \label{pentagon}
\end{align}
$G$ symbols also satisfy an `orthogonality identity',
\begin{align} 
\sum_{i_{12}} G^{b_1b_2i_{12}}_{a_2a_1f}G^{b_1b_2i_{12}}_{a_2a_1f'}d_{i_{12}}=\frac{1}{d_f}\delta_{f,f'}\delta_{a_1a_2f}\delta_{b_1b_2f}. \label{orthogonal}
\end{align}
$G$-symbols are normalized as
\begin{align} 
G^{b_1b_2i_{12}}_{a_2a_10}=\delta_{a_1,b_1}\delta_{a_2,b_2}\delta_{a_1b_1i_{12}}(d_{a_1}d_{a_2})^{-\frac{1}{2}}. \label{Gnorm}
\end{align}
Cyclic products of $G$ symbols satisfy the following equation:
\begin{eqnarray}
\sum_{\lbrace b_j \rbrace }\prod_{j=1}^n( G^{b_jb_{j+1}i_{j,j+1}}_{a_{j+1}a_jf}  G^{c_jc_{j+1}i_{j,j+1}}_{b_{j+1}b_jf} d_{b_j}) 
\nonumber \\
=\sum_s \delta_{ff's} \prod_{j=1}^n( G^{c_jc_{j+1}i_{j,j+1}}_{a_{j+1}a_js}). \label{cyclic}
\end{eqnarray}
 Plaquette operators $B^f_p$ correspondingly satisfy 
\begin{align} 
B^f_pB^{f'}_p = \sum_s \delta_{ff's} B^s_p.
\end{align}
We know that if we contract an $f$-type loop we get a factor of $d_f$. Combining this with the last two equation, we find that quantum dimensions satisfy the same identity: 
\begin{align}
d_fd_{f'}=\sum_s \delta_{ff's}d_s, \label{dfdf}
\end{align}
where $d_f$ are nothing but the eigenvalues of the plaquette operators $B^f$ operators where the eigenstate is the string-net ground state.\par 
Define matrix $N^k$ as $N^k_{a,b}=\delta_{k,a,b}$. Since $N^k$ matrices are real symmetric matrices, and commute with each other for different values of $k$, they share a complete set of orthogonal eigenvectors.We write the $q$th such simultaneous eigenvector of $N^k, \, \forall k$ as 
\begin{eqnarray}\label{sq}
|s_q\rangle = \sum_a s_{q;a}|a\rangle.
\end{eqnarray}
Since quantum dimensions form one such eigenvector, we fix $s_{0;a}=d_a$. The following equations follow
\begin{eqnarray}
\langle s_q|s_{q'} \rangle &\propto & \delta_{q,q'}, \\
\langle s_q|N^k|s_{q'}\rangle &=& \sum_{a,b}s_{q;a}\delta_{k,a,b}s_{q';b} \propto \delta_{q,q'}.
\end{eqnarray}
The branching tensor $\delta_{ijk}$ is part of a fusion category data. Under the additional assumptions of braiding defined on the fusion category and braiding being sufficiently non-trivial (modularity), the $s$ above are just the columns of $S$ matrix. But we don't really need this for our results. 
\section{Triple-line TNR of string-net states}
\label{string net TNR derivation}
We now briefly describe the derivation of triple-line TNR along the lines described in the original paper by \citet{Gu09}. It is important to understand this derivation as it gives us a way to apply string-operators on triple-line TNR. \par 
 String net RG fixed point ground state can be constructed by applying plaquette operator $B_p = \sum_a d_a B_p^a $ to the vacuum state $ |0\rangle$. $B_p^a$ creates an $a$-type string loop on the plaquette $p$. 
\begin{eqnarray}
|\Psi_{\text{gs}}\rangle &=& \prod_p B_p |0\rangle=  \prod_p \sum_{a}d_aB^a_p |0\rangle \nonumber \\
& = & \sum_{a_1,a_2,..} d_{a_1}d_{a_2}.. |a_1,a_2,...\rangle , \label{loopstate}
\end{eqnarray}
where 
\begin{eqnarray}
|a_1,a_2,...\rangle = B_{p_1}^{a_1}B_{p_2}^{a_2}\ldots |0\rangle.
\end{eqnarray}
$ |a_1,a_2,...\rangle$ is a string configuration on the `fattened lattice'. We will refer to $d_{a_1}d_{a_2}.. |a_1,a_2,...\rangle$ as the  `loop state'. See Fig.~\ref{fig:Fatlattice}.
\begin{figure}
    \centering
    \includegraphics[width=0.9\columnwidth]{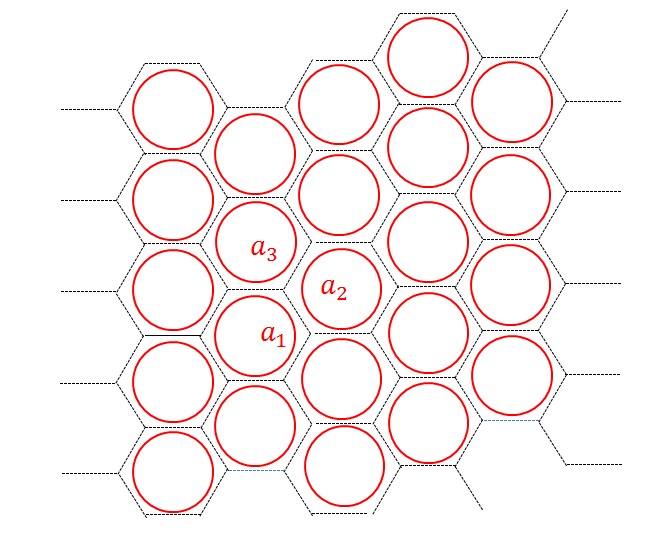}
    \caption{A loop state on the fat lattice. Fat lattice means strings are allowed to move away from the edges, as long as they dont cross the center of the plauqettes.  }
    \label{fig:Fatlattice}
\end{figure}
We need to fuse these loops together to get the  physical state. We then fuse these strings together to get the final physical state,
\begin{eqnarray}
|a_1,a_2,...\rangle = \sum_{i_{12},i_{23}...} \Phi^{i_{12}i_{23}..}_{a_1a_2a_3..} |i_{12},i_{23},..\rangle.
\end{eqnarray}
There are essentially 3 steps leading up to the expression of the triple-line TNR. We mention them here explicitly as we will need to refer back to them for other calculations.\par 
\textbf{Step 1:} We start with the `loop state' on the fattened lattice. $j$th plaquette has a loop in state $a_j$. The ground state is
\begin{eqnarray}
|\Psi_0\rangle= \sum_{ a_1,a_2,.. }d_{a_1}d_{a_2}.. |a_1,a_2,.. \rangle.
\end{eqnarray}
So every plaquette contributes a factor of $d_{a_j}$. We distribute it uniformly among the 6 vertices, so each vertex gets a factor of $d_{a_j}^{1/6}$ from each vertex. \par 
\textbf{Step 2:} We fuse all loops with nearby loops producing a string on the links: 
\begin{eqnarray}
\includegraphics[scale=0.4]{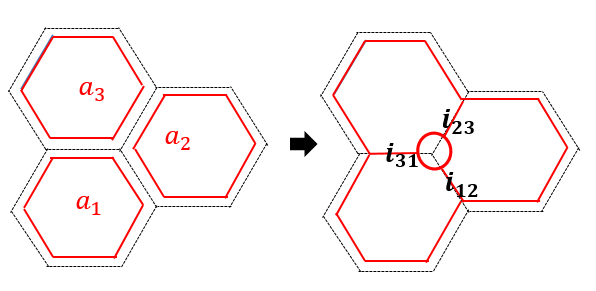}
\end{eqnarray}
We assume a 0-string between them and perform an $F$-move. It produces a factor of $\sum_{i_{j,k}} \sqrt{\frac{d_{i_{j,k}}}{d_{a_j}d_{a_k}}}  $ on each link between plaquette $j$ and $k$. A link is shared between two vertices, so each vertex gets a factor of  $\sqrt[4]{\frac{d_{i_{j,k}}}{d_{a_j}d_{a_k}}} $. \par  
\textbf{Step 3: }After the previous step, we are left with a `bubble' on the vertex. Now we remove it,
\begin{eqnarray}
\includegraphics[scale=0.4]{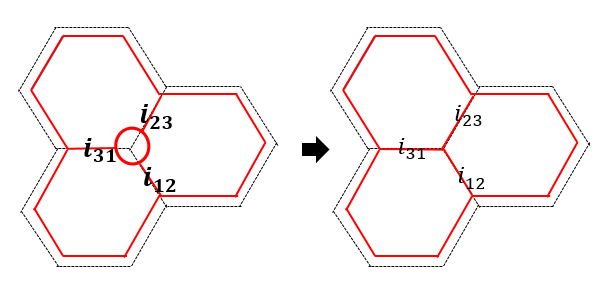}.
\end{eqnarray}
Removing it produces a factor of $\sqrt[]{d_{a_j} d_{a_k}d_{a_l}}G^{i_{kl}i_{lj}i_{jk}}_{a_ja_ka_l} $. 
\par 
Putting the 3 steps together, we get
\begin{eqnarray}
(T^0)^{i_{jk}i_{kl}i_{lj}}_{s_ls_js_k} = \sqrt[4]{d_{i_{jk}}d_{i_{kl}}d_{i_{lj}}}G^{i_{jk}i_{kl}i_{lj}}_{a_ja_ka_l} \sqrt[6]{d_{a_j}d_{a_k}d_{a_l}}.
\end{eqnarray} 
A general triple-line Tensor is represented diagrammatically as:
\begin{equation}
    \includegraphics[scale=0.4]{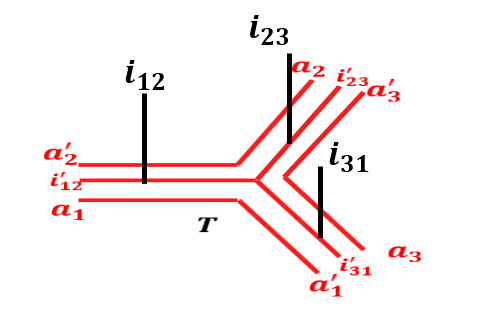}.
\end{equation}
For the specific RG fixed point tensor we have $a'_j = a_j, i'_{j,j+1}=i_{j,j+1},\, \forall j  $. So it would be represented diagrammatically as:
\begin{equation}\label{triplelineT0}
    \includegraphics[scale=0.4]{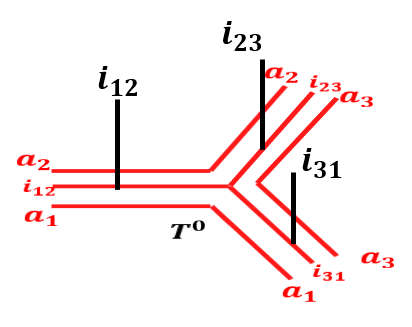}.
\end{equation}

 \section{Proof of theorem \ref{thm:string net M0}}
\label{stand-alone triple-line section app}
To calculate the stand-alone space, we need to know how to contract double-tensors on a large region. First we need to define the concept of \textit{boundary operators} that show up in double-tensor contraction. 
\subsection{Boundary operators}
It is more convenient to work with the dual lattice of honeycomb lattice. The dual lattice of honeycomb lattice is the triangular lattice. We label the vertices with an integer $j=1,2,..$. The edges are labeled by the two vertices on its ends, $(j_1,j_2)$.  The triple line tensor is represented as a triangle,
\begin{eqnarray} \label{Ttriangle}
  \raisebox{-0.5\height}{\includegraphics[width=0.4\columnwidth]{triplelineT0}} = \raisebox{-0.5\height}{\includegraphics[width=0.4\columnwidth]{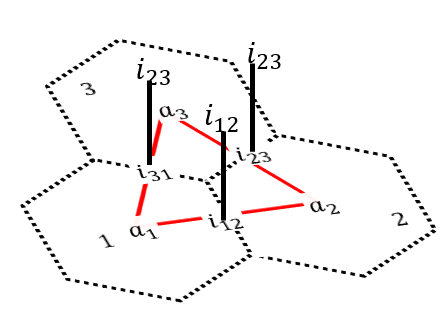}}.
\end{eqnarray}
The inner indices $a_1, a_2,..$ sit on the vertices of the triangles, and the physical legs and the middle legs on the edges. We denote the inner index sitting on vertex $j$ as $a_j$, and the physical and middle legs sitting on the edge are denoted as $i_{j_1,j_2}$. With this construction, the tensor component can be written as 
\begin{eqnarray}\label{SNtensor}
(T^0)^{i_{23}i_{31}i_{12}}_{a_1 a_2a_3} =  \prod_{j=1}^3 \left(d_{i_{j,j+1}}^{\frac{1}{4}}d_{a_j}^{\frac{1}{6}} \right) G^{i_{23}i_{31}i_{12}}_{a_1 a_2 a_3}.
\end{eqnarray}
A double tensor of a tensor $T$ is defined as $ \mathbb{T}= \sum_{I}T^{I}(T^{*})^I$ and is denoted by $\mathbb{T}$. $I$ denotes the set of physical indices. So we get the double tensor of a tensor by contracting the physical indices between $T$ and its complex conjugate, $T^{\dagger}$. Since the tensor $T$ is represented by a triangle, the double tensor $\mathbb{T}$ can be represented by a double layer triangle.

 The edge labels are the same bottom to top, only the labels of the vertices change. We label the upper vertices as $b_1,b_2,..$. With this a double tensor can be written as 
\begin{eqnarray}
\mathbb{T}^0& =&  \prod_{j=1}^n \left(d_{i_{j,j+1}}^{\frac{1}{2}}(d_{a_j}d_{b_j})^{\frac{1}{6}} \right) G^{i_{23}i_{31}i_{12}}_{a_1 a_2 a_3} G^{i_{23}i_{31}i_{12}}_{b_1 b_2 b_3}.
\end{eqnarray}
Using the pentagon equation $ G^{i_{23}i_{31}i_{12}}_{a_1a_2a_3}G^{i_{23}i_{31}i_{12}}_{b_1b_2b_3} =\sum_fd_f  \prod_{j=1}^3 (G^{b_jb_{j+1}i_{j,j+1}}_{a_{j+1}a_jf})$ we get 
\begin{eqnarray} \label{double-tensor}
\mathbb{T}^0 &=& \sum_f d_f B_f \\ 
B_f &=& \prod_{j=1}^3 \left(d_{i_{j,j+1}}^{\frac{1}{2}}(d_{a_j}d_{b_j})^{\frac{1}{6}}, G^{b_jb_{j+1}i_{j,j+1}}_{a_{j+1}a_jf} \right).
\end{eqnarray}
$B_f$ can be represented as the boundary of double-layer triangle,
\begin{eqnarray} \label{BfTriangle}
\mathbb{T}^0 &=&  \sum_f d_f  \raisebox{-0.5\height}{\includegraphics[scale=0.5]{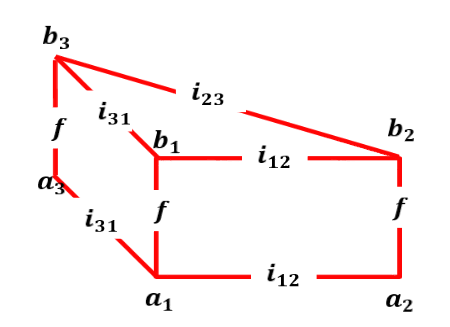}}
\end{eqnarray}
It is useful to decompose $B_f$ into terms that sit on the edge of the triangle and terms that sit on the vertices, 
\begin{eqnarray}
B_f &=&  \prod_{j=1}^3  \left(d_{i_{j,j+1}}^{\frac{1}{2}}G^{b_jb_{j+1}i_{j,j+1}}_{a_{j+1}a_jf}  \right)  \prod_{j=1}^3  \left( (d_{a_j}d_{b_j})^{\frac{1}{6}}\right). 
\end{eqnarray}
The first cyclic product on the RHS sits on the edges while the second term sits on the vertices. 
So we see that the double tensor on a triangle is (we will denote triangle as $\Delta$)
\begin{eqnarray}
\mathbb{T}^0 (\Delta )&=& \sum_f d_f B_f(\partial \Delta) \label{bulkboundarytriangle}.
\end{eqnarray}

The tensor resulting from contracting tensors $\mathbb{T}$ on a region $R$ will be denoted as $\mathbb{T}(R)$. We call $B_f$, \textit{the $f$-type boundary operator}. It lives on the boundary $\partial R$ of a region $R$,
\begin{eqnarray} \label{BfRfig}
\includegraphics[scale=0.4]{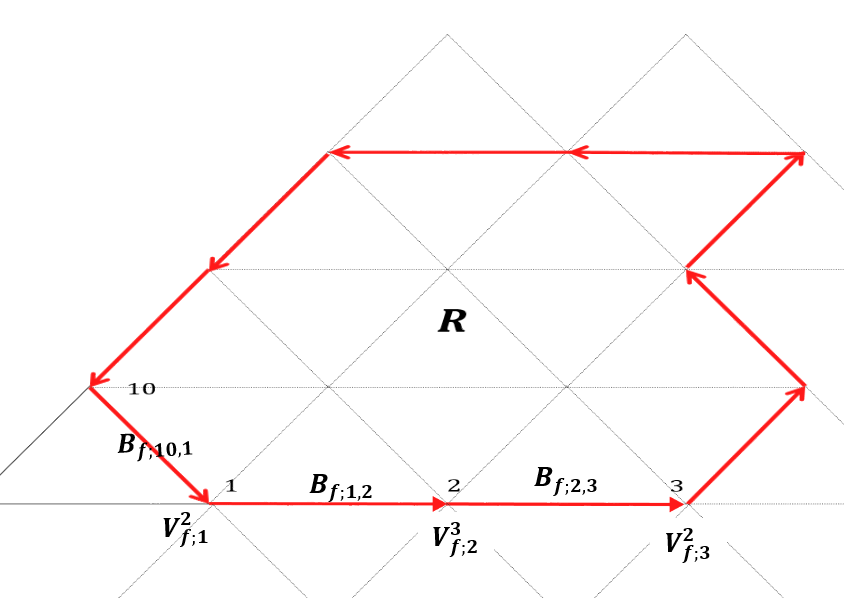}.
\end{eqnarray}
See Fig.~\ref{BfRfig}. Let's say the vertices on the boundary of a region $R$ on the triangluar lattice are labeled as $j=1,2,..,n.$. We associate with each vertex a factor of $(a_jb_j)^{\frac{m_j}{6}}$. $m_j$ denotes the number of the triangles inside $R$ meeting at vertex $j$. It can simply be written as $m_j =\theta_j/(2\pi/6)$, where $\theta_j$ is the angle the loop makes on vertex $j$. Finally, on every edge $(j,j+1)$ we associate an operator $d_{i_{j,j+1}}^{\frac{1}{2}} G^{b_j b_{j+1}i_{j,j+1}}_{a_{j+1}a_j f} $. With this construction, $B_f(\partial R)$ can be written as,
\begin{align} \label{BfR}
B_f(\partial R) = \prod_{j=1}^n  \left(d_{i_{j,j+1}}^{\frac{1}{2}}G^{b_jb_{j+1}i_{j,j+1}}_{a_{j+1}a_jf}  \right)  \prod_{j=1}^n \left( (d_{a_j}d_{b_j})^{\frac{m_j}{6}}\right).
\end{align}
Now we are ready to contract tensors on individual triangles with each other in order to find the double tensor on a region $R$.\\ 
\subsection{Double-tensor/Virtual density matrix on a general region $R$ }
\label{double tensor on R}
We present the result in a lemma. 
\begin{mylemm}\label{Result 1}
We find that the double tensor $\mathbb{T}^0(R)$ satisfies the general version of Eq.~\eqref{bulkboundarytriangle}:
\begin{eqnarray}
\mathbb{T}^0(R) = D^V\sum_f d_f^{\chi_R} B_f(\partial R),
\end{eqnarray}
where $ \chi_R = V-E+F$ is the Euler characteristic of region $R$. $V$,$E$, and $F$ are the number of vertices, edges, and faces that are completely inside the region $R$ (that is, they are inside the region where tensors have been contracted).
\end{mylemm}
\begin{proof}
There is a simple proof of this result. We have to contract $\mathbb{T}^0$ on each triangle with each other on the common edges and vertices  to get $\mathbb{T}^0(R)$, 
\begin{eqnarray}
\mathbb{T}^0(R) &=& \textrm{tTr}\left( \mathbb{T}^0(\Delta_1) \mathbb{T}^0(\Delta_2)\ldots \right) \nonumber \\ 
&=& \sum_{f_1,f_2, \ldots } (d_{f_1}d_{f_2}\ldots) B_{f_1}(\partial \Delta_1) B_{f_2}(\partial \Delta_2) \ldots . \label{Bf1Bf2}
\end{eqnarray}
where, as defined before, $\textrm{tTr}$ denotes the operation of contracting a set of tensors along shared indices. So we basically have to see how $B_{f_1}$ contracts with $B_{f_2}$. They can be contracted in two steps. First we contract all the edges, and then we contract all the vertices, and we will be left with terms sitting only on the boundary of the region.  Using the orthogonality identity, Eq.~\eqref{orthogonal}, edge contraction on the edge $(j,j+1)$ between $B_f$ and $B_{f'}$ gives 
\begin{eqnarray}
Ev(B_fB_{f'})&\propto&\sum_{i_{j,j+1}} d_{i_{j,j+1}}^{\frac{1}{2}} G^{b_jb_{j+1}i_{j,j+1}}_{a_{j+1}a_jf} d_{i_{j,j+1}}^{\frac{1}{2}} G^{b_jb_{j+1}i_{j,j+1}}_{a_{j+1}a_jf'} \nonumber \\
&=& \frac{1}{d_f}\delta_{f,f'}\delta_{a_ja_{j+1}f}\delta_{b_jb_{j+1}f}. \label{edgecontraction}
\end{eqnarray}
The factor $\delta_{f,f'}$ implies that $B_f$ only contracts with $B_f$. So the expression in Eq.~\eqref{Bf1Bf2} is only non-zero for $f_1=f_2=...f$. So we have
\begin{eqnarray}
\mathbb{T}^0(R)= \sum_f d_f^F B_f(\partial \Delta_1)B_f(\partial \Delta_2)\ldots ,
\end{eqnarray}
where $F$ is the number of faces in region $R$.  Then there are factors of $\delta_{a_ja_{j+1}f}\delta_{b_jb_{j+1}f}$ in Eq.~\eqref{edgecontraction} that will be used in the second step of vertex contraction. Finally note a factor of $d_f^{-1}$ that comes out of every edge contraction. So when we are done with all the edges, we will have an overall factor of $d_f^{-E}$, where $E$ is the number of edges. \par 
Now we do tensor contraction on each vertex. Note that  each of the six triangles around a vertex $j$ contribute a factor of $(d_{a_j}d_{b_j})^{\frac{1}{6}}$, so we have a total factor $d_{a_j}d_{b_j}$ on each vertex. We multiply this with the factor $\delta_{a_jb_jf}$ that came out of edge contraction. So, finally we have the vertex contraction using identity \eqref{dfdf},
\begin{eqnarray}
\sum_{a_j,b_j} d_{a_j}d_{b_j}\delta_{a_jb_jf}= \sum_{a_j}d_f d_{a_j}d_{a_j}=Dd_f.\label{vtxcontraction}
\end{eqnarray}
So we see that contraction of 6 tensors on each vertex simply produces a factor of $Dd_f$ for every $f$-type boundary operator. When we are done with all the vertex contractions, we will have an overall $(Dd_f)^V=D^Vd_f^V$ factor. Putting all the factors together, we get
\begin{eqnarray}
\mathbb{T}^0(R) &=& \sum_f d_f^Fd_f^{-E}(Dd_f)^V B_f(\partial R ) \nonumber \\ 
&=&D^V\sum_f d_f^{\chi_R} B_f.
\end{eqnarray}

This completes the proof.\par 
To calculate the stand-alone space, we need to know how these boundary operators behave on a large region. We present the result of this calculation in the following important lemma: 
\end{proof}
\par 
\begin{mylemm} \label{trace Bf zero}
 \begin{eqnarray}
\lim_{|\partial R|\rightarrow \infty }\frac{\textrm{Tr}(B_{f\neq0}(\partial R))}{\textrm{Tr}(B_0(\partial R))}  =0
\end{eqnarray}
\end{mylemm}
\begin{proof}
 To prove this, we would calculate $S_{\text{topo}}$ on a sphere using the virtual-density method laid out in \ref{alg} in the previous chapter,  and compare it to the known result, $S_{\textrm{topo}}=\log D$.  We divide the sphere in symmetric two halves, let's say $R$ and $L$, and calculate $\mathbb{T}(R)$.  We assume the state has the appropriate symmetry such that $\sigma_L^T=\sigma_R=\sigma_b=\mathbb{T}(R)$. Using the result by \cite{Cirac11}, we know that the physical density matrix $\rho_R$ has the same spectrum as $\sigma_b^2$, that is, $\rho_R \propto \sigma_b^2$. Let's say $\rho_R = N\sigma_b^2$, where $N$ is the normalization factor. We first calculate $N$. To do that, we first need to calculate the algebra and the trace of $B_f$. \\
Let's put the string-net tensor network state on a sphere. Consider the left hemisphere, denoted as $L$, and right hemisphere, denoted as $R$. Let's denote the indices of the vertices on the boundary $\partial R$ as $j=1,2,\ldots, n$. Then $B_f$ on this boundary is given by, 
\begin{eqnarray}
B_f(\partial R)  =\prod_{j=1}^n  \left(d_{i_{j,j+1}}^{\frac{1}{2}}(d_{a_j}d_{b_j})^{\frac{m_j}{6}}G^{b_jb_{j+1}i_{j,j+1}}_{a_{j+1}a_jf}  \right).
\end{eqnarray}
Since $R$ divides the region in to exact two halves, we assume that the boundary $\partial R$ divides the boundary plaquette in to exact two halves, setting $m_j = 3, \, \forall j$. So we get,
\begin{eqnarray}
B_f(\partial R) &=& \prod_{j=1}^n  \left(d_{i_{j,j+1}}^{\frac{1}{2}}(d_{a_j}d_{b_j})^{\frac{1}{2}}G^{b_jb_{j+1}i_{j,j+1}}_{a_{j+1}a_jf}  \right) .
\end{eqnarray}
Note that, using relation \eqref{Gnorm} 
\begin{eqnarray} \label{B0}
B_0(\partial R) &=& \prod_{j=1}^n  \left(d_{i_{j,j+1}}^{\frac{1}{2}}(d_{a_j}d_{b_j})^{\frac{1}{2}}G^{b_jb_{j+1}i_{j,j+1}}_{a_{j+1}a_j0}  \right) \\
&=&  \prod_{j=1}^n (d_{i_{j,j+1}}^{\frac{1}{2}})\delta_{a_jb_{j+1}i_{j,j+1}}.
\end{eqnarray}

Now, using identity \eqref{cyclic} the algebra of $B_f$ operators is, 
\begin{eqnarray}
B_fB_{f'} &=& \sum_s\delta_{ff's}B_s \times  \prod_{j=1}^n (d_{i_{j,j+1}}^{\frac{1}{2}}) \\
&=& \sum_s\delta_{ff's}B_sB_0.
\end{eqnarray}
We also know how to contract $B_f$ with each other through the calculations done previously in the privious subsection. We learned that $B_f$ only contracts with itself, and it gives a factor of $d_f^{-1}$ for every edge and a factor of $Dd_f$ for every vertex. On a loop the number of vertices is equal to number of edges. So we get,
\begin{eqnarray}
\textrm{Tr}(B_fB_{f'}) = \delta_{f,f'} D^n.
\end{eqnarray}
If calculate $\text{Tr}(B_f)$, we find
\begin{eqnarray}
\text{Tr}(B_f) &=& \sum_{ \lbrace a_j i_{j,j+1} \rbrace } \prod_{j=1}^n\left(d_{i_{j,j+1}}^{\frac{1}{2}}(d_{a_j}d_{a_j})^{\frac{1}{2}}G^{a_ja_{j+1}i_{j,j+1}}_{a_{j+1}a_jf}  \right) \nonumber\\
&=& \text{Tr}(A_f^n),
\end{eqnarray}
where $A_f$ is a matrix whose components $A_f(a,b)$ are $A_f(a,b) = \sum_{i}G^{abi}_{baf}(d_ad_b)^{\frac{1}{2}}d_i^{\frac{1}{2}} $. If $A_f^n$ has a non-degenerate highest eigen-value $\lambda_f$, for large $n$, $\text{Tr}(A_f^n) \approx \lambda_f^n$. Note that Perron-Frobenius theorem makes sure that $ \lambda_0$, highest eigen-value of $A_0$, will be non-degenerate. So we have
\begin{eqnarray}
\lim_{n\rightarrow \infty }\textrm{Tr}(B_0) = \lambda_0^n.
\end{eqnarray}
For abelian models, $\text{Tr}(B_{f\neq 0} =0$ since $G^{abi}_{baf}=0, f\neq 0$. For the double-Fibonacci model to be discussed in \ref{sec: fb}, a simple calculation shows $\lambda_0 = 1+\gamma^{3/2}$ and $\lambda_1 = 1- \gamma^{-\frac{1}{2}}$, where $\gamma=d_1=\frac{1+\sqrt{5}}{2}$ is the quantum dimension of the string. Because $\lambda_1 < 1$, for large $n$ $\text{Tr}(B_1)\approx \text{Tr}(A_1^n) \approx 0$. \par 

 On a hemisphere, $\chi_R=1$, so from lemma \ref{Result 1} we have $\sigma_b = \sum_f d_fB_f$ and $\rho_R = N \sigma_b^2$ where $N$ is a normalization factor. First we calculate the normalization factor $N$, 
\begin{eqnarray}
\text{Tr}(\sigma_b^2)&=&\text{Tr} (\sum_fd_fB_f)^2 \nonumber \\
&=& \sum_{f,f'}d_fd_{f'} \text{Tr}(B_fB_{f'}) \nonumber \\ 
&=& \sum_{f,f'} d_fd_{f'} \delta_{f,f'} D^n \nonumber \\
&=& D^n(\sum_f d_f^2).
\end{eqnarray}
Now, calculating Renyi entropy with renyi index $\alpha =1/2$, we get
\begin{eqnarray}
S_{1/2}(\rho_R) &=& \frac{1}{1-1/2}\log \text{Tr}(\rho_R^{\frac{1}{2}}) \nonumber \\
&=& 2\log  \frac{\text{Tr}(\sum_fd_fB_f)}{\sqrt{D^n\sum_fd_f^2}}\nonumber \\
&=& -n\log D-2\log\sum_f(d_f\textrm{Tr}B_f)- \log\sum_f d_f^2 \nonumber \\
&=& -n\log D-2n\log\lambda_0 - 2\log\left(1+\sum_{f>0} \frac{\textrm{Tr}B_f}{\lambda_0^n} \right) \nonumber \\ & & - \log\sum_f d_f^2.
\end{eqnarray}
We know that for a string-net model topological entanglement entropy is $\log\sum_fd_f^2$, which implies $\lim_{n\rightarrow \infty } \frac{\textrm{Tr}B_f}{\lambda_0^n} =0, \, \forall f>0$. This completes the proof.
\end{proof}
\subsection{String-net stand-alone subspace}
\label{Result 2}
Now we combine lemma \ref{Result 1} and lemma \ref{trace Bf zero} to prove theorem \ref{thm:string net M0}. That is, to prove that the stand alone space is given by 
 The stand alone space of the triple-line string net TNR is given by
\begin{eqnarray} 
M_0 = \delta_{a_1,a_2,i_{12}}\delta_{a_2,a_3,i_{23}}\delta_{a_3,a_1,i_{31}}.
\end{eqnarray}

\begin{proof}

Now we are ready to calculate the stand-alone space. Consider the same tensor network but on a very large disc with one triangle removed from the origin. We will denote this space as $D-\Delta$. This has two disconnected boundaries, one on the outer edge, one on the inner one. $\chi_R=0$ for this region, so using lemma \ref{Result 1}
\begin{eqnarray}
\mathbb{T}(D-\Delta) &=& \sum_f B_f(\partial (D-\Delta )) \\
&=& \sum_f B_f(\partial \Delta)\otimes B_f(\partial D)
\end{eqnarray}
To get the stand-alone space, we simply trace out the inner indices on the outer edge. But according to lemma \ref{trace Bf zero}, only $Tr(B_{ 0}$ contribute in the large disc limit.  So we simply get (up to an overall normalization factor which we ignore) $B_0$ on the triangle,
\begin{eqnarray}
\lim_{|D|\rightarrow \infty} \mathbb{T}_D(D-\Delta)&=& \lim_{|D|\rightarrow \infty}  \sum_f B_f(\partial \Delta)\otimes Tr(B_f(\partial D)) \\
&=&  B_0(\partial \Delta) \lambda_0^n.  
\end{eqnarray}
But using \eqref{B0} we get 
\begin{eqnarray}
 B_0(\partial \Delta)  = (d_{i_{12}}d_{i_{23}}d_{i_{31}})^{\frac{1}{2}} \delta_{a_1,a_2,i_{12}}\delta_{a_2,a_3,i_{23}}\delta_{a_3,a_1,i_{31}}.
\end{eqnarray}
Stand-alone projector, $M_0$, is simply the projector onto the support space of $B_0$, which is clearly $delta_{a_1,a_2,i_{12}}\delta_{a_2,a_3,i_{23}}\delta_{a_3,a_1,i_{31}}$. So we have proved that $M_0$ for triple-line TNR of general string-net is,
\begin{eqnarray}
M_0 = \delta_{a_1,a_2,i_{12}}\delta_{a_2,a_3,i_{23}}\delta_{a_3,a_1,i_{31}}.
\end{eqnarray}
This completes the proof.
\end{proof}

This is the projector on to the stand-alone space of triple-line TNR of general string-net models. For notational convenience we will denote these basis vectors as $| \lbrace \prod_{k=1}^3 \delta_{b_k,b_{k+1},i_{k,k+1}} \rbrace \rangle $, that is,
\begin{eqnarray}
| \lbrace \prod_{k=1}^3 \delta_{b_k,b_{k+1},i_{k,k+1}} \rbrace \rangle &=&  \delta_{b_1,b_2,i_{12}}\delta_{b_2,b_3,i_{23}}\delta_{b_3,b_1,i_{31}} \nonumber \\ & & |b_1,b_2,b_3;i_{12},i_{23},i_{31}\rangle
\end{eqnarray}
\begin{eqnarray}
\begin{centering}
| \lbrace \prod_{k=1}^3\delta_{b_k,b_{k+1},i_{k,k+1}\rangle \rbrace }=\raisebox{-15mm}{\includegraphics[scale=0.4]{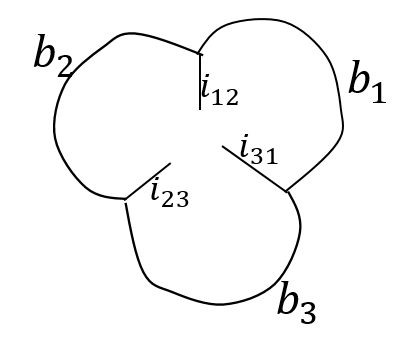}}
\end{centering}.
\label{triple-line stand-alone basis fig app}
\end{eqnarray}
So we get,
\begin{eqnarray}
dim(M_0)=\sum_{\substack{b_1,b_2,b_3;i_{12},i_{23},i_{31} }} \delta_{b_1,b_2,i_{12}}\delta_{b_2,b_3,i_{23}}\delta_{b_3,b_1,i_{31}}. \nonumber \\
\end{eqnarray}

\subsection{String-net MPO symmetries from Wilson-loop operators}
In last chapter, we showed how MPO symmetries or the MPO-injective subspace come from representation of anyonic Wilson-loops of the model on the stand-alone space. It is instructive to do the same with general string-net models. \par 
When an $f$-type simple string operator passes through the tensor $T^0$ on the physical level, it induces an operation on the virtual level in the way shown in Fig.~\ref{triple-line induction}. 

That is, it simply becomes a $f$-type string which is then fused with the plaquette legs. We consider the Wilson loop that encircles the 3 plaquettes of the tensor. This Wilson loop creates an $f$-type string that then fuses with the 3 plaquette loops.  Remember that we need to calculate the representation of this operator on the stand-alone space. That is, we need to calculate the matrix elements $ \langle  \lbrace \delta_{a_k,a_{k+1},i_{k,k+1}}|W_f| \lbrace \delta_{b_k,b_{k+1},i_{k,k+1}} \rbrace \rangle  $. So we imagine a tensor network in which the tensor in the stand-alone basis is surrounded by $T^0$. We now apply the Wilson loop encircling 3 plaquettes and calculate induced operator on the stand-alone basis. It can be done in a convenient way using string-net diagrams in Eq.~\eqref{triple-line stand-alone basis fig app}. 

There are essentially 3 steps: \par 
\textbf{Step 1:} Since the surrounding tensors are the fixed point tensor $T^0$, the Wilson loop on the physical level simply becomes an $f$-type string that fuses with the plaquette legs,
\begin{eqnarray}\label{triple-line induction}
\includegraphics[scale=0.4]{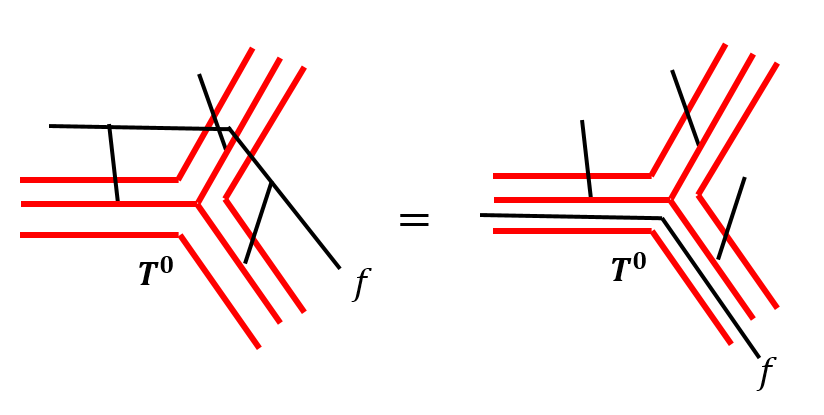}
\end{eqnarray}
 Since these plaquette legs are contracted with the plaquette legs of the stand-alone tensor, it is equivalent to fusing $f$-string loop with the 3 plaquette legs of the stand-alone tensor,
 \begin{eqnarray}
\includegraphics[scale=0.4]{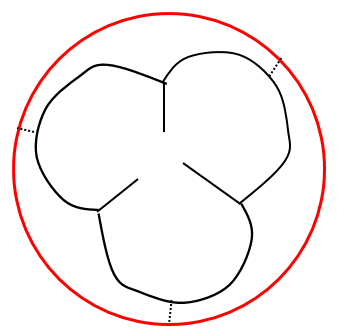}.
\end{eqnarray}
 \par 
\textbf{Step 2:} We fuse these strings with the three nearby strings $a_1,a_2,a_3$,
 \begin{eqnarray}
\includegraphics[scale=0.4]{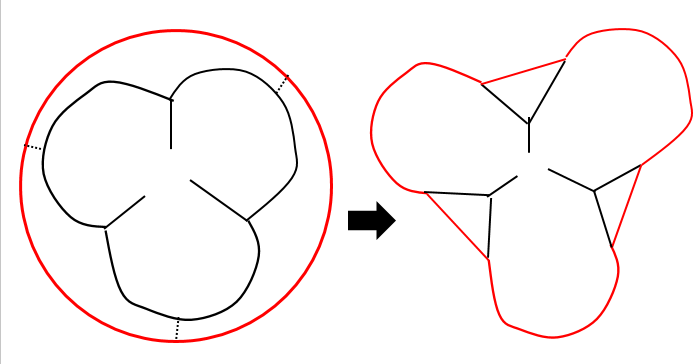}.
\end{eqnarray}
Let's say they fuse to make strings $b_1,b_2,b_3$. We gain factors $F^{a_j,a_j,0}_{f,f,b_j}= \frac{\sqrt{d_{b_j}}}{\sqrt{d_fd_{a_j}}}, \, j=1,2,3 $ for each fusion. \par 
\textbf{Step 3:} In the last step we remove the three bubbles created in the previous step,
 \begin{eqnarray}
\includegraphics[scale=0.4]{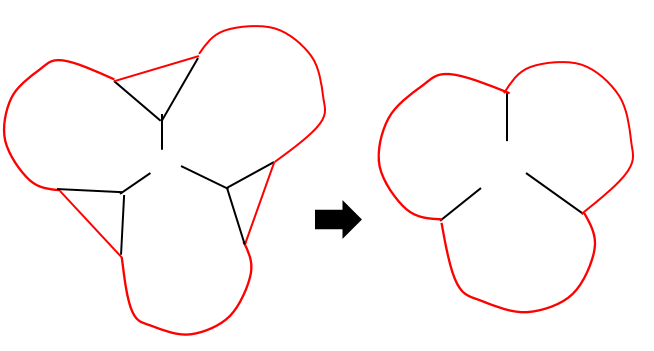}.
\end{eqnarray}
Each bubble removal produces a factor of $ \sqrt{d_fd_{a_j}d_{a_{j+1}}}G^{b_jb_{j+1}i_{j,j+1}}_{a_{j+1}a_jf}, \, j=1,2,3 $.

Collecting the factors from step 2 and step 3, we get 
\begin{eqnarray}
 \langle  \lbrace \prod_{k=1}^3\delta_{b_k,b_{k+1},i_{k,k+1}}|W_f| \lbrace  \prod_{k=1}^3\delta_{a_k,a_{k+1},i_{k,k+1}} \rbrace \rangle \nonumber \\ =\prod_{j=1}^3d^{\frac{1}{2}}_{b_j}d^{\frac{1}{2}}_{a_j}G^{b_jb_{j+1}i_{j,j+1}}_{a_{j+1}a_jf}.
\end{eqnarray}
This is the expression for $M_f=M_0W_fM_0$. Now considering the projector $\mathbb{M} = \sum_f \frac{d_f}{D}M_f $, we get 
\begin{eqnarray}
\mathbb{M} &=& \sum_f \frac{d_f}{D}\prod_{j=1}^3d^{\frac{1}{2}}_{b_j}d^{\frac{1}{2}}_{a_j}G^{b_jb_{j+1}i_{j,j+1}}_{a_{j+1}a_jf} \nonumber \\
&=& \frac{1}{D}d^{\frac{1}{6}}_{b_j}d^{\frac{1}{6}}_{a_j} G^{i_{23}i_{31}i_{12}}_{a_1a_2a_3} G^{i_{23}i_{31}i_{12}}_{b_1b_2b_3}.
\end{eqnarray}
It should be understood as an operator written in its components in the basis  $ | \lbrace a_k;i_{k,k+1} \rbrace \rangle  \langle \lbrace b_k;i_{k,k+1} \rbrace | $. 
We used pentagon identity in the second step. We can see that it projects on to the space with $G^{i_{23}i_{31}i_{12}}_{a_1a_2a_3}\neq 0 $, that is $\delta_{i_{23},i_{31},i_{12}} \neq 0 $.\par 
There is a small technical issue though. The factor $d^{\frac{1}{2}}_{b_j}d^{\frac{1}{2}}_{a_j}$ does not exactly match the factors in the $TT^{\dagger}$ support space given in Eq.~\eqref{vijk}). It is simply because we did not keep track of exactly how to distribute factors that share a plaquette while applying the Wilson loop. In fact, the Wilson loop around a single vertex is somewhat ill-defined. But we are only trying to get a symmetry condition on the individual tensors which makes sure that Wilson loop on a larger region is a symmetry of the state. We can show that this factor has to be exactly $d^{\frac{1}{6}}_{b_j}d^{\frac{1}{6}}_{a_j}$ if the Wilson loop is to be a symmetry of the state. The reason is simply, as concluded in the original string-net paper,  a Wilson loop commutes with the plaquette term $B_p=\sum_s a_s B_p^s$ only when $a_s=d_s$. In tensor network language, it translates to the fact that every tensor must contribute a factor of $d_s^{\frac{1}{6}}$ for the Wilson loop to be a symmetry. Also we know that an $f$-type Wilson loop applied to the ground state produces a factor of $d_f$. Combining all these we can write the exact Wilson loop operator on a single tensor as:
\begin{eqnarray}
M_f& =&  \prod_{j=1}^3d_{a_j}(d_{b_j}d^{-1}_{a_j})^{\frac{1}{6}} G^{b_jb_{j+1}i_{j,j+1}}_{a_{j+1}a_jf} \\
\Rightarrow \mathbb{M} &=& \frac{1}{D}d_{a_j}(d_{b_j}d^{-1}_{a_j})^{\frac{1}{6}} G^{i_{23}i_{31}i_{12}}_{a_1a_2a_3} G^{i_{23}i_{31}i_{12}}_{b_1b_2b_3}.
\label{Mf single tensor}
\end{eqnarray}
The fixed point triple-line tensor satisfies 
\begin{eqnarray}
M_fT^0 &=& d_f T^0, \\
\mathbb{M} T^0 &=& T^0.
\end{eqnarray}
One can check that $\mathbb{M} =\sum_f \frac{d_f}{D}M_f$ is indeed a projector and it projects onto the support space of $TT^{\dagger}$. \par 
Finally, just like boundary operators $B_f$, $f$-type MPO can be extended to an arbitrary large region as 
\begin{eqnarray}
M_f(\partial R)&=&  \prod_{j=1}^n  G^{b_jb_{j+1}i_{j,j+1}}_{a_{j+1}a_jf}d_{a_j} (d_{b_j}d_{a_j}^{-1})^{\frac{\theta_j}{2\pi}}.  
\label{Mf expression}
\end{eqnarray}
and it represents the operation induced on the virtual level by a Wilson loop applied on the boundary of the region $R$.
\section{0-type string operator is a zero-string operator of triple-line TNR}
\label{0-type string operator section}
In last chapter we argued how the reason for instability is that some of the non-trivial anyon operators might have a trivial representation on the virtual level. That is, they disappear identically on the ground state tensor network, even in the presence of a topological hole. We saw that, for the single-line and double-line TNR of toric code, $X$-string  and $Z$-string operators were the zero-string operators respectively. Indeed, the general string-net also has such an operator. These are the operators that only has string-type 0 in it. Remember that string-operators on the string net model act by adding a string-type (possibly more than one) to  the string-net and then fusing it with the string-net by some fusion rules.   
  The expression of Wilson-loop operators can be used to see how a string operator with open ends would act on the tensors along the path. It would look the same as in Eq.~\eqref{Mf expression} along the path with some changes at the end. But we don't worry too much about the details of how this operator looks at its ends, since those details can always be changed using local unitary operators at its ends.  Looking at the Wilson-loop operators in Eq.~\eqref{Mf expression}, it is immediately clear what the invisible string operators are for the triple-line TNR of general string-nets. For $f=0$, (using  identity \eqref{Gnorm})
\begin{eqnarray}
M_0(\partial R) &=& \prod_{j=1}^n \left( G^{b_jb_{j+1}i_{j,j+1}}_{a_{j+1}a_j0}d_{b_j} (d_{a_j}d_{b_j}^{-1})^{\frac{m_j}{6}} \right) \nonumber \\ 
&=&  \prod_{j=1}^n \left( \delta_{a_j,b_j} (d_{a_j}d_{b_j})^{-\frac{1}{2}}d_{b_{j}} (d_{a_j}d_{b_j}^{-1})^{\frac{m_j}{6}} \delta_{a_j,a_{j+1},i_{j,j+1}}\right) \nonumber \\
& = &   \prod_{j=1}^n \delta_{a_j,a_{j+1},i_{j,j+1}}.
\end{eqnarray}
But the final expression is the very definition of stand-alone space itself. It means this operator will act trivially on the stand-alone space. So, a $0$-type simple string operator is a non-trivial invisible string operator, that is, it is a zero-string operator. From this it should be clear why we denoted the stand-alone space $M_0$ and why we called non-trivial invisible string-operators zero-string operators. These names come from the general string-net formalism.  \par 
It is also clear that for $f\neq 0$, $M_f$  acts necessarily non-trivially on the tensors along the path. One should carefully note that, though non-zero-string operators change tensors along the path, it does not mean that this path is a physical observable. These paths can always be deformed as $M_f$ passes through $T^0$ without any phase accumulation. It is called the `pulling-through condition' \cite{csahinouglu2014characterizing}. When there is an MPO violating variation present at a tensor, $M_f$ cannot be pulled through it. Hence our conjecture can be alternatively worded as `the stand-alone variations which prohibit  the pulling-through property of fixed point tensors cause instability.'
\section{Proof of theorem \ref{thm:sninstability}}
 \label{sninstability proof}
 We will give an analytical proof of why all string-net triple-line TNR have at least one unstable direction which comes from the $M_0- \mathbb{M}$ subspace. We will do so by directly calculating $S_{\text{topo}}(\epsilon)$. 

\subsection{ Topological entanglement entropy on a cylinder with non-RG fixed point tensor}
\begin{mylemm}\label{Result 3}
 let's say we divide the cylinder in two halves (Fig.~\ref{algo0}(a)). We denote the right half as $R$. If any given tensor network on this cylinder satisfies, 
\begin{eqnarray}
\lim_{|R| \rightarrow \infty } \mathbb{T}(R) = C^{|R|} \sum_f c_f B_f(\partial R), 
\end{eqnarray}
where $C$ is some constant, then, $S_{\text{topo}}$, as given in Eq.~\eqref{Stopomethod}, is 
\begin{eqnarray}
S_{\text{topo}} = \log\sum_f \left( \frac{c^2_f}{c^2_0}\right).
\end{eqnarray}
\end{mylemm}
\begin{proof}
The proof is quite simple. We follow the same steps as used in the proof of lemma \ref{trace Bf zero}, replacing $d_f$ with $c_f$. We first calculate the normalization of the density matrix. 
\begin{eqnarray}
\text{Tr}(\sigma_b^2)&=&\text{Tr} (\sum_fc_fB_f)^2 \nonumber \\
&=& \sum_{f,f'}c_fc_{f'} \text{Tr}(B_fB_{f'}) \nonumber \\ 
&=& \sum_{f,f'} c_fc_{f'} \delta_{f,f'} D^n \nonumber \\
&=& D^n(\sum_f c_f^2).
\end{eqnarray}
By calculating Renyi entropy with renyi index $\alpha =1/2$, we get
\begin{eqnarray}
S_{1/2}(\rho_R) &=& \frac{1}{1-1/2}\log \text{Tr}(\rho_R^{\frac{1}{2}}) \nonumber \\
&=& 2\log  \frac{\text{Tr}(\sum_fc_fB_f)}{\sqrt{D^n\sum_fc_f^2}}\nonumber \\
&=& -n\log D-2\log\sum_f(d_f\textrm{Tr}B_f)- \log\sum_f c_f^2 \nonumber \\
&=& -n\log D+2n\log\lambda_0 - 2\log\left(1+\sum_{f>0} c_f\frac{\textrm{Tr}B_f}{\lambda_0^n} \right) \nonumber \\
& & +2\log c_0- \log\sum_f c_f^2. \nonumber 
\end{eqnarray}
When we let $n \rightarrow \infty$ and using Eq.~\eqref{trace Bf zero} 
\begin{eqnarray}
S_{1/2}(\rho_R) &=& n\log\frac{\lambda_0^2}{D}-\log\sum_f \left( \frac{c^2_f}{c^2_0}\right)\\
\Rightarrow S_{\textrm{topo}} &= &\log\sum_f \left( \frac{c^2_f}{c^2_0}\right)\label{Stopo cf}.
\end{eqnarray}
This completes the the proof.

\end{proof}
 Finally we are ready to show the unstable tensor perturbations in the triple line TNR of the string-net models.

\par 
\subsection{Instability in string-net }
Now we give proof of theorem \ref{thm:sninstability}. 
\begin{proof}
 Combination of the lemma \ref{Result 1}, theorem \ref{Result 2} and lemma \ref{Result 3} gives a clue to why $T^0 \rightarrow T^0+ \epsilon T^q$, are unstable variations.  We will choose particular variations in $M_0-\mathbb{M}$ for analytical simplicity, but it should be understood that any arbitrary variation that has a component in those directions will result in instability. 
We discussed in section \ref{tensors in unstable space} that there are two kinds of variations in $M_0-\mathbb{M}$; vertex variations (that violate the vertex term) and plaquette variations (that violate the plaquette term). We will treat them one by one. \par 
Before we do any analytical calculation, let us describe in simple words what the reason for instability is. We saw in the proof of lemma \ref{Result 1} that as fixed point tensors contract, every face, every edge, and every vertex contribute a factor of $d_f$, $d_f^{-1}$, and $d_f$ respectively. It can be visualized like this,
\begin{equation} \label{chiR}
    \includegraphics[width=\columnwidth]{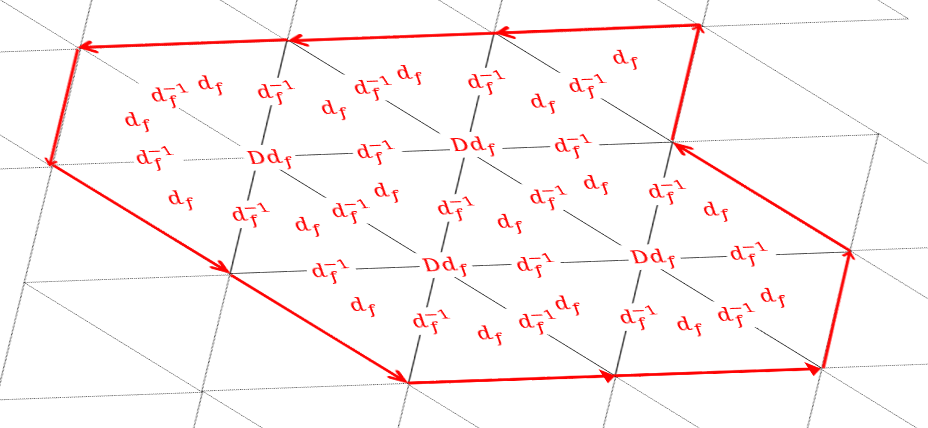}
\end{equation}
It combines to give $c_f=d_f^{F-E+V}=d_f^{\chi_R}$ which is a topological invariant of the lattice. If a tensor variation changes the double tensor in such a way that one of these factors (face, edge or vertices) are changed, even infinitesimally, then the $c_f$ we get is not a topological invariant, and $S_{\textrm{topo}}$ due to lemma \ref{Result 3} changes. We will now show that this is precisely what variations in $M_0-\mathbb{M}$ do. In particular, the vertex variations change the vertex factors, and the plaquette variations change the face factors. \par  
Let's choose a particular tensor variation
\begin{eqnarray}
T^q&=&  \prod_{j=1}^3 (d_{i_{j,j+1}}^{\frac{1}{4}})\delta_{a_jb_{j+1}i_{j,j+1}},
\end{eqnarray}
such that
\begin{eqnarray}
\mathbb{T}^q= T^q(T^q)^{\dagger}= B_0. 
\end{eqnarray}
This tensor is supported on the full $M_0$ space and clearly has components outside the MPO-injective subspace because as we showed $M_0>\mathbb{M}$. So $(M_0-\mathbb{M})T^q\neq 0$. Now, the double tensor for the varied tensor is 
\begin{eqnarray}
\mathbb{T}=(T^0+\epsilon T^q)(T^0+\epsilon T^q)^{\dagger } \approx \mathbb{T}^0+\epsilon^2B_0\\
=(1+\epsilon^2)B_0+\sum_{f>0}d_fB_f.
\end{eqnarray}
We have ignored the linear terms in $\epsilon$ as they are contained within the MPO-injective subspace, and we don't need to worry about them. This double tensor will contract with itself in exactly the same way as $\mathbb{T}^0$ did, but the only difference is, now every face will contribute a factor of $r_f$, where, $r_0=(1+\epsilon^2)$, and $r_{f>0}=d_f$. The vertex factors and edge factors will remain to be $d_f$ and $d_f^{-1}$, respectively. After contracting it on a large region we will get a double tensor $\mathbb{T}(R) = \sum_f c_f B_f (\partial R) $, where $c_f=r_f^Fd_f^{V-E}$. So $c_0=(1+\epsilon^2)^F$ and $c_{f>0}=d_f^{\chi_R}$. So we see that $c_0$ is exponentially larger than $c_{f>0}$ even for an infinitesimal $\epsilon$, hence, using Eq.~\eqref{Stopo cf}, $S_{\textrm{topo}}=0$. \par 
Now we look an example of plaquette variations. Consider tensors that are exactly the same as the fixed point tensors, except the plaquette factors $d_a^{1/6}$ are replaced by a factor of $(d_a+\epsilon s_{q;a})^{1/6}$, where $s_{q;a}$ is the $a$th component of the $q$th eigenvector of $\delta$, as explained in Eq.~\eqref{sq}.
\begin{align}
(T^q) =  \prod_{j=1}^3 \left(d_{i_{j,j+1}}^{\frac{1}{4}}(d_{a_j}+\epsilon s_{q;a_j})^{\frac{1}{6}} \right) G^{i_{23}i_{31}i_{12}}_{a_1 a_2 a_3}
\end{align}
This tensor is clearly supported on the stand-alone space, and is outside the MPO-injective subspace as to be inside the MPO-injective subspace it has to have  $d_a^{1/6}$ factors. The double tensor will again produce a factor of $d_f$ on the faces, and $d_f^{-1}$ on the edges upon contraction. But now the factors on the vertices would be
\begin{eqnarray}
\sum_{a,b}\delta_{a,b,f }(d_a+\epsilon s_{q;a})(d_b+\epsilon s_{q;b})=D(d_f+e_{q;f}\epsilon^2), \nonumber \\
\end{eqnarray}
where $s_{q}$ is normalized to give $\langle s_q|s_q \rangle=D$ and $e_{q;f}$ is the $q$th eigenvalue of the matrix $N^f_{a,b}=\delta_{a,b,f}$. A conclusion similar to that for vertex variation case follows. $c_f = d_f^{F-E}(d_f+\epsilon^2 e_{q;f})^{V}=d_f^{\chi_R}(1+ \epsilon^2 \frac{e_{q;f}}{d_f} )^V $ is not a topological invariant, as it extensively depends on the number of vertices $V$. As a result, the weight of one of the boundary operator in  $\mathbb{T}=\sum_f c_f B_f$ becomes exponentially larger than the others even for an infinitesimal variation $\epsilon$, and hence the topological order is lost. \par Result I-IV  together  complete the proof that  triple-line TNR of general string-net states have at least one unstable direction.  
 
\end{proof}
\section{Dependence of $S_{\text{topo}}$ on boundary conditions in cylindrical geometry}
\label{boundary issue}
\begin{figure}
\includegraphics[trim=10mm 40mm 10mm 10mm,width=6cm]{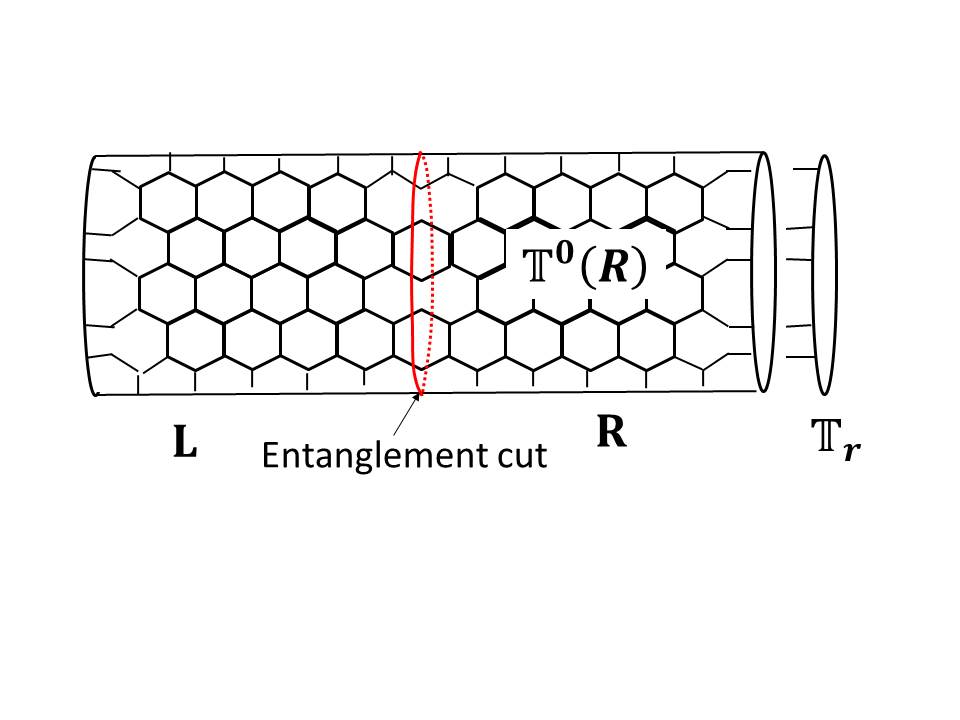}
\caption{We calculate entanglement entropy of the right-half of the cylinder with a certain boundary condition $\mathbb{T}_r$. The entanglement cut is in the middle of the cylinder.  }
\label{boundary appendix fig}
\end{figure}
Topological entanglement entropy calculation is done by calculating the entanglement entropy of a subsystem $A$. When the boundary of $A$ consists of topologically trivial loops, for example when $A$ has a disc geometry, $S_{\text{topo}}$ is known to depend only on the total quantum dimension $D$, $S_{\text{topo}} = \log D$. However when the boundary of $A$ consists of non-contractible topologically non-trivial loops, for example when a torus or cylinder is divided into two cylinders, it has been shown by \citet{Zhang12} that $S_{\text{topo}}$ also depends on the linear combination of ground states.  For a ground state wave function on a torus
\begin{eqnarray}
|\Psi\rangle = \sum_a c_a |\Xi_a\rangle
\end{eqnarray}
where the sum is over the degenerate ground states labeled by quasi-particles of the model,
the $n$th R\'enyi entropy is given by 
\begin{eqnarray}
S_{n} &=& \alpha_nL - S_{\text{topo}} \label{MESSn},\\ 
S_{\text{topo}}&=& 2\log D - \frac{1}{1-n}\log \left( \sum_a p_a^nd_a^{2(1-n)} \right)\nonumber \\
\end{eqnarray}
where $d_a$ is the quantum dimension of $a$th quasi-particle and $p_a=|c_a|^2$. $|\Xi_a\rangle$ are special basis for which $S_{\text{topo}}$ is maximal, or entanglement entropy is minimal. These states are called the \textit{Minimum Entropic States (MES) }. It was shown that MES correspond to eigenstates of Wilson-loop operators along the entanglement cut. \par 
This dependence of $S_{\text{topo}}$ on the ground state is of crucial importance to us since we have used cylinder with a boundary for $S_{\text{topo}}$ calculations. So, numerically obtained $S_{\text{topo}}$ contain information about the boundary as well. For example, consider the toric code. 
\begin{eqnarray}
S_{\textrm{topo}}= 2\log 2 - \frac{1}{1-n}\log(p_1^n+p_2^n+p_3^n+p_4^n) \nonumber
\end{eqnarray}
when $p_1=p_2=p_3=p_4=\frac{1}{4}$ we get $S_{\textrm{topo}}=0$ although the the topological order is not lost. So one has to be careful using $S_{\text{topo}}$ as an indicator of topological order. \par 
\begin{figure}
\includegraphics[trim=0mm 20mm 50mm 10mm,width=7cm]{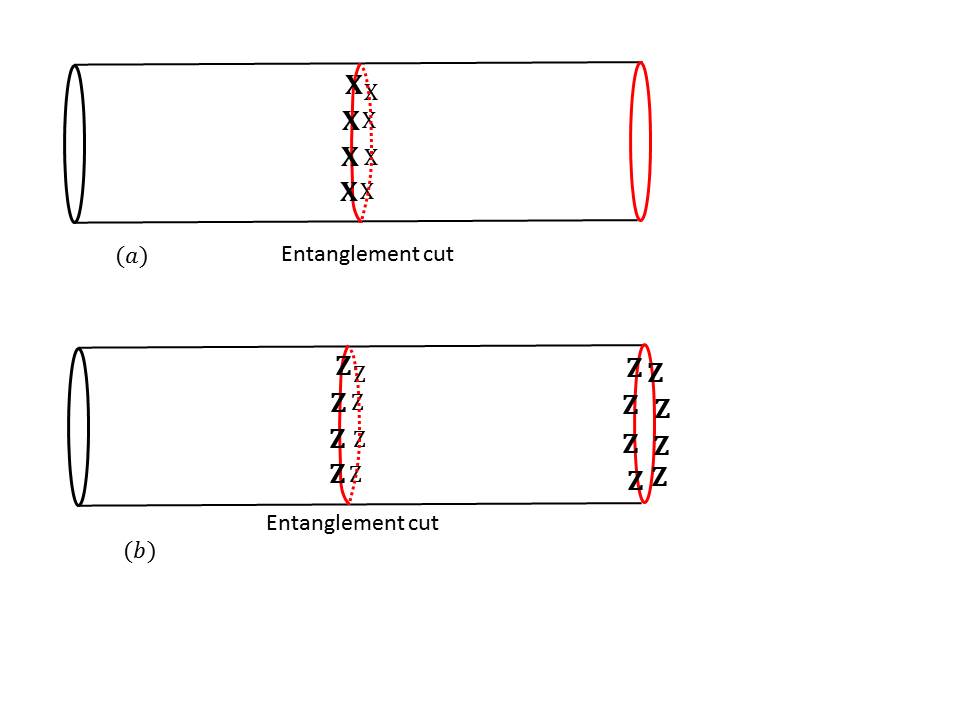}
\caption{MESs are eigenstates of different Wilson loop operators at the entanglement cut. (a)For fixed point single-line TNR, the state on the cylinder is always in +1 eignestate of $X$-loop, as it identically disappears. (b)The state is also in +1 eigenstate of simultaneous operation of two $Z$-loops, one at the entanglement cut, other at the right-most boundary. It implies, we can be in two MESs depending on the the boundary tensor choice. If the boundary tensor is in +1 eigenstate of the boundary $Z$ loop, then the state is in +1 eigenstate of the entanglement-cut $Z$ loop. Similarly, if the boundary tensor is in -1 eigenstate of the boundary $Z$ loop, then the state is in -1 eigenstate of the entanglement-cut $Z$-loop.}
 \label{TC boundary}
\end{figure}
Let's first take the example of the single-line TNR of the toric code. See Fig.~\ref{TC boundary}. We put our system on a cylinder with some boundary conditions to be determined later. The entanglement cut is in the middle of the cylinder, and the right half cylinder, denoted as $R$, is the subsystem whose entanglement entropy we are calculating (see Fig.~\ref{boundary appendix fig}).  The four MES correspond to four eigen states of $e$ and $m$ Wilson-loops on the entanglement cut. But, since $e$-Wilson loop is a zero-string operator, the state is always in its $+1$ eigenstate (Fig.~\ref{TC boundary}(a)). So we have access to only two MES corresponding to $\pm 1$ eigenstates of $m$-Wilson loop. We also know that the state is in $+1$ eigenstate of the $Z^{\otimes }_{\partial R}= Z^{\otimes L}_{ec} \otimes Z^{\otimes L}_{r}$, where subscript $ec$ stands for loop at entanglement cut, and $r$ stands for the loop at the right boundary of $R$. Since the state is in +1 eigenstate of $Z^{\otimes L}_{ec} \otimes Z^{\otimes L}_{r}$ (see Fig.~\ref{TC boundary})(b), the state can be either in $+1$ eigen-state of both $ Z^{\otimes L}_{r}$ and $ Z^{\otimes L}_{ec} $ or in $-1$ eigenstate of the both. The boundary tensor determines which eigenstate of $Z^{\otimes L}_{r}$ the wave function is in, and consequently also which eigenstate of $Z^{\otimes L}_{ec}$. This is how the boundary tensors and MES are connected. Since we have access to only two MES
\begin{eqnarray}
S_{\text{topo}}= \log 2 - \frac{1}{1-n}\log (p_1^2+p_2^2) \label{MESTC}.
\end{eqnarray}
A similar analysis follows in the double-line TNR, with the role of $e$ and $m$ Wilson loop operators reversed: now the state is always in the $+1$ eigen state of $m$-Wilson loop and the two MES correspond to the two eigenstates of $e$ Wilson loop at the entanglement cut, which in turn depends on the boundary tensors. \par 
We saw in the Appendix \ref{Result 3} $\rho_R = N\sigma_b^2$ where 
\begin{eqnarray}
\sigma_b = \mathbb{T}^0(R) \mathbb{T}_r
\end{eqnarray}
where $\mathbb{T}_r$ denotes the double tensor on the boundary. We know that, up to an irrelevant normalization constant,
\begin{eqnarray}
\mathbb{T}^0(R) &=& \sum_f d_f^{\chi_R}B_f(\partial R) \nonumber \\
&=& (B_0)_{ec}\otimes (B_0)_{r}+  (B_1)_{ec}\otimes (B_1)_{r},
\end{eqnarray}
where $B_0=I^{\otimes L}$ and $B_1=Z^{\otimes L}$ for the single-line TNR and $B_1=X^{\otimes L}$ for the double-line TNR. Let's say the boundary double tensor $\mathbb{T}_r$ contracts with $(B_f)_r$ to produce the constants $c_f$ (see Fig.~\ref{boundary cf fig}) 
\begin{figure}
\includegraphics[trim=0mm 100mm 50mm 10mm,clip,width=7cm]{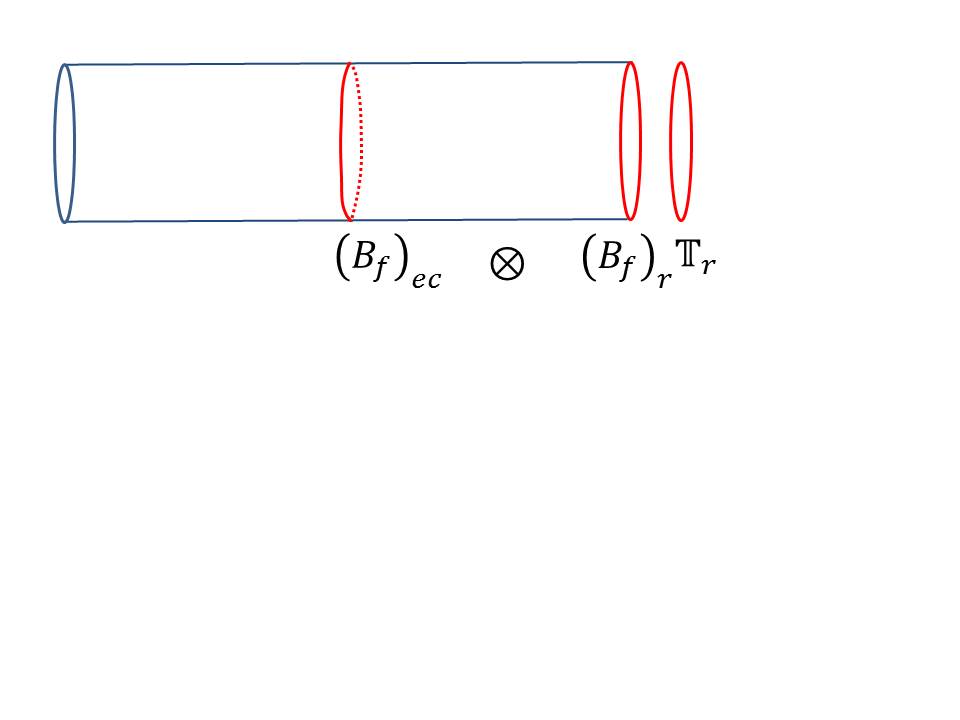}
\caption{ Bulk double tensor is a sum of tensor product between $(B_f)_{ec}$ ($B_f$ on the entanglement cut) and $(B_f)_r$ ($B_f$ on the right boundary). So, when we contract a boundary tensor $\mathbb{T}_r$ with the bulk tensor, it contract with $(B_f)_r$ giving a scalar $c_f$. So resulting tensor is $\textrm{Ev}(\mathbb{T}(R)\mathbb{T}_r)=\sum_fc_f (B_f)_{ec} $. Consequently, $S_{\textrm{topo}}$ using Eq.~\eqref{Stopo cf} is simply $\log\left(\sum_f \frac{c_f^2}{c_0^2}\right) $. }
\label{boundary cf fig}
\end{figure}
\begin{eqnarray}
\sigma_b&=& \left((B_0)_{ec}\otimes (B_0)_{r}+  (B_1)_{ec}\otimes (B_1)_{r} \right) \mathbb{T}_r \nonumber \\
&=& c_0 (B_0)_{ec}+c_1(B_1)_{ec} \nonumber \\
&=&  c_-B_-+c_+B_+.
\end{eqnarray}
where $c_0=(B_0)_r\mathbb{T}_r, \, c_1=(B_1)_r\mathbb{T}_r$ and $B_{\pm }= \frac{1}{2}(B_{0} \pm B_1)$ and $c_{\pm }= (c_0 \pm c_1)$. Note that $B_{\pm}$ satisfy the following, 
\begin{eqnarray}
B_{\pm}^2 = B_{\pm}, \quad \text{Tr}(B_{\pm}) = 2^{L-1}.
\end{eqnarray}
With this, we get the normalized density matrix as,
\begin{eqnarray}
\rho_R &=& \frac{1}{2^L}\left( \frac{c_-^2}{c_-^2+c_+^2}B_-+ \frac{c_+^2}{c_-^2+c_+^2}B_+ \right)\\
&=& \frac{1}{2^L} \left( p_- B_- + p_+B_+ \right).
\end{eqnarray}
The $n$th Renyi entropy is, 
\begin{eqnarray}
S_n(\rho_R) &=& \frac{1}{1-n}\log \text{Tr}(\rho_R^n) \nonumber \\
&=& \frac{1}{1-n}\log \text{Tr}\left( \frac{1}{2^{nL}}(p_-^nB_-+p_+^nB_+ ) \right) \nonumber \\
&=& \frac{1}{1-n}\log\left( \frac{1}{2^{nL}}(p_-^n2^{L-1}+p_+^n2^{L-1} ) \right) \nonumber \\ 
&=& L\log 2- \left(\log 2- \frac{1}{1-n}\log (p_-^n+p_+^n)  \right).\nonumber  \\
\end{eqnarray}
Comparing it with the MES formula in Eq.~\eqref{MESTC}, we see that $p_1=p_-=c_0-c_1$ and $p_2=p_+=c_-+c_+$. So the state is an MES if $p_{\pm}= 0 \Rightarrow c_0=\pm c_1$ for which we get maximal topological entanglement entropy, $S_{\text{topo}}=\log 2 $. This illustrates  the direct dependence of $S_{\text{topo}}$ on  $\mathbb{T}_r$. \par 
\begin{figure}
\includegraphics[trim=30mm 30mm 30mm 20mm,clip, width=9cm]{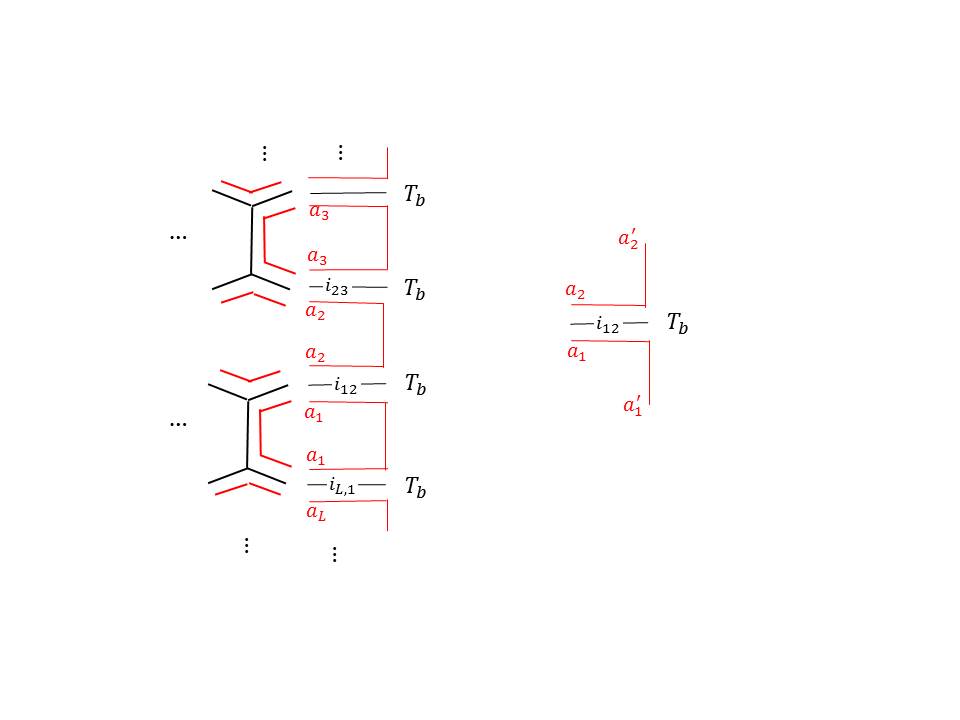}
\caption{ Smooth boundary condition for triple-line tensor network. Tensors $T_b$ are used on the boundary. $T_b$ has 5 virtual legs, $a_1,a'_1,a_2,a'_2,i_{12}$ and 1 physical leg, $i_{12}$. Physical leg and the middle leg take the same values. We assign a particular value to the components of this tensor, $(T_b)^{i_{12}}_{i_{12}a_1a'_1;a_2a'_2}=\delta_{i_{12},0} \delta_{a_1,a'_1}\delta_{a_2,a'_2}\delta_{a_1a_2i_{12}} $. \label{Tb}}
\end{figure}
Of course the above analysis is done for the RG fixed point tensors only. We have to choose a boundary double tensor $\mathbb{T}_r$ such that $S_{\text{topo}}$ is truly indicative for topological order, or lack of it, for both RG fixed point and varied tensors. We choose the following boundary tensor for our numerical calculations: For any tensor network, fixed point or varied, We use a `smooth boundary condition'. It is explained in the Fig.~\ref{Tb}. First we will explain it for the triple-line tensors. For double-line and single-line an appropriately reduced version of $T_b$ will be used. Note that we haven't drawn the physical index explicitly and it should be understood the same as the middle index (the index in black color). So the boundary tensor $T_b$ has four virtual indices, and we fix its components to be,
\begin{eqnarray}
(T_b)^{i_{12}}_{i_{12}a_1a'_1;a_2a'_2}=\delta_{i_{12},0} \delta_{a_1,a'_1}\delta_{a_2,a'_2}\delta_{a_1a_2i_{12}}
\end{eqnarray}
that is, we put the physical/middle index to zero (vacuum) and allow the plaquette legs to vary with this restriction. For double-line we don't have a middle leg, but we can simply put the physical leg to 0. For single-line we only have the middle legs and we put them to zero.  \par 
Before we discuss why we choose this particular boundary, let us calculate what $S_{\textrm{topo}}$ we are supposed to get with this particular choice of boundary tensor. For that, we need to calculate $c_f= B_f\mathbb{T}_r$. Note that $\delta_{a_j,a_{j+1},0}$ implies $a_j=a_{j+1}$. So the double tensor $\mathbb{T}_r$ is
\begin{eqnarray}
\mathbb{T}_r = \sum_{a,b} |a,a,a..;000..\rangle\langle b,b,b,...;000...|.
\end{eqnarray}
So
\begin{eqnarray}
c_f &=& \textrm{Ev}(B_f\mathbb{T}_r) \nonumber \\
&=& \sum_{a,b} \prod_{j=1}^m G^{b,b,0}_{a,a,f}(d_ad_b)^{\frac{1}{2}} \nonumber \\
&=& \sum_{a,b} \prod_{j=1}^m \delta_{a,b,f}\nonumber \\
&=& \sum_{a,b} \delta_{a,b,f}. \label{cf of Tp}
\end{eqnarray}
Then  using Eq.~\eqref{Stopo cf}, $S_{\textrm{topo}}$ is simply $\log(\sum_f \frac{c_f^2}{c_0^2}) $. For the toric code, and double semion models $c_0 = c_1=2$, so we get $S_{\textrm{topo}}=\log 2$. For the double Fibonacci model, however, we get \begin{eqnarray}
c_0 = \sum_{a,b}\delta_{a,b,0}= \delta_{0,0,0}+\delta_{1,1,0}=2 \\
c_1 = \sum_{a,b}\delta_{a,b,1}= \delta_{1,0,1}+\delta_{0,1,1}+\delta_{1,1,1}=3. \\
\end{eqnarray}
So we get $S_{\textrm{topo}}=\log(1+\frac{3^2}{2^2})= \log(1+\frac{9}{4})$, which is consistent with our numerical result.
 \par 
There are mainly two reasons why we choose this particular boundary condition \\
1- This is a very simple boundary condition which gives us a precise analytical value of the topological entanglement entropy (namely,  $\log(\sum_f \frac{c_f^2}{c_0^2}) $, with $c_f$ given in Eq.~\eqref{cf of Tp}) against which numerical calculations can be checked.  \\
2- Though situation for non-abelian cases is more complicated, this boundary is definitely MPO symmetric for abelian models. That is, we expect the tensor network state to be an MES with maximal $S_{\textrm{topo}}$ (=$\log D$). 
\par 
\begin{figure}
\includegraphics[trim=10mm 0mm 10mm 10mm, width=7cm]{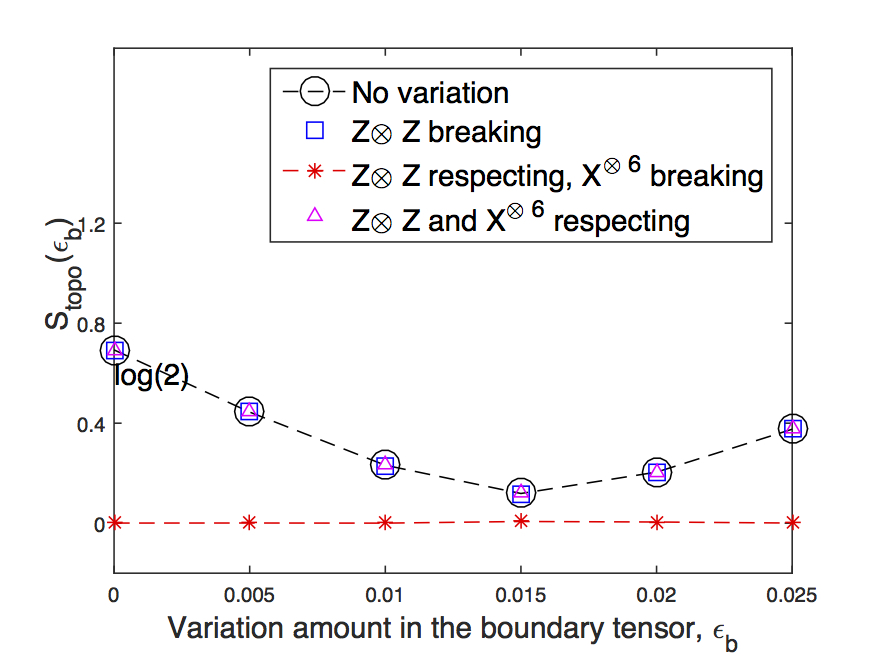} 
\caption{Dependence of $S_{\textrm{topo}}$ on boundary condition for toric-code double line TNR. We start with the boundary tensor, $T_b$, shown in Fig.~\ref{Tb}. We add a random variation $\epsilon_b T_b^r$ to $ T_b$ and calculate $S_{\textrm{topo}}(\epsilon_b )$ for random bulk variations in different subspaces. We keep $T_b^r$ fixed and increase the variation strength $\epsilon_b$. We see all classes of stable bulk variations  have the same $S_{\textrm{topo}}$ for each $\epsilon_b$ as the fixed point (no-variation) tensor. And the unstable class of bulk variation shows no dependence on $\epsilon_b$. It shows that stable variations indeed are in the same topological phase as the RG fixed point state, and unstable variation is a trivial phase. }
\label{DLTC boundary dependence}
\end{figure}
Numerical calculations of $S_{\textrm{topo}}$ will be checked against the analytical result in Eq.~\eqref{cf of Tp}. Now the remaining question is about the trustworthiness of the same calculation for varied tensor. That is, how can we deduce the conclusion about the topological order of the varied tensor by $S_{\textrm{topo}}(\epsilon)$? First point is, if $S_{\textrm{topo}}(\epsilon)= S_{\textrm{topo}}(0) $, then we can definitely say that the state is in the same topological phase. But $S_{\textrm{topo}}(\epsilon)=0$ needs to be further verified as it might be because of the particular boundary conditions imposed. To verify, we will test for $S_{\text{topo}}$ dependence on infinitesimal variation on the boundary tensors. The reason for this is clear by looking at the dependence of $S_{\text{topo}}$ on $p_1,p_2$ etc. So, if the state indeed has a topological order, $S_{\text{topo}}$ should sensitively depend on the $c_0= (B_0)_r\mathbb{T}_r, c_1=(B_1)_r\mathbb{T}_r$. If the state has lost its topological order, $S_{\text{topo}}$ will remain zero under any changes of the boundary tensor. This way, we can avoid getting any `accidental $S_{\text{topo}}=0$' cases, for example when $p_1=p_2=\frac{1}{2}$. \par 
One such verification is shown in Fig.~\ref{DLTC boundary dependence}. We first fix the boundary tensor to be $T_b$ given in Eq.~\eqref{Tb} and calculate the $S_{\textrm{topo}}$ for variations in $I_V-M_0,M_0-\mathbb{M}$ and $\mathbb{M}$ subspaces added to the fixed point bulk tensor. Now we add an infinitesimal random variation to the boundary tensor, $T_b\rightarrow T_b+\epsilon_b T_b^r$. $\epsilon_b$ (different from $\epsilon$, which the bulk variation strength) is the strength of the boundary variation. We increase $\epsilon_b$  slowly and for each value of the $\epsilon_b$ we calculate $S_{\textrm{topo}}(\epsilon)$ for random bulk variations in different subspaces. Fig.~\ref{DLTC boundary dependence} shows $S_{\textrm{topo}}$ as a function of $\epsilon_b$ for bulk variations in different subspaces. (the bulk variation strength $\epsilon$ is kept fixed throughout). We observe that \\
1- The variations which are unstable (i.e.  $S_{\textrm{topo}}=0$) for $T_b$, continue to be unstable for $T_b+\epsilon_b T_b^r$ for all values of $\epsilon_b $. It implies that we get $S_{\textrm{topo}}=0$ for these variation because the bulk topological order is indeed destroyed and not because of a specific boundary tensor chosen which gave an accidental zero. \\
2-The variations which are stable (i.e.  $S_{\textrm{topo}}=\log 2$) for $T_b$, have the same value of $S_{\textrm{topo}}$ as the fixed point tensor for all boundary tensors. It implies that tensor network state with these variations indeed have the same topological order as the fixed point tensor network state. 
Though this verification is shown for double-line toric code only, we find the same behavior for all numerical calculations presented in this paper.

It should be noted that any strictly positive value of $S_{\text{topo}}$ (assuming sufficiently large cylinder was considered) is a sufficient condition for topological order but it is not a necessary condition. So all we need to do is to  avoid getting accidental zeros.

\section{ Details of numerical calculations}
\label{numerical data}
\begin{figure}
\includegraphics[width=\columnwidth]{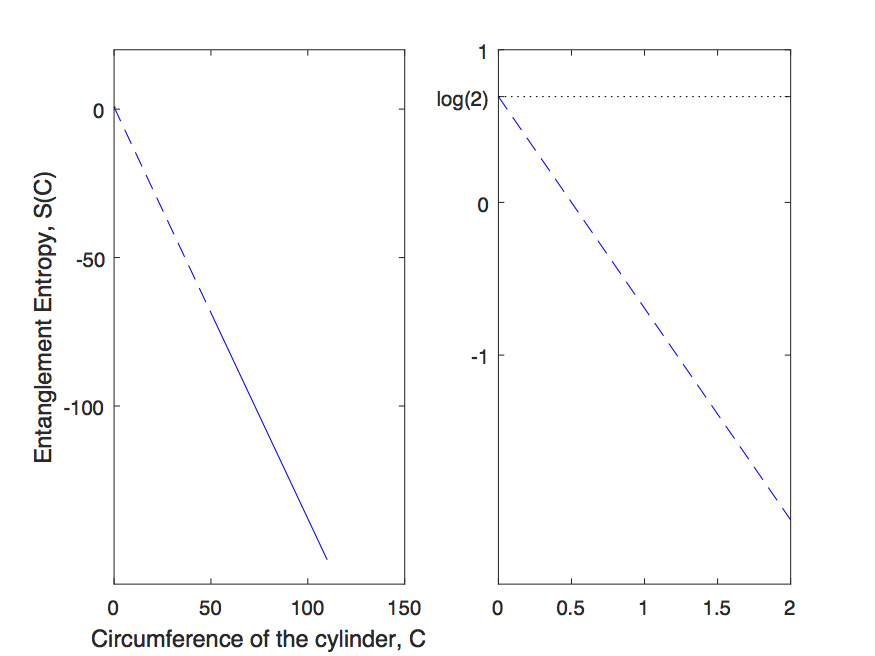}
\caption{Calculation of $S_{\textrm{topo}}$ for single-line toric code fixed point tensor network state. We fix half cylinder length as $L=500$. Circumference is varied from 50 to 110. $S$ varies linearly with $C$. This line is extrapolated back to $C=0$. Its intersection with the y-axis gives $S_{\textrm{topo}}$.  Right figure is a zoomed in version of the left figure to show the intersection point clearly. We find $S_{\textrm{topo}} \approx \log(2)$ }
\label{convgF1}
\end{figure}
\begin{figure}
\includegraphics[width=\columnwidth]{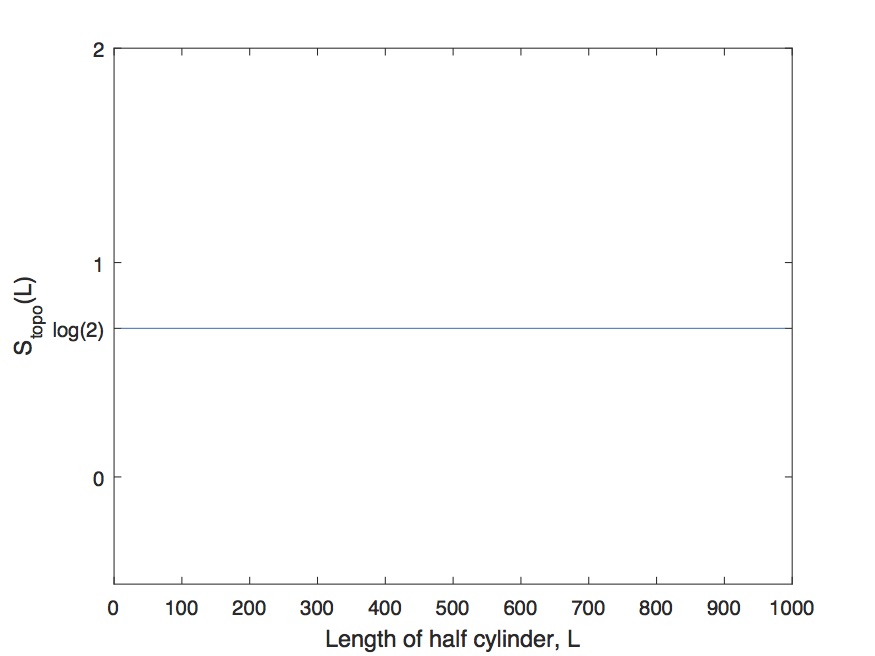}
\caption{ $S_{\textrm{topo}}$ was calculated for a fixed half cylinder length, $L=500$, in Fig.~\ref{convgF1}. We now vary $L$ from 10 to 1000. We see that $S_{\textrm{topo}}$ is  converged even for small values of $L$. So one does not need to large cylinder length to get the right $S_{\textrm{topo}}$ value. It is expected as it is an RG fixed point tensor network state.}
\label{convgF2}
\end{figure}
Here we will provide the various numerical details and data regarding the numerical calculations whose results were presented in the main text. \par 
First, we will show convergence of numerical calculation of $S_{\textrm{topo}}$. We choose the simplest case, the single-line TNR of toric code. We first repeat the algorithm described in section \ref{alg} in simple words here for convenience. In the first step, the transfer matrix is calculated using the tensor given (fixed point or varied). Then we choose a specific boundary double tensor as explained in the Appendix \ref{boundary issue}. We apply the transfer matrix on this boundary double tensor and approximate the resulting tensor as an MPS of bond dimensions $D_{\textrm{cut}}=8$. We apply transfer matrix again and approximate the resulting tensor as an MPS of bond dimension 8. We repeat this process and each repetition physically corresponds to increasing the longitudinal length of the our cylindrical subsystem by one unit. Let's say we repeat this process until the length of the half cylinder subsystem is equal to $L$. This process gives us the virtual density matrix $\sigma$, and assuming the mirror symmetry of transfer matrix, the physical reduced density matrix of the half cylinder is $\rho_L \propto \sigma^2$.   With this reduced density matrix we calculate the entanglement entropy $S$ of the half cylinder subsystem for different circumferences $C$. We plot $-S$ vs $C$ and extrapolate it to $C=0$ which gives us the topological entanglement entropy $S_{\textrm{topo}}=S(C=0)$. In principle, one needs to take infinitely large cylinder to achieve the precise value of $S_{\textrm{topo}}$. Practically, we need to keep increasing $L$ until we get a fixed point MPS and keep increasing $C$ until the $S_{topo}$ value converges to a fixed point. 
\par  
Let's first look at the calculation for the single-line toric code fixed point tensor in Eq.~\eqref{SLTNReq}. Half cylinder length is fixed at $L=500$.  $C$ is varied from 50 to 110. Fig.~\ref{convgF1} shows the entanglement entropy $S$ vs the circumference $C$. We get a straight line which is extrapolated to $C=0$. The right figure is a zoomed in version of the left figure to see clearly where the extrapolated line crosses the y-axis. We get  $S_{\textrm{topo}}=S(C=0)\approx \log(2)$ as expected. Fig.~\ref{convgF2} shows the dependence of $S_{\textrm{topo}}$ on the half cylinder length $L$. We see that there is no dependence, that is, fixed point MPS is achieved immediately. It is expected as it is an RG fixed point tensor network state.
\par 
Now we look at the calculation  for single-line toric code fixed point tensor\textit{ varied with an MPO symmetry breaking tensor}. Remember that it is claimed in the main text that this is a trivial state. The variation strength is fixed at $\epsilon=0.01$. Half cylinder length is fixed at $L=500$.  $C$ is varied from 50 to 110. Fig.~\ref{convgU1} shows entanglement entropy $S$ vs the circumference $C$. We get a straight line which is extrapolated to $C=0$. The right figure is a zoomed in version of the left figure to show clearly where the extrapolated line crosses the y-axis. We see $S_{\textrm{topo}} \approx 0$. To see the effect of cylinder length we calculate $S_{\textrm{topo}}$ again but with different cylinder lengths. The results are shown in Fig.~\ref{convgU2}. We see that $S_{\textrm{topo}}$ is $\log(2)$ for small cylinders but converges to  zero as the length is increased. Comparing it to Fig.~\ref{convgF2} we see that, unlike the fixed point case, we need to consider large enough cylinder ($L>600$ in this case) to calculate the correct $S_{\textrm{topo}}$ value for the non-fixed point tensor network state. \par 
Finally we show the effect of variation strength, $\epsilon$, on the convergence. In above calculation we fixed $\epsilon=0.01$. Now we vary $\epsilon$ from 0.01 to 0.02 (making sure it is well below any critical points) and calculate corresponding convergence plots similar to Fig.~\ref{convgU2}. The results are shown in Fig.~\ref{convgU3}. We see that the strength of the variation has a huge effect on convergence. Bigger variations lead to faster convergence. \par 
Though we have presented details of calculation only for one case (single-line toric code TNR), it should be understood that similar patterns are followed in all other cases. For completeness, we present the numerical data plotted in the main text and the relevant parameters used in each case. 
\begin{figure}
\includegraphics[width=\columnwidth]{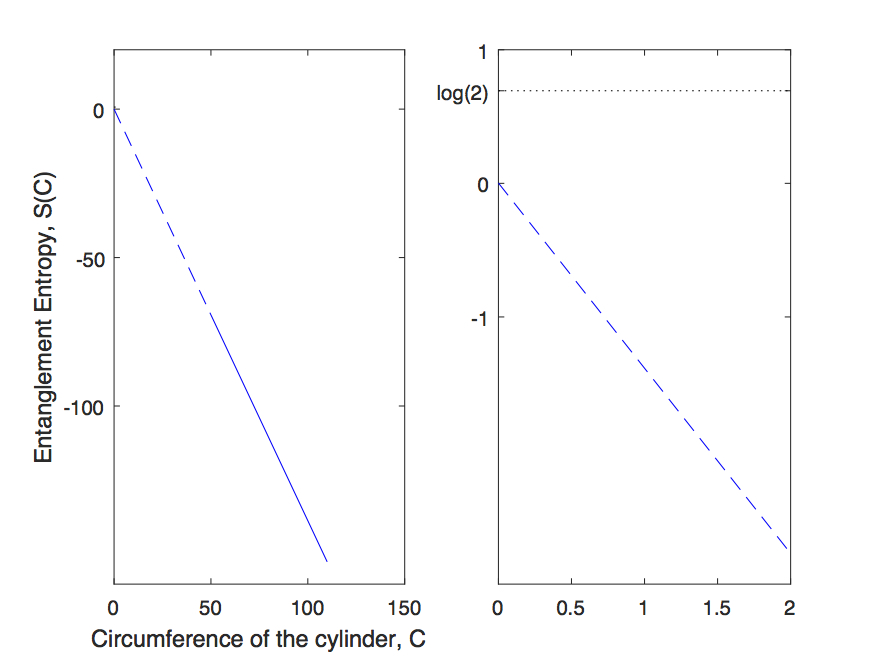}
\caption{Calculation of $S_{\textrm{topo}}$ for a state represented by single-line toric code fixed point tensor varied with an MPO violating tensor. The strength of the variation is fixed at $\epsilon=0.01$. We fix half cylinder length as $L=500$. Circumference is varied from 50 to 110. $S$ varies linearly with $C$. This line is extrapolated back to $C=0$. Its intersection with the y-axis gives $S_{\textrm{topo}}$. Right figure is a zoomed in version of the left figure to show the intersection point clearly. We find $S_{\textrm{topo}} \approx 0$, that is, it is a trivial state.}
\label{convgU1}
\end{figure}
\begin{figure}
\includegraphics[width=\columnwidth]{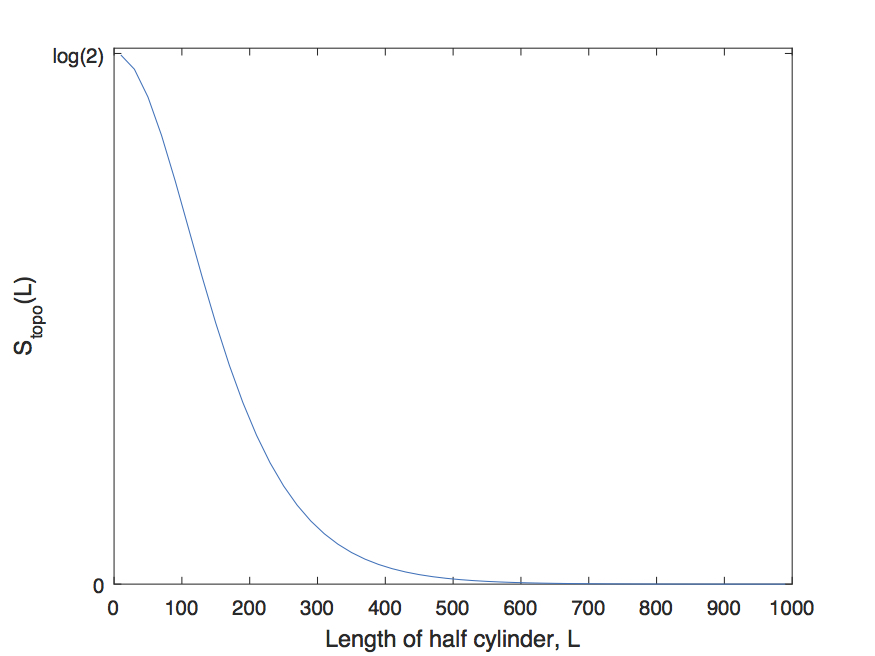}
\caption{ $S_{\textrm{topo}}$ was calculated for a fixed half cylinder length, $L=500$, in Fig.~\ref{convgU1}. We now vary $L$ from 10 to 1000. We see that $S_{\textrm{topo}}$ is close to $\log (2)$ for small cylinders but converges to zero cylinder length $L$ is increased from 1 to 1000. So it is indeed a topologically trivial state. }
\label{convgU2}
\end{figure}
\begin{figure}
\includegraphics[width=\columnwidth]{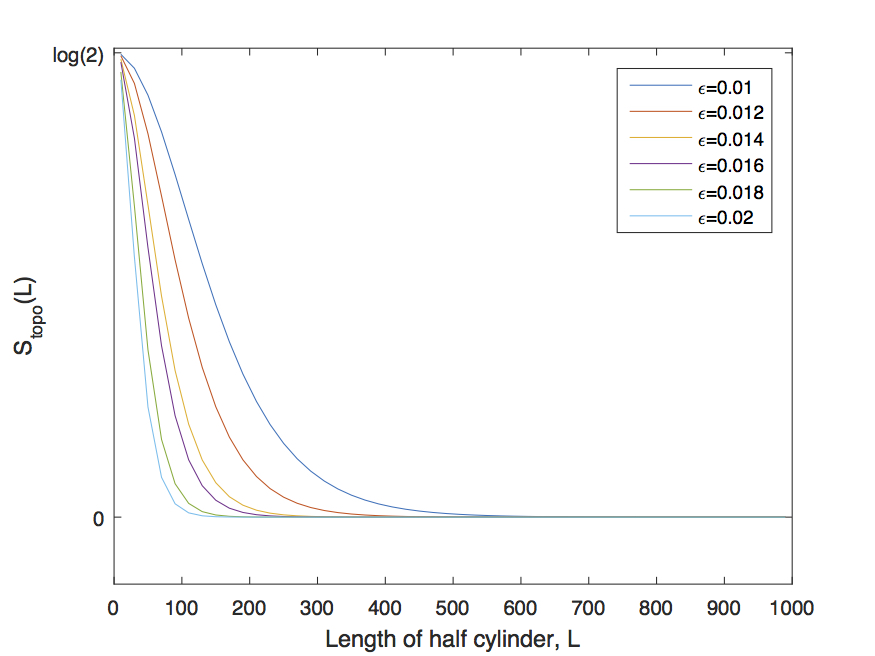}
\caption{ The variation strength $\epsilon$ affects convergence. Higher the variation strength (as long as it is below any critical points) faster is the convergence with the length of the size of the system. }
\label{convgU3}
\end{figure}
\subsection{Single-line TNR toric code}
\label{NDSLTC}
The bond dimension of the MPS is kept fixed at  $D_{\textrm{cut}}=8$ at each step of the iteration. The starting MPS is as explained in the Appendix \ref{boundary issue}. The strength of the variations is fixed at $\epsilon=0.01$. Half cylinder length is  either the length at which convergence of $S_{\textrm{topo}}$ is reached (convergence is reached when $S_{\textrm{topo}}$ value in two successive steps differ by less than $10^{-7}$) or $L=1000$, whichever is smaller.  The circumference is varied from 50 to 110.  

 Following table contains the exact values of the $S_{\textrm{topo}}$ plotted in Fig.~\ref{fig:SLTCrand}.
\begin{center}
\begin{tabular}{ |m{3cm}|m{5cm} | } 
 \hline
No Variation  &  0.6931 \\ 
 \hline 
$Z^{\otimes 3}$ respecting variations  & 0.6931    0.6931    0.6931    0.6931    0.6931    0.6931    0.6931    0.6931    0.6931
\\ 
 \hline 
$Z^{\otimes 3}$ violating variations  &  $10^{-12}\times $

    0.9095         0   -0.4547   -0.4547         0         0    0.9095    0.4547   -0.4547  \\ 
 \hline
\end{tabular}
\end{center}
\subsection{Double-line TNR toric code }
\label{NDDLTC}
The bond dimension of the MPS is kept fixed at  $D_{\textrm{cut}}=16$ at each step of the iteration. The starting MPS is as explained in the Appendix \ref{boundary issue}. The strength of the variations is fixed at $\epsilon=0.01$. Half cylinder length is  either the length at which convergence of $S_{\textrm{topo}}$ is reached (convergence is reached when $S_{\textrm{topo}}$ value in two successive steps differ by less than $10^{-7}$) or $L=1000$, whichever is smaller.  The circumference is varied from 50 to 110.  Following table contains the exact values of the $S_{\textrm{topo}}$ plotted in Fig.~\ref{DLTCrand}. 
\begin{center}
\begin{tabular}{ |m{3cm}|m{5cm} | } 
 \hline
No Variation  &  0.6931 \\ 
 \hline 
$Z\otimes Z$ breaking variations  &  0.6931    0.6931    0.6931    0.6931    0.6931
\\ 
 \hline 
$Z\otimes Z$ respecting, $X^{\otimes 6}$ breaking variations   &      0.0000    0.0015    0.0000    0.0000    0.0002 \\ 
 \hline
 $Z\otimes Z$ and $X^{\otimes 6}$ respecting variations & 0.6931    0.6931    0.6931    0.6931    0.6931 \\ \hline 
\end{tabular}
\end{center}
\subsection{Double-line TNR double semion code}
\label{NDDLDS}
The bond dimension of the MPS is kept fixed at  $D_{\textrm{cut}}=16$ at each step of the iteration. The starting MPS is as explained in the Appendix \ref{boundary issue}. The strength of the variations is fixed at $\epsilon=0.01$. Half cylinder length is  either the length at which convergence of $S_{\textrm{topo}}$ is reached (convergence is reached when $S_{\textrm{topo}}$ value in two successive steps differ by less than $10^{-7}$) or $L=1000$, whichever is smaller.  The circumference is varied from 50 to 110.   Following table contains the exact values of the $S_{\textrm{topo}}$ plotted in Fig.~\ref{DLDSrand}.
\begin{center}
\begin{tabular}{ |m{3cm}|m{5cm} | } 
 \hline
No Variation  &  0.6931 \\ 
 \hline 
$Z\otimes Z$ breaking variations  &  0.6931    0.6931    0.6931    0.6931    0.6931
\\ 
 \hline 
$Z\otimes Z$ respecting, $X^{\otimes 6}$ breaking variations   &  

0.0133    0.0047    0.0191    0.0086    0.0063
 \\ 
 \hline
 $Z\otimes Z$ and $X^{\otimes 6} $ respecting variations & 0.6931    0.6931    0.6931    0.6931    0.6931 \\ \hline 
\end{tabular}
\end{center}
\subsection{Triple-line toric code}
\label{NDTLTC}
The bond dimension of the MPS is kept fixed at  $D_{\textrm{cut}}=16$ at each step of the iteration. The starting MPS is as explained in the Appendix \ref{boundary issue}. The strength of the variations is fixed at $\epsilon=0.2$. Half cylinder length is  either the length at which convergence of $S_{\textrm{topo}}$ is reached (convergence is reached when $S_{\textrm{topo}}$ value in two successive steps differ by less than $10^{-7}$) or $L=1000$, whichever is smaller.  The circumference is varied from 50 to 110.  Following table contains the exact values of the $S_{\textrm{topo}}$ plotted in Fig.~\ref{TLTCrand}.

\begin{center}
\begin{tabular}{ |m{3cm}|m{5cm} | } 
 \hline
No Variation  &  0.6931  \\ 
 \hline 
Variations in $I_V-M_0$  &  0.6931    0.6931    0.6931    0.6931    0.6931

\\ 
 \hline 
Variations in $M_0-\mathbb{M}$  & $10^{-3}\times $   

  0.2467    0.0986    0.2658    0.0257    0.0005
\\ 
 \hline
 Variations in $\mathbb{M}$ &0.6931    0.6931    0.6931    0.6931    0.6931 \\ \hline 
\end{tabular}
\end{center}

\subsection{Triple-line double-semion} 
\label{NDTLDS}
The bond dimension of the MPS is kept fixed at  $D_{\textrm{cut}}=16$ at each step of the iteration. The starting MPS is as explained in the Appendix \ref{boundary issue}. The strength of the variations is fixed at $\epsilon=0.2$. Half cylinder length is  either the length at which convergence of $S_{\textrm{topo}}$ is reached (convergence is reached when $S_{\textrm{topo}}$ value in two successive steps differ by less than $10^{-7}$) or $L=1000$, whichever is smaller.  The circumference is varied from 50 to 110.    Following table contains the exact values of the $S_{\textrm{topo}}$ plotted in Fig.~\ref{TLDSrand}.
\begin{center}
\begin{tabular}{ |m{3cm}|m{5cm} | } 
 \hline
No Variation  &  0.6931  \\ 
 \hline 
Variations in $I_V-M_0$  &  0.6932    0.6931    0.6932    0.6931    0.6932

\\ 
 \hline 
Variations in $M_0-\mathbb{M}$  & $10^{-7}\times $   

 0.7877    0.0849    0.0003    0.0006    0.0000
\\ 
 \hline
 Variations in $\mathbb{M}$ &0.6931    0.6931    0.6931    0.6931    0.6931 \\ \hline 
\end{tabular}
\end{center}
\subsection{Triple-line Fibonacci model}
\label{NDTLF}
The bond dimension of the MPS is kept fixed at  $D_{\textrm{cut}}=16$ at each step of the iteration. The starting MPS is as explained in the Appendix \ref{boundary issue}. The strength of the variations is fixed at $\epsilon=0.1$. Half cylinder length is  either the length at which convergence of $S_{\textrm{topo}}$ is reached (convergence is reached when $S_{\textrm{topo}}$ value in two successive steps differ by less than $10^{-7}$) or $L=2000$, whichever is smaller.  The circumference is varied from 50 to 110. Following table contains the exact values of the $S_{\textrm{topo}}$ plotted in Fig.~\ref{TLFIBrand}.
\begin{center}
\begin{tabular}{ |m{3cm}|m{5cm} | } 
 \hline
No Variation  &  1.1787 \\ 
 \hline 
Variations in $I_V-M_0$  &  1.1779    1.1776    1.1774    1.1778    1.1779
\\ 
 \hline 
Variations in $M_0-\mathbb{M}$  & $10^{-7}\times $   

   -0.2330    0.2841    0.0517    0.0335    0.0299
\\ 
 \hline
 Variations in $\mathbb{M}$ & 1.1535    1.1623    1.1556    1.1386    1.1667 \\ \hline 
\end{tabular}
\end{center}
\pagebreak
\par 

\bibliography{TN_stability.bib}

\begin{thebibliography}{27}%
\makeatletter
\providecommand \@ifxundefined [1]{%
 \@ifx{#1\undefined}
}%
\providecommand \@ifnum [1]{%
 \ifnum #1\expandafter \@firstoftwo
 \else \expandafter \@secondoftwo
 \fi
}%
\providecommand \@ifx [1]{%
 \ifx #1\expandafter \@firstoftwo
 \else \expandafter \@secondoftwo
 \fi
}%
\providecommand \natexlab [1]{#1}%
\providecommand \enquote  [1]{``#1''}%
\providecommand \bibnamefont  [1]{#1}%
\providecommand \bibfnamefont [1]{#1}%
\providecommand \citenamefont [1]{#1}%
\providecommand \href@noop [0]{\@secondoftwo}%
\providecommand \href [0]{\begingroup \@sanitize@url \@href}%
\providecommand \@href[1]{\@@startlink{#1}\@@href}%
\providecommand \@@href[1]{\endgroup#1\@@endlink}%
\providecommand \@sanitize@url [0]{\catcode `\\12\catcode `\$12\catcode
  `\&12\catcode `\#12\catcode `\^12\catcode `\_12\catcode `\%12\relax}%
\providecommand \@@startlink[1]{}%
\providecommand \@@endlink[0]{}%
\providecommand \url  [0]{\begingroup\@sanitize@url \@url }%
\providecommand \@url [1]{\endgroup\@href {#1}{\urlprefix }}%
\providecommand \urlprefix  [0]{URL }%
\providecommand \Eprint [0]{\href }%
\providecommand \doibase [0]{http://dx.doi.org/}%
\providecommand \selectlanguage [0]{\@gobble}%
\providecommand \bibinfo  [0]{\@secondoftwo}%
\providecommand \bibfield  [0]{\@secondoftwo}%
\providecommand \translation [1]{[#1]}%
\providecommand \BibitemOpen [0]{}%
\providecommand \bibitemStop [0]{}%
\providecommand \bibitemNoStop [0]{.\EOS\space}%
\providecommand \EOS [0]{\spacefactor3000\relax}%
\providecommand \BibitemShut  [1]{\csname bibitem#1\endcsname}%
\let\auto@bib@innerbib\@empty
\bibitem [{\citenamefont {Chen}\ \emph {et~al.}(2010)\citenamefont {Chen},
  \citenamefont {Zeng}, \citenamefont {Gu}, \citenamefont {Chuang},\ and\
  \citenamefont {Wen}}]{Chen10}%
  \BibitemOpen
  \bibfield  {author} {\bibinfo {author} {\bibfnamefont {X.}~\bibnamefont
  {Chen}}, \bibinfo {author} {\bibfnamefont {B.}~\bibnamefont {Zeng}}, \bibinfo
  {author} {\bibfnamefont {Z.-C.}\ \bibnamefont {Gu}}, \bibinfo {author}
  {\bibfnamefont {I.~L.}\ \bibnamefont {Chuang}}, \ and\ \bibinfo {author}
  {\bibfnamefont {X.-G.}\ \bibnamefont {Wen}},\ }\href {\doibase
  10.1103/PhysRevB.82.165119} {\bibfield  {journal} {\bibinfo  {journal} {Phys.
  Rev. B}\ }\textbf {\bibinfo {volume} {82}},\ \bibinfo {pages} {165119}
  (\bibinfo {year} {2010})}\BibitemShut {NoStop}%
\bibitem [{\citenamefont {{\c{S}}ahino{\u{g}}lu}\ \emph
  {et~al.}(2014)\citenamefont {{\c{S}}ahino{\u{g}}lu}, \citenamefont
  {Williamson}, \citenamefont {Bultinck}, \citenamefont {Mari{\"e}n},
  \citenamefont {Haegeman}, \citenamefont {Schuch},\ and\ \citenamefont
  {Verstraete}}]{csahinouglu2014characterizing}%
  \BibitemOpen
  \bibfield  {author} {\bibinfo {author} {\bibfnamefont {M.~B.}\ \bibnamefont
  {{\c{S}}ahino{\u{g}}lu}}, \bibinfo {author} {\bibfnamefont {D.}~\bibnamefont
  {Williamson}}, \bibinfo {author} {\bibfnamefont {N.}~\bibnamefont
  {Bultinck}}, \bibinfo {author} {\bibfnamefont {M.}~\bibnamefont
  {Mari{\"e}n}}, \bibinfo {author} {\bibfnamefont {J.}~\bibnamefont
  {Haegeman}}, \bibinfo {author} {\bibfnamefont {N.}~\bibnamefont {Schuch}}, \
  and\ \bibinfo {author} {\bibfnamefont {F.}~\bibnamefont {Verstraete}},\
  }\href@noop {} {\bibfield  {journal} {\bibinfo  {journal} {arXiv preprint
  arXiv:1409.2150}\ } (\bibinfo {year} {2014})}\BibitemShut {NoStop}%
\bibitem [{\citenamefont {Fannes}\ \emph {et~al.}(1992)\citenamefont {Fannes},
  \citenamefont {Nachtergaele},\ and\ \citenamefont {Werner}}]{Fannes92}%
  \BibitemOpen
  \bibfield  {author} {\bibinfo {author} {\bibfnamefont {M.}~\bibnamefont
  {Fannes}}, \bibinfo {author} {\bibfnamefont {B.}~\bibnamefont
  {Nachtergaele}}, \ and\ \bibinfo {author} {\bibfnamefont {R.}~\bibnamefont
  {Werner}},\ }\href {http://dx.doi.org/10.1007/BF02099178} {\bibfield
  {journal} {\bibinfo  {journal} {Communications in Mathematical Physics}\
  }\textbf {\bibinfo {volume} {144}},\ \bibinfo {pages} {443} (\bibinfo {year}
  {1992})},\ \bibinfo {note} {10.1007/BF02099178}\BibitemShut {NoStop}%
\bibitem [{\citenamefont {White}(1993)}]{White93}%
  \BibitemOpen
  \bibfield  {author} {\bibinfo {author} {\bibfnamefont {S.~R.}\ \bibnamefont
  {White}},\ }\href {\doibase 10.1103/PhysRevB.48.10345} {\bibfield  {journal}
  {\bibinfo  {journal} {Phys. Rev. B}\ }\textbf {\bibinfo {volume} {48}},\
  \bibinfo {pages} {10345} (\bibinfo {year} {1993})}\BibitemShut {NoStop}%
\bibitem [{\citenamefont {Verstraete}\ \emph {et~al.}(2008)\citenamefont
  {Verstraete}, \citenamefont {Murg},\ and\ \citenamefont
  {Cirac}}]{Verstraete08}%
  \BibitemOpen
  \bibfield  {author} {\bibinfo {author} {\bibfnamefont {F.}~\bibnamefont
  {Verstraete}}, \bibinfo {author} {\bibfnamefont {V.}~\bibnamefont {Murg}}, \
  and\ \bibinfo {author} {\bibfnamefont {J.}~\bibnamefont {Cirac}},\ }\href
  {\doibase 10.1080/14789940801912366} {\bibfield  {journal} {\bibinfo
  {journal} {Advances in Physics}\ }\textbf {\bibinfo {volume} {57}},\ \bibinfo
  {pages} {143} (\bibinfo {year} {2008})}\BibitemShut {NoStop}%
\bibitem [{\citenamefont {Vidal}(2009)}]{Vidal09}%
  \BibitemOpen
  \bibfield  {author} {\bibinfo {author} {\bibfnamefont {G.}~\bibnamefont
  {Vidal}},\ }\href@noop {} {\bibfield  {journal} {\bibinfo  {journal} {ArXiv
  e-prints 0912.1651}\ } (\bibinfo {year} {2009})},\ \Eprint
  {http://arxiv.org/abs/0912.1651} {arXiv:0912.1651 [cond-mat.str-el]}
  \BibitemShut {NoStop}%
\bibitem [{\citenamefont {{Levin}}\ and\ \citenamefont
  {{Wen}}(2005)}]{Levin05}%
  \BibitemOpen
  \bibfield  {author} {\bibinfo {author} {\bibfnamefont {M.~A.}\ \bibnamefont
  {{Levin}}}\ and\ \bibinfo {author} {\bibfnamefont {X.-G.}\ \bibnamefont
  {{Wen}}},\ }\href {\doibase 10.1103/PhysRevB.71.045110} {\bibfield  {journal}
  {\bibinfo  {journal} {\prb}\ }\textbf {\bibinfo {volume} {71}},\ \bibinfo
  {eid} {045110} (\bibinfo {year} {2005})},\ \Eprint
  {http://arxiv.org/abs/cond-mat/0404617} {cond-mat/0404617} \BibitemShut
  {NoStop}%
\bibitem [{\citenamefont {{Gu}}\ \emph {et~al.}(2009)\citenamefont {{Gu}},
  \citenamefont {{Levin}}, \citenamefont {{Swingle}},\ and\ \citenamefont
  {{Wen}}}]{Gu09}%
  \BibitemOpen
  \bibfield  {author} {\bibinfo {author} {\bibfnamefont {Z.-C.}\ \bibnamefont
  {{Gu}}}, \bibinfo {author} {\bibfnamefont {M.}~\bibnamefont {{Levin}}},
  \bibinfo {author} {\bibfnamefont {B.}~\bibnamefont {{Swingle}}}, \ and\
  \bibinfo {author} {\bibfnamefont {X.-G.}\ \bibnamefont {{Wen}}},\ }\href
  {\doibase 10.1103/PhysRevB.79.085118} {\bibfield  {journal} {\bibinfo
  {journal} {\prb}\ }\textbf {\bibinfo {volume} {79}},\ \bibinfo {eid} {085118}
  (\bibinfo {year} {2009})},\ \Eprint {http://arxiv.org/abs/0809.2821}
  {arXiv:0809.2821 [cond-mat.str-el]} \BibitemShut {NoStop}%
\bibitem [{\citenamefont {Buerschaper}\ \emph {et~al.}(2009)\citenamefont
  {Buerschaper}, \citenamefont {Aguado},\ and\ \citenamefont
  {Vidal}}]{Buerschaper09}%
  \BibitemOpen
  \bibfield  {author} {\bibinfo {author} {\bibfnamefont {O.}~\bibnamefont
  {Buerschaper}}, \bibinfo {author} {\bibfnamefont {M.}~\bibnamefont {Aguado}},
  \ and\ \bibinfo {author} {\bibfnamefont {G.}~\bibnamefont {Vidal}},\ }\href
  {\doibase 10.1103/PhysRevB.79.085119} {\bibfield  {journal} {\bibinfo
  {journal} {Phys. Rev. B}\ }\textbf {\bibinfo {volume} {79}},\ \bibinfo
  {pages} {085119} (\bibinfo {year} {2009})}\BibitemShut {NoStop}%
\bibitem [{\citenamefont {Yan}\ \emph {et~al.}(2011)\citenamefont {Yan},
  \citenamefont {Huse},\ and\ \citenamefont {White}}]{Yan11}%
  \BibitemOpen
  \bibfield  {author} {\bibinfo {author} {\bibfnamefont {S.}~\bibnamefont
  {Yan}}, \bibinfo {author} {\bibfnamefont {D.~A.}\ \bibnamefont {Huse}}, \
  and\ \bibinfo {author} {\bibfnamefont {S.~R.}\ \bibnamefont {White}},\ }\href
  {\doibase 10.1126/science.1201080} {\bibfield  {journal} {\bibinfo  {journal}
  {Science}\ }\textbf {\bibinfo {volume} {332}},\ \bibinfo {pages} {1173}
  (\bibinfo {year} {2011})}\BibitemShut {NoStop}%
\bibitem [{\citenamefont {Jiang}\ \emph {et~al.}(2012)\citenamefont {Jiang},
  \citenamefont {Wang},\ and\ \citenamefont {Balents}}]{Jiang12}%
  \BibitemOpen
  \bibfield  {author} {\bibinfo {author} {\bibfnamefont {H.-C.}\ \bibnamefont
  {Jiang}}, \bibinfo {author} {\bibfnamefont {Z.}~\bibnamefont {Wang}}, \ and\
  \bibinfo {author} {\bibfnamefont {L.}~\bibnamefont {Balents}},\ }\href
  {http://dx.doi.org/10.1038/nphys2465} {\bibfield  {journal} {\bibinfo
  {journal} {Nat Phys}\ }\textbf {\bibinfo {volume} {8}},\ \bibinfo {pages}
  {902} (\bibinfo {year} {2012})}\BibitemShut {NoStop}%
\bibitem [{\citenamefont {Depenbrock}\ \emph {et~al.}(2012)\citenamefont
  {Depenbrock}, \citenamefont {McCulloch},\ and\ \citenamefont
  {Schollw\"ock}}]{Depenbrock12}%
  \BibitemOpen
  \bibfield  {author} {\bibinfo {author} {\bibfnamefont {S.}~\bibnamefont
  {Depenbrock}}, \bibinfo {author} {\bibfnamefont {I.~P.}\ \bibnamefont
  {McCulloch}}, \ and\ \bibinfo {author} {\bibfnamefont {U.}~\bibnamefont
  {Schollw\"ock}},\ }\href {\doibase 10.1103/PhysRevLett.109.067201} {\bibfield
   {journal} {\bibinfo  {journal} {Phys. Rev. Lett.}\ }\textbf {\bibinfo
  {volume} {109}},\ \bibinfo {pages} {067201} (\bibinfo {year}
  {2012})}\BibitemShut {NoStop}%
\bibitem [{\citenamefont {Bravyi}\ \emph {et~al.}(2010)\citenamefont {Bravyi},
  \citenamefont {Hastings},\ and\ \citenamefont {Michalakis}}]{Bravyi10}%
  \BibitemOpen
  \bibfield  {author} {\bibinfo {author} {\bibfnamefont {S.}~\bibnamefont
  {Bravyi}}, \bibinfo {author} {\bibfnamefont {M.~B.}\ \bibnamefont
  {Hastings}}, \ and\ \bibinfo {author} {\bibfnamefont {S.}~\bibnamefont
  {Michalakis}},\ }\href {\doibase 10.1063/1.3490195} {\bibfield  {journal}
  {\bibinfo  {journal} {Journal of Mathematical Physics}\ }\textbf {\bibinfo
  {volume} {51}},\ \bibinfo {eid} {093512} (\bibinfo {year}
  {2010})}\BibitemShut {NoStop}%
\bibitem [{\citenamefont {{Buerschaper}}(2014)}]{Buerschaper2014}%
  \BibitemOpen
  \bibfield  {author} {\bibinfo {author} {\bibfnamefont {O.}~\bibnamefont
  {{Buerschaper}}},\ }\href {\doibase 10.1016/j.aop.2014.09.007} {\bibfield
  {journal} {\bibinfo  {journal} {Annals of Physics}\ }\textbf {\bibinfo
  {volume} {351}},\ \bibinfo {pages} {447} (\bibinfo {year} {2014})},\ \Eprint
  {http://arxiv.org/abs/1307.7763} {arXiv:1307.7763 [cond-mat.str-el]}
  \BibitemShut {NoStop}%
\bibitem [{\citenamefont {Kitaev}\ and\ \citenamefont
  {Preskill}(2006)}]{Kitaev06}%
  \BibitemOpen
  \bibfield  {author} {\bibinfo {author} {\bibfnamefont {A.}~\bibnamefont
  {Kitaev}}\ and\ \bibinfo {author} {\bibfnamefont {J.}~\bibnamefont
  {Preskill}},\ }\href {\doibase 10.1103/PhysRevLett.96.110404} {\bibfield
  {journal} {\bibinfo  {journal} {Phys. Rev. Lett.}\ }\textbf {\bibinfo
  {volume} {96}},\ \bibinfo {pages} {110404} (\bibinfo {year}
  {2006})}\BibitemShut {NoStop}%
\bibitem [{\citenamefont {Levin}\ and\ \citenamefont {Wen}(2006)}]{Levin06}%
  \BibitemOpen
  \bibfield  {author} {\bibinfo {author} {\bibfnamefont {M.}~\bibnamefont
  {Levin}}\ and\ \bibinfo {author} {\bibfnamefont {X.-G.}\ \bibnamefont
  {Wen}},\ }\href {\doibase 10.1103/PhysRevLett.96.110405} {\bibfield
  {journal} {\bibinfo  {journal} {Phys. Rev. Lett.}\ }\textbf {\bibinfo
  {volume} {96}},\ \bibinfo {pages} {110405} (\bibinfo {year}
  {2006})}\BibitemShut {NoStop}%
\bibitem [{\citenamefont {{Gu}}\ \emph {et~al.}(2008)\citenamefont {{Gu}},
  \citenamefont {{Levin}},\ and\ \citenamefont {{Wen}}}]{Gu08}%
  \BibitemOpen
  \bibfield  {author} {\bibinfo {author} {\bibfnamefont {Z.-C.}\ \bibnamefont
  {{Gu}}}, \bibinfo {author} {\bibfnamefont {M.}~\bibnamefont {{Levin}}}, \
  and\ \bibinfo {author} {\bibfnamefont {X.-G.}\ \bibnamefont {{Wen}}},\ }\href
  {\doibase 10.1103/PhysRevB.78.205116} {\bibfield  {journal} {\bibinfo
  {journal} {\prb}\ }\textbf {\bibinfo {volume} {78}},\ \bibinfo {eid} {205116}
  (\bibinfo {year} {2008})},\ \Eprint {http://arxiv.org/abs/0806.3509}
  {arXiv:0806.3509 [cond-mat.str-el]} \BibitemShut {NoStop}%
\bibitem [{\citenamefont {Kitaev}(2003)}]{Kitaev03}%
  \BibitemOpen
  \bibfield  {author} {\bibinfo {author} {\bibfnamefont {A.}~\bibnamefont
  {Kitaev}},\ }\href {\doibase http://dx.doi.org/10.1016/S0003-4916(02)00018-0}
  {\bibfield  {journal} {\bibinfo  {journal} {Annals of Physics}\ }\textbf
  {\bibinfo {volume} {303}},\ \bibinfo {pages} {2 } (\bibinfo {year}
  {2003})}\BibitemShut {NoStop}%
\bibitem [{\citenamefont {{Schuch}}\ \emph {et~al.}(2010)\citenamefont
  {{Schuch}}, \citenamefont {{Cirac}},\ and\ \citenamefont
  {{P{\'e}rez-Garc{\'{\i}}a}}}]{Schuch2010}%
  \BibitemOpen
  \bibfield  {author} {\bibinfo {author} {\bibfnamefont {N.}~\bibnamefont
  {{Schuch}}}, \bibinfo {author} {\bibfnamefont {I.}~\bibnamefont {{Cirac}}}, \
  and\ \bibinfo {author} {\bibfnamefont {D.}~\bibnamefont
  {{P{\'e}rez-Garc{\'{\i}}a}}},\ }\href {\doibase 10.1016/j.aop.2010.05.008}
  {\bibfield  {journal} {\bibinfo  {journal} {Annals of Physics}\ }\textbf
  {\bibinfo {volume} {325}},\ \bibinfo {pages} {2153} (\bibinfo {year}
  {2010})},\ \Eprint {http://arxiv.org/abs/1001.3807} {arXiv:1001.3807
  [quant-ph]} \BibitemShut {NoStop}%
\bibitem [{\citenamefont {Cirac}\ \emph {et~al.}(2011)\citenamefont {Cirac},
  \citenamefont {Poilblanc}, \citenamefont {Schuch},\ and\ \citenamefont
  {Verstraete}}]{Cirac11}%
  \BibitemOpen
  \bibfield  {author} {\bibinfo {author} {\bibfnamefont {J.~I.}\ \bibnamefont
  {Cirac}}, \bibinfo {author} {\bibfnamefont {D.}~\bibnamefont {Poilblanc}},
  \bibinfo {author} {\bibfnamefont {N.}~\bibnamefont {Schuch}}, \ and\ \bibinfo
  {author} {\bibfnamefont {F.}~\bibnamefont {Verstraete}},\ }\href {\doibase
  10.1103/PhysRevB.83.245134} {\bibfield  {journal} {\bibinfo  {journal} {Phys.
  Rev. B}\ }\textbf {\bibinfo {volume} {83}},\ \bibinfo {pages} {245134}
  (\bibinfo {year} {2011})}\BibitemShut {NoStop}%
\bibitem [{\citenamefont {Flammia}\ \emph {et~al.}(2009)\citenamefont
  {Flammia}, \citenamefont {Hamma}, \citenamefont {Hughes},\ and\ \citenamefont
  {Wen}}]{Flammia09}%
  \BibitemOpen
  \bibfield  {author} {\bibinfo {author} {\bibfnamefont {S.~T.}\ \bibnamefont
  {Flammia}}, \bibinfo {author} {\bibfnamefont {A.}~\bibnamefont {Hamma}},
  \bibinfo {author} {\bibfnamefont {T.~L.}\ \bibnamefont {Hughes}}, \ and\
  \bibinfo {author} {\bibfnamefont {X.-G.}\ \bibnamefont {Wen}},\ }\href
  {\doibase 10.1103/PhysRevLett.103.261601} {\bibfield  {journal} {\bibinfo
  {journal} {Phys. Rev. Lett.}\ }\textbf {\bibinfo {volume} {103}},\ \bibinfo
  {pages} {261601} (\bibinfo {year} {2009})}\BibitemShut {NoStop}%
\bibitem [{\citenamefont {{Dong}}\ \emph {et~al.}(2008)\citenamefont {{Dong}},
  \citenamefont {{Fradkin}}, \citenamefont {{Leigh}},\ and\ \citenamefont
  {{Nowling}}}]{DongFradkinLeighNowling}%
  \BibitemOpen
  \bibfield  {author} {\bibinfo {author} {\bibfnamefont {S.}~\bibnamefont
  {{Dong}}}, \bibinfo {author} {\bibfnamefont {E.}~\bibnamefont {{Fradkin}}},
  \bibinfo {author} {\bibfnamefont {R.~G.}\ \bibnamefont {{Leigh}}}, \ and\
  \bibinfo {author} {\bibfnamefont {S.}~\bibnamefont {{Nowling}}},\ }\href
  {\doibase 10.1088/1126-6708/2008/05/016} {\bibfield  {journal} {\bibinfo
  {journal} {Journal of High Energy Physics}\ }\textbf {\bibinfo {volume}
  {5}},\ \bibinfo {eid} {016} (\bibinfo {year} {2008})},\ \Eprint
  {http://arxiv.org/abs/0802.3231} {arXiv:0802.3231 [hep-th]} \BibitemShut
  {NoStop}%
\bibitem [{\citenamefont {Zhang}\ \emph {et~al.}(2012)\citenamefont {Zhang},
  \citenamefont {Grover}, \citenamefont {Turner}, \citenamefont {Oshikawa},\
  and\ \citenamefont {Vishwanath}}]{Zhang12}%
  \BibitemOpen
  \bibfield  {author} {\bibinfo {author} {\bibfnamefont {Y.}~\bibnamefont
  {Zhang}}, \bibinfo {author} {\bibfnamefont {T.}~\bibnamefont {Grover}},
  \bibinfo {author} {\bibfnamefont {A.}~\bibnamefont {Turner}}, \bibinfo
  {author} {\bibfnamefont {M.}~\bibnamefont {Oshikawa}}, \ and\ \bibinfo
  {author} {\bibfnamefont {A.}~\bibnamefont {Vishwanath}},\ }\href {\doibase
  10.1103/PhysRevB.85.235151} {\bibfield  {journal} {\bibinfo  {journal} {Phys.
  Rev.}\ }\textbf {\bibinfo {volume} {B85}},\ \bibinfo {pages} {235151}
  (\bibinfo {year} {2012})},\ \Eprint {http://arxiv.org/abs/1111.2342}
  {arXiv:1111.2342 [cond-mat.str-el]} \BibitemShut {NoStop}%
\bibitem [{\citenamefont {Vidal}(2003)}]{Vidal2003}%
  \BibitemOpen
  \bibfield  {author} {\bibinfo {author} {\bibfnamefont {G.}~\bibnamefont
  {Vidal}},\ }\href {\doibase 10.1103/PhysRevLett.91.147902} {\bibfield
  {journal} {\bibinfo  {journal} {Phys. Rev. Lett.}\ }\textbf {\bibinfo
  {volume} {91}},\ \bibinfo {pages} {147902} (\bibinfo {year}
  {2003})}\BibitemShut {NoStop}%
\bibitem [{Note1()}]{Note1}%
  \BibitemOpen
  \bibinfo {note} {Not to be confused with TNR that we use for referring to
  tensor network representation.}\BibitemShut {Stop}%
\bibitem [{\citenamefont {Evenbly}\ and\ \citenamefont
  {Vidal}(2015)}]{EvenblyVidalTNR}%
  \BibitemOpen
  \bibfield  {author} {\bibinfo {author} {\bibfnamefont {G.}~\bibnamefont
  {Evenbly}}\ and\ \bibinfo {author} {\bibfnamefont {G.}~\bibnamefont
  {Vidal}},\ }\href {\doibase 10.1103/PhysRevLett.115.180405} {\bibfield
  {journal} {\bibinfo  {journal} {Phys. Rev. Lett.}\ }\textbf {\bibinfo
  {volume} {115}},\ \bibinfo {pages} {180405} (\bibinfo {year}
  {2015})}\BibitemShut {NoStop}%
\bibitem [{\citenamefont {Freedman}\ \emph {et~al.}(2004)\citenamefont
  {Freedman}, \citenamefont {Nayak}, \citenamefont {Shtengel}, \citenamefont
  {Walker},\ and\ \citenamefont {Wang}}]{Freedman04}%
  \BibitemOpen
  \bibfield  {author} {\bibinfo {author} {\bibfnamefont {M.}~\bibnamefont
  {Freedman}}, \bibinfo {author} {\bibfnamefont {C.}~\bibnamefont {Nayak}},
  \bibinfo {author} {\bibfnamefont {K.}~\bibnamefont {Shtengel}}, \bibinfo
  {author} {\bibfnamefont {K.}~\bibnamefont {Walker}}, \ and\ \bibinfo {author}
  {\bibfnamefont {Z.}~\bibnamefont {Wang}},\ }\href {\doibase
  http://dx.doi.org/10.1016/j.aop.2004.01.006} {\bibfield  {journal} {\bibinfo
  {journal} {Annals of Physics}\ }\textbf {\bibinfo {volume} {310}},\ \bibinfo
  {pages} {428 } (\bibinfo {year} {2004})}\BibitemShut {NoStop}%
\end{thebibliography}%


%

\end{document}